\documentclass[12pt,halfline,a4paper]{ouparticle}
\usepackage{amssymb}
\usepackage{amsthm, color}
\usepackage{amsmath}
\usepackage{multicol}
\usepackage{algorithm}
\usepackage{algpseudocode} 
\usepackage{amsmath}

\usepackage{multirow}
\usepackage[authoryear,round]{natbib}
\usepackage{subfigure}
\begin{document}

\title{Autoregressive Hypergraph}

\author{%
\name{Xianghe Zhu  \quad and \quad Qiwei Yao}
\address{Department of Statistics, London School of Economics and Political Science, London, WC2A 2AE, 
United Kingdom}
\email{X.Zhu58@lse.ac.uk \quad and \quad q.yao@lse.ac.uk}
}

\abstract{
Traditional graph representations are insufficient for modelling real-world phenomena involving multi-entity interactions, such as collaborative projects or protein complexes, necessitating the use of hypergraphs. While hypergraphs preserve the intrinsic nature of such complex relationships, existing models often overlook temporal evolution in relational data. To address this, we introduce a first-order autoregressive (i.e. AR(1)) model for dynamic non-uniform hypergraphs. This is the first dynamic hypergraph model with provable theoretical guarantees, explicitly defining the temporal evolution of hyperedge presence through transition probabilities that govern persistence and change dynamics. This framework provides closed-form expressions for key probabilistic properties and facilitates straightforward maximum-likelihood inference with uniform error bounds and asymptotic normality, along with a permutation-based diagnostic test. We also consider an AR(1) hypergraph stochastic block model (HSBM), where a novel Laplacian enables exact and efficient latent community recovery via a spectral clustering algorithm. Furthermore, we develop a likelihood-based change-point estimator for the HSBM to detect structural breaks within the time series. The efficacy and practical value of our methods are comprehensively demonstrated through extensive simulation studies and compelling applications to a primary school interaction data set and the Enron email corpus, revealing insightful community structures and significant temporal changes.}


\maketitle

\noindent\textbf{MOS subject classification: } 62M10, 05C65, 62H30

\section{Introduction}

Network analysis has emerged as an effective approach for modelling interactions between entities across diverse domains. Traditional networks, represented as graphs with pairwise connections, have proven valuable for understanding complex systems ranging from social relationships to biological interactions. However, many real-world phenomena inherently involve simultaneous interactions among multiple entities that cannot be adequately captured through pairwise edges alone. For instance, collaborative projects typically involve multiple individuals working together \citep{xie2021distributed}, protein complexes often comprise numerous interacting molecules \citep{klimm_hypergraphs_2021,klamt2009hypergraphs}, and communication patterns frequently encompass group interactions rather than merely bilateral exchanges \citep{benson2022fauci}. These higher-order interactions are also prevalent in global trading networks \citep{yi2022structure} and sports analytics \citep{josephs_hypergraph_2024}, demonstrating the widespread need for more sophisticated representational frameworks beyond traditional graphs.

Hypergraphs offer a powerful generalisation of conventional graph structures, explicitly modelling higher-order interactions through hyperedges that can connect arbitrary numbers of vertices simultaneously. While a popular approach is to project hypergraphs onto standard graphs and then apply established community detection methods, \cite{ke_community_2020} have demonstrated that such projections often cause significant and unwanted information loss. By preserving the intrinsic multi-entity nature of complex relationships, hypergraph representations enable more accurate analysis of dependencies and structural patterns that would otherwise be obscured when reduced to pairwise relationships. This maintenance of multi-way connection integrity facilitates deeper insights into system organisation and behaviour across numerous domains, from social networks and ecological systems to economic interactions and technological infrastructures.

The limitations of traditional graph representations have spurred significant interest in hypergraph models. The simplest such structure is the $k$-uniform hypergraph, wherein every hyperedge connects exactly $k$ nodes, providing a clear analytical framework. Numerous studies have leveraged this structure \citep{lyu_latent_2023,ghoshdastidar_provable_2015, JMLR:v18:16-100, yuan2022statistical}, developing methods for clustering, inference, and structural analysis. However, the $k$-uniform assumption severely constrains applicability to real-world scenarios where interaction sizes naturally vary.

Recent research has consequently shifted towards non-uniform hypergraphs, where hyperedges may differ in cardinality. \cite{lunagomez2017geometric} examined hereditary hypergraphs, which require all subsets of a hyperedge  also form hyperedges—a constraint that proves overly restrictive in practice. \cite{wu_general_2024} studied sparse hypergraphs with multiplicity to accommodate recurring interactions, whilst \cite{yu_modeling_2025} developed models capturing node diversity and heterogeneous popularity. Perhaps most prominently, hypergraph stochastic blockmodels and related clustering methods  have emerged as prevailing approaches, see  \cite{zhen2023community, ghoshdastidar_consistency_2017, ke_community_2020, brusa2024model}. Additionally, \cite{benson_simplicial_2018} have provided valuable collections of hypergraph datasets and developed methods for higher-order link prediction, further facilitating empirical research in this area. These developments adopt a static perspective, neglecting the temporal evolution inherent in relational data.

In parallel, dynamic network analysis has seen substantial development in recent years,  as evidenced in the scholarly works of \cite{jiang_autoregressive_2023, chen2023community, hung2024bayesian, mantziou2023gnar, knight2016modelling,liu2023new, passino2023mutually, khabou2025markov, chang2024autoregressive, jiang2023two}. These contributions have employed diverse methodologies to address temporal phenomena, including change-point detection \citep{padilla2022change,corneck2025online, hallgren2024changepoint}, anomaly detection \citep{chen_multiple_2025}, and other evolving characteristics. 


In this paper, we introduce a first-order autoregressive (i.e. AR(1)) model for dynamic non-uniform hypergraphs, in which the presence of each hyperedge evolves over time with explicit temporal dependence. The setting can be viewed as an extension of AR networks of \cite{jiang_autoregressive_2023} to AR hypergraphs.
This AR(1) hypergraph model admits some simple and nice
probabilistic properties such that we can identify the conditions
for strict stationarity and $\alpha$-mixing with explicit convergence rates.
It also admits straightforward maximum-likelihood inference with uniform error bounds and asymptotic normality. We also provide a permutation-based diagnostic test to assess model fit. Building on this AR(1) framework, we further propose a dynamic hypergraph stochastic block model (HSBM) in which node memberships drive time-varying hyperedge transition probabilities. By defining a novel Laplacian, we show that latent communities can be recovered exactly and efficiently via a tailored spectral-clustering algorithm. Note that the Laplacian adopted in \cite{jiang_autoregressive_2023} for AR networks is not applicable for hypergraphs.
Finally, to handle some non-stationary hypergraphs, we develop a likelihood-based change-point estimator that partitions the time series and detects structural breaks under the HSBM.

Our contributions advance both the theoretical understanding and practical application of dynamic hypergraph models in several significant ways:

\begin{itemize}
    \item We propose the first autoregressive framework for non-uniform hypergraphs with temporal dependence, establishing fundamental properties including stationarity conditions, mixing rates, and asymptotic distributions of parameter estimators. 

    \item Our work introduces a novel spectral clustering algorithm based on transition probabilities rather than static connections, with provable guarantees for perfect community recovery under specified conditions. Additionally, we develop a principled approach to change-point detection specifically tailored to the unique challenges of dynamic hypergraph data, with established consistency properties and error rates.

    \item The methods developed herein enable new analytical capabilities for real-world applications where multi-entity interactions evolve over time. Our empirical results demonstrate the practical value of these techniques through improved community detection in social interaction networks and meaningful change-point identification in communication patterns, offering insights that would be inaccessible through either static hypergraph analysis or conventional dynamic network methods alone.
\end{itemize}

The remainder of this paper is organised as follows. Section 2 introduces the AR(1) hypergraph model and  its probabilistic properties. We also develop maximum likelihood estimation procedures with theoretical guarantees. A permutation test is also introduced for model diagnostic checking. Section 3 extends our framework to the dynamic hypergraph stochastic block model, establishing the theoretical properties of a novel spectral clustering algorithm based on transition probabilities. We further propose a change-point detection procedure, with rigorous asymptotic analysis of the resulting estimators. Section 4 presents comprehensive simulation studies validating our theoretical results and demonstrating the performance advantages of our methods over existing approaches across various scenarios. Section 5 applies our methodology to multiple real-world dynamic hypergraph datasets, including a primary school interaction data set and the Enron email corpus, revealing insights into community structures and significant temporal changes. All technical proofs and additional mathematical details are provided in the Appendix.

\section{AR(1) Hypergraph process}

\subsection{Model}

Let $\mathcal{V} = \{1, \cdots, p\}$ denote $p$ fixed nodes which do not change over time. Let $$\mathcal{E}^k = \{(j_1, \ldots, j_k):1 \le j_1 <\dots <j_k \leq p\}, 
\quad {\rm and} \quad  \mathcal{E} = \bigcup_{k=2}^K\mathcal{E}^k,$$
where $2\le K\le p$ is an fixed integer. A dynamic undirected hypergraph with the maximum $K$ nodes in each edge is denoted
by ${\bf X}_t = \{ X_\xi^t, \, \xi \in \mathcal{E}\}$ for $t=0, 1, 2, \ldots$, where
$X_\xi^t$ is binary, and $X_\xi^t$ takes value 1 when there is an edge connects all the nodes and only the notes in $\xi$ at time $t$, and value 0 when such an edge does not exist. For the directed hypographs, we change the definition of $\mathcal{E}^k$ to $$\mathcal{E}^k = \{(j_1, \ldots, j_k):1 \le j_1, \ldots , j_k\le q, \; {\rm and} \; j_1, \ldots, j_k \mbox{ are different}\}.$$ Then $X_\xi^t =1$ for $\xi = (j_1, \ldots, j_k)$ implies that there is a directed edge starting at $j_1$, going through $j_2, j_3, \ldots$, and ending at $j_k$.
The setting can be easily extended for the hypergraphs with selfloops.
We always assume in this paper that the edge processes  $\{ X_{\xi}^t, t= 0, 1, \ldots\}$ for different $ \xi \in \mathcal{E}$ are independent with each other.

\begin{definition} \label{def:ar1}
    An AR(1) hypergraph process is defined as
$$
X_{ \xi}^t=X_{ \xi}^{t-1} I\left(\varepsilon_{\xi}^t=0\right)+I\left(\varepsilon_{\xi}^t=1\right), \quad t \geq 1,
$$
where $I(\cdot)$ denotes the indicator function, and   $\varepsilon_{\xi}^t,\xi \in \mathcal{E}$ are independent innovations defined as follows.
$$
P\left(\varepsilon_{\xi}^t=1\right)=\alpha_{\xi}^t, \quad P\left(\varepsilon_{\xi}^t=-1\right)=\beta_{\xi}^t, \quad P\left(\varepsilon_{\xi}^t=0\right)=1-\alpha_{\xi}^t-\beta_{\xi}^t .
$$
with $\alpha_{\xi}^t, \beta_{\xi}^t$ are non-negative constants, and $\alpha_{\xi}^t+\beta_{\xi}^t \leq 1$.
\end{definition} 

The innovation term $\varepsilon_{\xi}^t$ is used via the two indicator functions to ensure that $X_{\xi}^t$ is binary. 
If we set the maximum size of hyperedges $K=2$, it reduces to AR(1) networks of \cite{jiang_autoregressive_2023}.

We can observe that the process $\left\{\mathbf{X}_t, t=0,1,2, \cdots\right\}$ is a Markov chain with 
$$
P\left(X_{\xi}^t=1 \mid X_{\xi}^{t-1}=0\right)=\alpha_{\xi}^t, \quad P\left(X_{\xi}^t=0 \mid X_{\xi}^{t-1}=1\right)=\beta_{\xi}^t,
$$
and $\left\{\mathbf{X}_t\right\}$ is a homogeneous Markov chain if
\begin{equation}
     \alpha_{\xi}^t \equiv \alpha_{\xi} \quad \text { and } \quad \beta_{\xi}^t \equiv \beta_{\xi} \quad \text { for all } t \geq 1 \text { and }\xi \in \mathcal{E} .\label{eq:ho}
\end{equation}
   

\subsection{Properties}

We firstly show the condition of strictly stationary AR(1) hypergraph process in the following proposition. 
\begin{proposition}
    $\left\{\mathbf{X}_t, t=0,1,2, \cdots\right\}$ is a strictly stationary process if condition (\ref{eq:ho}) holds and the connection probability of the initial hypergraph $\mathbf{X}_0 = \left(X_{\xi}^0\right)$ is 
\begin{equation}
   P\left(X_{\xi}^0=1\right)=\alpha_{\xi} /\left(\alpha_{\xi}+\beta_{\xi}\right), \xi \in \mathcal{E} . \label{eq:sta}
\end{equation}

Furthermore for any $\xi,\xi' \in$ $\mathcal{E}$ and $t, s \geq 0$,
$$
\begin{gathered}
E\left(X_{\xi}^t\right)=\frac{\alpha_{\xi}}{\alpha_{\xi}+\beta_{\xi}}, \quad \operatorname{Var}\left(X_{\xi}^t\right)=\frac{\alpha_{\xi} \beta_{\xi}}{\left(\alpha_{\xi}+\beta_{\xi}\right)^2}, \\
\rho_{\xi,\xi'}(|t-s|) \equiv \operatorname{Corr}\left(X_{\xi}^t, X_{\xi'}^s\right)= \begin{cases}\left(1-\alpha_{\xi}-\beta_{\xi}\right)^{|t-s|} & \text { if }\xi=\xi', \\
0 & \text { otherwise. }\end{cases}
\end{gathered}
$$
\end{proposition}

The Hamming distance counts the number of different hyperedges in the two hypergraphs, and is a measure the closeness of two hypergraphs.

\begin{definition}
    For any two sequence $\mathbf{A}=\left(A_{\xi}\right)$ and $\mathbf{B}=\left(B_{\xi}\right)$ of the same length, the Hamming distance is defined as $D_H(\mathbf{A}, \mathbf{B})=\sum_{\xi \in \mathcal{E}} I\left(A_{\xi} \neq B_{\xi}\right)$.
\end{definition}
\begin{proposition}
    Let $\left\{\mathbf{X}_t, t=0,1, \cdots\right\}$ be a stationary hypergraph process satisfying conditions (\ref{eq:ho}) and (\ref{eq:sta}). Let $d_H(|t-s|)=E\left\{D_H\left(\mathbf{X}_t, \mathbf{X}_s\right)\right\}$ for any $t, s \geq 0$. Then $d_H(0)=0$, and it holds for any $k \geq 1$ that
$$
\begin{aligned}
d_H(k) & =d_H(k-1)+\sum_{\xi \in \mathcal{E}} \frac{2 \alpha_{\xi} \beta_{\xi}}{\alpha_{\xi}+\beta_{\xi}}\left(1-\alpha_{\xi}-\beta_{\xi}\right)^{k-1} \\
& =\sum_{\xi \in \mathcal{E}} \frac{2 \alpha_{\xi} \beta_{\xi}}{\left(\alpha_{\xi}+\beta_{\xi}\right)^2}\left\{1-\left(1-\alpha_{\xi}-\beta_{\xi}\right)^k\right\} .
\end{aligned}
$$
\end{proposition}

The recursive formula shows how as $k$ increases, the expected distance grows and approaches to the limit $\lim_{k \rightarrow\infty}d_H(k) = \sum_{\xi \in \mathcal{E}} \frac{2 \alpha_{\xi} \beta_{\xi}}{\left(\alpha_{\xi}+\beta_{\xi}\right)^2}$.
This limit equals twice the sum of the variances of all hyperedges, indicating that as the time difference becomes large, the hypergraphs become effectively independent.

Next proposition below shows that $\left\{\mathbf{X}_t, t=0,1, \cdots\right\}$ is $\alpha$-mixing with exponentially decaying coefficients. This mixing property is crucial for establishing asymptotic properties of estimators for transition probabilities. 
\begin{definition}
    Let $\mathcal{F}_a^b$ be the $\sigma$-algebra generated by $\left\{X_{\xi}^k, a \leq k \leq b\right\}$. The $\alpha$-mixing coefficient of process $\left\{X_{\xi}^t, t=0,1, \cdots\right\}$ is defined as
$$
\alpha^{\xi}(\tau)=\sup _{k \in \mathbb{N}} \sup _{A \in \mathcal{F}_0^k, B \in \mathcal{F}_{k+\tau}^{\infty},}|P(A \cap B)-P(A) P(B)| .
$$
\end{definition}

\begin{proposition}
    Let $\{X_\xi\}$ be homogeneous Markov chain with $\alpha_{\xi}, \beta_{\xi}>0$. Then $\alpha^{\xi}(\tau) \leq$ $\rho_{\xi}(\tau)=\left(1-\alpha_{\xi}-\beta_{\xi}\right)^\tau$ for any $\tau \geq 1$.
\end{proposition}

\subsection{Estimation}
We assume that condition (\ref{eq:ho}) and (\ref{eq:sta}) hold.  By the Markovian property, the joint probability function of $\mathbf{X}_1, \cdots, \mathbf{X}_n$, conditional on $\mathbf{X}_0$, is 
\begin{align*}
     &\prod_{1 \leq t \leq n} P\left(\mathbf{X}_t \mid \mathbf{X}_{t-1}\right)= \prod_{1 \leq t \leq n}\prod_{\xi \in \mathcal{E}} P\left(X_{\xi}^t \mid X_{\xi}^{t-1}\right) \\
=& \prod_{1 \leq t \leq n}\prod_{\xi \in \mathcal{E}}\left(\alpha_{\xi}^t\right)^{X_{\xi}^t\left(1-X_{\xi}^{t-1}\right)}\left(1-\alpha_{\xi}^t\right)^{\left(1-X_{\xi}^t\right)\left(1-X_{\xi}^{t-1}\right)}\left(\beta_{\xi}^t\right)^{\left(1-X_{\xi}^t\right) X_{\xi}^{t-1}\left(1-\beta_{\xi}^t\right)^{X_{\xi}^t X_{\xi}^{t-1}} .}
\end{align*}
 Due to the independence of the existance of each hyepredges, the parameters $\left(\alpha_{\xi}, \beta_{\xi}\right)$, for different $\xi$, can be estimated separately. Then the maximum likelihood estimators are
\begin{equation}
    \widehat{\alpha}_{\xi}=\frac{\sum_{t=1}^n X_{\xi}^t\left(1-X_{\xi}^{t-1}\right)}{\sum_{t=1}^n\left(1-X_{\xi}^{t-1}\right)}, \quad \widehat{\beta}_{\xi}=\frac{\sum_{t=1}^n\left(1-X_{\xi}^t\right) X_{\xi}^{t-1}}{\sum_{t=1}^n X_{\xi}^{t-1}} .\label{eq:abe}
\end{equation}

For definiteness we shall set $0 / 0=1$. To state the asymptotic properties, we list some regularity conditions first.

C1. There exists a constant $l>0$ such that $\alpha_{\xi}, \beta_{\xi} \geq l$ and $\alpha_{\xi}+\beta_{\xi} \leq 1$ for all $\xi \in \mathcal{E}$.

C2. $n, p \rightarrow \infty$, and $(\log n)(\log \log n) \sqrt{\frac{\log p}{n}} \rightarrow 0$.

Condition $\mathrm{C} 1$ defines the parameter space, and condition $\mathrm{C} 2$ indicates that the number of nodes is allowed to diverge in a smaller order than $\exp \left\{\frac{n}{(\log n)^2(\log \log n)^2}\right\}$.

\begin{proposition}
    Let conditions (\ref{eq:ho}), C1 and C2 hold. For any constant $c \geq K$, there exists a large enough constant $C>0$
\begin{equation}
     \mathbb{P}\left(\max _{\xi \in \mathcal{E}}\left|\widehat{\alpha}_{\xi}-\alpha_{\xi}\right| \geq l^{-1} C \sqrt{\frac{\log p}{n}}\right) \leq 2 p^{K}\exp \{-c \log p\}, \label{eq:alpha_rate}
    \end{equation}
    \begin{equation}
 \mathbb{P}\left(\max _{\xi \in \mathcal{E}}\left|\widehat{\beta}_{\xi}-\beta_{\xi}\right| \geq l^{-1} C \sqrt{\frac{\log p}{n}}\right) \leq 2 p^{K}\exp \{-c \log p\} .\label{eq:beta_rate}
\end{equation}

\end{proposition}
\begin{remark}
    If we set the $K=2$, the result is the same to Proposition 6 in \cite{jiang_autoregressive_2023}. If we set $K = p$ such that including all possible hyperedges, then the result can be changed to for any constant $c\geq \log 2$, there exists a large enough constant $C>0$ such that,
    $$
\begin{aligned}
& \mathbb{P}\left(\max _{\xi \in \mathcal{E}}\left|\widehat{\alpha}_{\xi}-\alpha_{\xi}\right| \geq l^{-1} C \sqrt{\frac{ p}{n}}\right) \leq 2^{p+1}\exp \{-c p\}, \\
& \mathbb{P}\left(\max _{\xi \in \mathcal{E}}\left|\widehat{\beta}_{\xi}-\beta_{\xi}\right| \geq l^{-1} C \sqrt{\frac{ p}{n}}\right) \leq 2^{p+1}\exp \{-c p\} .
\end{aligned}
$$
The result can be proved similarly to Proposition 4.
\end{remark}

Proposition below implies that any fixed number of estimators $\left\{\widehat{\alpha}_{\xi}, \widehat{\beta}_{\xi}\right\}$ are jointly asymptotically normal.
To state this joint asymptotic normality, we introduce some notation first: let $\mathcal{E}_1=\left\{\xi_1,\xi_2,\cdots, \xi_{m_1}\right\}, \mathcal{E}_2=$ $\left\{k_1, k_2, \ldots,k_{m_2}\right\}$ be two arbitrary subsets of $\mathcal{E}$ with $m_1, m_2 \geq 1$ fixed. Denote $\Theta_{\mathcal{E}_1, \mathcal{E}_2}=\left(\alpha_{\xi_1}, \ldots, \alpha_{\xi_{m_1}}, \beta_{k_1}, \ldots, \beta_{k_{m_2}}\right)^{\top}$, and correspondingly denote the MLEs as $\widehat{\boldsymbol{\Theta}}_{\mathcal{E}_1, \mathcal{E}_2}=\left(\widehat{\alpha}_{\xi_1}, \ldots, \widehat{\alpha}_{\xi_{m_1}}, \widehat{\beta}_{k_1}, \ldots, \widehat{\beta}_{k_{m_2}}\right)^{\top}$.

\begin{proposition}
    Let conditions (\ref{eq:ho}), C1 and C2 hold. Then $\sqrt{n}\left(\widehat{\boldsymbol{\Theta}}_{\mathcal{E}_1, \mathcal{E}_2}-\boldsymbol{\Theta}_{\mathcal{E}_1, \mathcal{E}_2}\right) \rightarrow N\left(\mathbf{0}, \boldsymbol{\Sigma}_{\mathcal{E}_1, \mathcal{E}_2}\right)$, where $\boldsymbol{\Sigma}_{\mathcal{E}_1, \mathcal{E}_2}=\operatorname{diag}\left(\sigma_{11}, \ldots, \sigma_{m_1+m_2, m_1+m_2}\right)$ is a diagonal matrix with
$$
\begin{aligned}
& \sigma_{r r}=\frac{\alpha_{\xi_r}\left(1-\alpha_{\xi_r}\right)\left(\alpha_{\xi_r}+\beta_{\xi_r}\right)}{\beta_{\xi_r}}, \quad 1 \leq r \leq m_1, \\
& \sigma_{r r}=\frac{\beta_{k_r}\left(1-\beta_{k_r}\right)\left(\alpha_{k_r}+\beta_{k_r}\right)}{\alpha_{k_r}}, \quad m_1+1 \leq r \leq m_1+m_2 . \\
&
\end{aligned}
$$
\end{proposition}

\subsection{Model diagnostic}

As the only assumption in our AR(1) Hypergraph process is the independence of the innovation terms, we need to check the adequacy of the model by testing the independence of resedual sequence $\widehat{\mathcal{R}}_t \equiv\left(\widehat{\varepsilon}_{\xi}^t\right)$ for $t=1, \cdots, n$, where residual $\hat{\varepsilon}_{\xi}^t$ the estimated value of $E\left(\varepsilon_{\xi}^t \mid X_{\xi}^t, X_{\xi}^{t-1}\right)$ with expression 
$$
\begin{aligned}
\widehat{\varepsilon}_{\xi}^t & =\frac{\widehat{\alpha}_{\xi}}{1-\widehat{\beta}_{\xi}} I\left(X_{\xi}^t=1, X_{\xi}^{t-1}=1\right)-\frac{\widehat{\beta}_{\xi}}{1-\widehat{\alpha}_{\xi}} I\left(X_{\xi}^t=0, X_{\xi}^{t-1}=0\right) \\
& +I\left(X_{\xi}^t=1, X_{\xi}^{t-1}=0\right)-I\left(X_{\xi}^t=0, X_{\xi}^{t-1}=1\right), \xi \in \mathcal{E}, t=1, \cdots, n .
\end{aligned}
$$

Since $\hat{\varepsilon}_{\xi}^t, t=1, \cdots, n$, only take 4 different values for each $\xi \in \mathcal{E}$, we adopt the two-way, or three-way contingency table to test the independence of $\widehat{\mathcal{R}}_t$ and $\widehat{\mathcal{R}}_{t-1}$, or $\widehat{\mathcal{R}}_t, \widehat{\mathcal{R}}_{t-1}$ and $\widehat{\mathcal{R}}_{t-2}$. For example, the standard test statistic for the two-way contingency table is
$$
T=\frac{1}{n|\mathcal{E}|} \sum_{\xi \in \mathcal{E}} \sum_{k, \ell=1}^4\left\{n_{\xi}(k, \ell)-n_{\xi}(k, \cdot) n_{\xi}(\cdot, \ell) /(n-1)\right\}^2 /\left\{n_{\xi}(k, \cdot) n_{\xi}(\cdot, \ell) /(n-1)\right\},
$$
where $|\mathcal{E}|$ denotes the cardinality of $\mathcal{E}$, and for $1 \leq k, \ell \leq 4$,
$$
\begin{aligned}
& n_{\xi}(k, \ell)=\sum_{t=2}^n I\left\{\widehat{\varepsilon}_{\xi}^t=u_{\xi}(k), \;\hat{\varepsilon}_{\xi}^{t-1}=u_{\xi}(\ell)\right\}, \\
& n_{\xi}(k, \cdot)=\sum_{t=2}^n I\left\{\hat{\varepsilon}_{\xi}^t=u_{\xi}(k)\right\}, \quad n_{\xi}(\cdot, \ell)=\sum_{t=2}^n I\left\{\hat{\varepsilon}_{\xi}^{t-1}=u_{\xi}(\ell)\right\} .
\end{aligned}
$$
In the above expressions, $u_{\xi}(1)=-1, u_{\xi}(2)=-\frac{\widehat{\beta}_{\xi}}{1-\widehat{\alpha}_{\xi}}, u_{\xi}(3)=\frac{\widehat{\alpha}_{\xi}}{1-\widehat{\beta}_{\xi}}$ and $u_{\xi}(4)=1$. We calculate the $P$-values of the test $T$ based on the permutation algorithm in Algorithm 1.

\begin{algorithm}
\caption{Permutation Test for Residual Independence}
\begin{algorithmic}[1]
\Require Original residual sequence $\hat{\mathcal{E}}_1, \ldots, \hat{\mathcal{E}}_n$, test statistic $T$, number of permutations $M$
\Ensure $p$-value for the test
\State $\text{count} \gets 0$
\For{$j = 1$ to $M$}
    \State Generate a random permutation $\pi$ of indices $\{1, 2, \ldots, n\}$
    \State Calculate test statistic $T_j^{\star}$ in the same manner as $T$ but based on the permuted sequence 
    $\hat{\mathcal{E}}_{\pi(1)}, \ldots, \hat{\mathcal{E}}_{\pi(n)}$    
    \If{$T < T_j^{\star}$}
         $\text{count} \gets \text{count} + 1$
  \EndIf
\EndFor
\State $p\text{-value} \gets \frac{\text{count}}{M}$
\State \Return $p\text{-value}$
\end{algorithmic}
\end{algorithm}

\section{AR(1) Hypergraph Stochastic Blockmodel}


The AR(1) framework proposed in Section 2 can serve as a building block to accommodate various hypergraph structures. As as illustration, we present in this section a new dynamic hypergraph stochastic block (DHSB) model based on the AR(1) framework.

\subsection{Models}

Hypergraph stochastic block models are undirected and contain no self-loops, the transition probabilities of an edge depend on the memberships of the nodes in the edge. We denote the node membership at time $t$ by $\psi_t$, i.e. for any $i\in [p]$, node $i$ belongs to
community $\psi_t(i) \in [q]$, where $q (<<p)$ is the number of communities, which is assumed unchanged over time. Note that some communities may contain no members at some times, indicating effectively the disappearance of those communities at those times. 

For any edge $\xi = (j_1, \ldots, j_k)$ with $k$ notes, let
\[
{\boldsymbol{\psi}}_t(\xi) = \{ \psi_t(j_1), \cdots, \psi_t(j_k)\}.
\]



\begin{definition}
    An AR(1) DHSB process $\boldsymbol{X}_t = \{ X_{\xi}^t, \, \xi \in \mathcal{E}\} $ is defined as in Definition \ref{def:ar1} but with
    $$
\alpha_{\xi}^t=\theta_{\boldsymbol{\psi}_t(\xi)}^t \qquad {\rm and} \qquad
\beta_{\xi}^t=\eta_{\boldsymbol{\psi}_t(\xi)}^t, 
$$
where $ \theta_\zeta^t , \eta_\zeta^t \ge 0$ and $ \theta_\zeta^t + \eta_\zeta^t\le 1$ for any
\[
\zeta \in \mathcal{E}_q = \bigcup_{k=2}^K 
\{ (j_1, \ldots, j_k): 1 \le j_1, \ldots, j_k \le q\}.
\]
Furthermore the values of $\theta^t_\zeta$ and $\eta^t_{\zeta}$ are invariant with respect to any 
permutation of the elements in $\zeta$.
\end{definition}

The community structure is central to the stochastic block model, making community detection a fundamental analytical task. For identifying the latent community structures of the static hypergraph models,  \cite{ke_community_2020} introduced a regularized tensor power iteration algorithm, and \cite{ghoshdastidar_consistency_2017} employed a spectral clustering algorithm. In Section 3.2, we propose a new spectral clustering algorithm specifically designed for AR(1) DHSB processes. 


\subsection{Estimating community $\psi(\cdot)$}

We assume that the membership of the nodes and the transition probilities unchanged over time, i.e.
$$\psi_t(\cdot) \equiv \psi(\cdot) \quad \text { and } \quad\left(\theta_{\boldsymbol{\psi}(\cdot)}^t, \eta_{\boldsymbol{\psi}(\cdot)}^t\right) \equiv\left(\theta_{\boldsymbol{\psi}(\cdot)}, \eta_{\boldsymbol{\psi}(\cdot)}\right), \quad t=1, \cdots, n. $$
The static hypergraph partition problem is formulated via minimizing normalized hypergraph cut. Set $\{\mathcal{V}_1, \ldots, \mathcal{V}_q\}$ be a partition of $\mathcal{V}$ and the boundary set $$\partial \mathcal{V}_j=\left\{\xi \in \mathcal{E}: \xi \cap \mathcal{V}_j \neq \emptyset, \xi \cap \mathcal{V}_j^c \neq \emptyset\right\}$$ denotes the set of edges that are cut when the nodes are divided into $\mathcal{V}_j$ and $\mathcal{V}_j^c=\mathcal{V} \backslash \mathcal{V}_j$. Then normalized hypergraph cut is defined by 
$$\operatorname{Ncut}\left(\mathcal{V}_1, \ldots, \mathcal{V}_q\right)=\sum_{j=1}^q \frac{\operatorname{vol}\left(\partial \mathcal{V}_j\right)}{\operatorname{vol}\left(\mathcal{V}_j\right)}$$ where $$ \operatorname{vol}\left(\mathcal{V}_j\right)=\sum_{v \in \mathcal{V}_j} \operatorname{deg}(v),\quad \operatorname{vol}\left(\partial \mathcal{V}_j\right)=\sum_{\xi \in \partial \mathcal{V}_j} \frac{\left|\xi \cap \mathcal{V}_j\right| \cdot\left|\xi \cap \mathcal{V}_j^c\right|}{|\xi|},$$ 
where $|\xi|$ denotes the cardinarity of
$\xi$, and deg$(v)$ denotes the degree of node $v$ which is
defined as the total number of edges containing $v$.
Then minimizing normalized hypergraph cut is to seek the partition that minimizes the connections between communities and maximizes the connections within communities. 
By \cite{zhou2006learning} and \cite{ghoshdastidar_consistency_2017}, a relaxed version of the solution for this combinatorial optimization problem leads to the eigen-decomposition of hypergraph Laplacian:
$$
\mathcal{L}= \mathbf{I}-\mathcal{D}^{-1 / 2} \mathbf{H} \mathbf{\Delta}^{-1} \mathbf{H}^T \mathcal{D}^{-1 / 2}
$$
where $\mathbf{H} \in \{0,1\}^{p \times |\mathcal{E}|}$ is the hypergraph incident matrix with $H_{ij} = 1$ if node $i$ is contained in the $j$-th
hyperedge, and 0 otherwise; the matrices $\mathcal{D} \in \mathbb{R}^{p \times p}, \mathbf{\Delta} \in \mathbb{R}^{|\mathcal{E}| \times|\mathcal{E}|}$ are diagonal with $(\mathcal{D})_{ii}=\operatorname{deg}(i)$ and $(\mathbf{\Delta})_{jj}=|\xi_j|$. 
Let $\omega_{\xi}=P(X_\xi=1)$, the population hypergraph Laplacian can be rewritten as 
$$
\textbf{L}=\mathbf{I}-\mathbf{D}^{-1 / 2} \mathbf{A}  \mathbf{D}^{-1 / 2} \text{ where  }
\textbf{D}_{i i}=\sum_{\xi \in \mathcal{E}} \omega_\xi\left(a_{\xi}\right)_i, \mathbf{A} = \sum_{\xi \in \mathcal{E}} \frac{\omega_\xi}{\left|\xi\right|}  a_{\xi} a_{\xi}^T,
$$
where $\left(a_{\xi}\right)_i=1$, if node $i \in \xi$, and 0 otherwise.
The effectiveness of this spectral approach stems from the fact that nodes in the same community exhibit similar behavior exhibited in matrix $\mathbf{A}$. Without the terms $1/|\xi|$, $(\mathbf{A})_{i,j}$, is the expected number of hyperedges containing nodes $i$ and $j$. The weight $1/|\xi|$ is applied in order to differentiate the relative importance of the hyperedges in different orders: the smaller the order is, the more important the link is.

To leverage the dynamic structure of hypergraphs, we define the Laplacian based on transition probabilities instead. As the nodes in the same community are more likely to be connected with each other, both $\alpha_\xi$ and $1-\beta_\xi$ can be regarded as the similarity measures for the nodes in edge $\xi$. Hence we define two $p\times p$ similarity matrices $\mathbf{A}_1$ and $\mathbf{A}_2$ as follows:
\begin{align*}
    \left(\mathbf{A}_1\right)_{i,j}&=
\begin{cases}\sum_{m=2}^K \sum_{|\xi|=m,i,j \in \xi}\frac{1}{m}\alpha_{\xi} & \text { if }i \neq j, \\
\sum_{m=2}^K \sum_{|\xi|=m,i \in \xi}\frac{1}{m}\alpha_{\xi} 
& \text { if } i=j,\end{cases}
\\
\left(\mathbf{A}_2\right)_{i,j}&=
\begin{cases}\sum_{m=2}^K \sum_{|\xi|=m,i,j \in \xi}\frac{1}{m}(1-\beta_{\xi}) & \text { if }i \neq j, \\ 
\sum_{m=2}^K \sum_{|\xi|=m,i \in \xi}\frac{1}{m}(1-\beta_{\xi}) 
& \text { if } i=j.\end{cases}
\end{align*}

Let $\mathbf{Z}=\left(z_{\xi}\right) \in\{0,1\}^{p \times q}$ is the membership matrix such that it has exactly one 1 in each row and at least one 1 in each column, i.e. each column collects all the nodes in one 
of the $q$ communities. Furthermore,  let $\mathbf{\Omega}_1$ and $ \mathbf{\Omega}_2 $ be
$q\times q$ matrices defined as follows.
$$
\begin{aligned}
& (\mathbf{\Omega}_1)_{i,j}=\sum_{m=2}^K\sum_{i_3<\ldots<i_m, i_1,i_2 \notin \{i_3,\ldots,i_m\}, \psi(i_1)=i, \psi(i_2)=j}\frac{1}{m}\theta_{\psi(i_1),\psi(i_2),\psi(i_3),\ldots, \psi(i_m)}, \\
& (\mathbf{\Omega}_2)_{i,j}=\sum_{m=2}^K\sum_{i_3<\ldots<i_m, i_1,i_2 \notin \{i_3,\ldots,i_m\}, \psi(i_1)=i, \psi(i_2)=j}\frac{1}{m}(1-\eta_{\psi(i_1),\psi(i_2),\psi(i_3),\ldots, \psi(i_m)}).
\end{aligned}
$$
Then it holds that
\begin{align*}
    \mathbf{A}_1=\mathbf{Z} \mathbf{\Omega}_1 \mathbf{Z}^{\top} - \mathbf{J}_1, \qquad
 \mathbf{A}_2=\mathbf{Z} \mathbf{\Omega}_2 \mathbf{Z}^{\top}-
 \mathbf{J}_2,
\end{align*}
where $\mathbf{J}_1,\mathbf{J}_2 \in \mathbb{R}^{p \times p}$ is diagonal matrix with $(\mathbf{J}_1)_{i i}=(\mathbf{\Omega}_1)_{\psi(i),\psi(i)}-(\mathbf{A}_1)_{ii}=(\mathbf{J}_1)_{j j}$ and $(\mathbf{J}_2)_{i i}=(\mathbf{\Omega}_2)_{\psi(i),\psi(i)}-(\mathbf{A}_2)_{ii}=(\mathbf{J}_2)_{j j}$. 

Let $d_{i,k} = \sum_{j=1}^p (\mathbf{A}_k)_{i,j}$, and
$\mathbf{D}_k$ be the $p\times p$ diagonal matrix with $d_{i,k}$ as its $i$-th main diagonal element, $i\in[p], \, k=1, 2$. Define
$$
\mathbf{L}_1=\mathbf{I}-\mathbf{D}_1^{-1 / 2} \mathbf{A}_1 \mathbf{D}_1^{-1 / 2}, \quad \mathbf{L}_2=\mathbf{I}-\mathbf{D}_2^{-1 / 2} \mathbf{A}_2 \mathbf{D}_2^{-1 / 2}, \quad \mathbf{L}=\mathbf{L}_1+\mathbf{L}_2.
$$
The new spectral clustering algorithm is to apply the
$K$-mean clustering to the $q$ leading eigenvectors 
of the normalized Laplacian $\mathbf{L}$. The intuition
behind this algorithm can be understood as follows:
For any unit vector $\mathbf{\gamma} = (\gamma_1, \cdots, \gamma_p)^\top$, it holds
\[
\mathbf{\gamma}^\top \mathbf{L}\mathbf{\gamma}= \sum_{i<j} (\mathbf{A}_1)_{i,j}\big( {\gamma_i \over \sqrt{d_{i,1}}} - {\gamma_j \over \sqrt{d_{j,1}}} \big)^2
+\sum_{i<j} (\mathbf{A}_2)_{i,j}\big( {\gamma_i \over \sqrt{d_{i,1}}} - {\gamma_j \over \sqrt{d_{j,1}}} \big)^2.
\]
Thus the eigenvectors of $\mathbf{L}$ corresponding to the smallest eigenvalues will minimize the two terms on the RHS of the above equation. Note that for
$k=1, 2$, the values of $(\mathbf{A}_k)_{i,j}$ for nodes $i$ and $j$ in the same community tend to be greater than those when $i$ and $j$ in different community. Furthermore, $d_{i,k} = d_{j,k}$ when nodes $i$ and $j$ are in the same community. Hence the leading eigenvectors of $\mathbf{L}$ tends to take same values at the components belonging to the same communities, in order to minimize $\mathbf{\gamma}^\top \mathbf{L}\mathbf{\gamma}$. Here, the leading eigenvalues refer to the $q$ smallest eigenvalues of $\mathbf{L},$ and the leading eigenvectors refer to the eigenvectors that correspond to
the $q$ smallest eigenvalues of $\mathbf{L}.$
The following lemma shows that the block structure in the membership matrix $\mathbf{Z}$ can be
recovered by the leading eigenvectors of $\mathbf{L}$.

\begin{proposition}
 There are $q$ orthonormal eigenvectors of matrix $\mathbf{L}$ which are the columns of the matrix $\mathbf{\Gamma} = \mathbf{Z}(\mathbf{Z}^T \mathbf{Z})^{-1/2} \mathbf{U}$, where $\mathbf{U} \in \mathbb{R}^{q \times q}$ is orthonormal.

    Moverover, define 
    \begin{align*}
        \delta=&\lambda_{\min}(\boldsymbol{\Omega}_1)\min_{1 \leq i\leq q}\frac{(\mathbf{Z}^T \mathbf{Z})_{ii}}{(\mathbf{D}_1)_{ii}}+\lambda_{\min}(\boldsymbol{\Omega}_2)\min_{1 \leq i\leq q}\frac{(\mathbf{Z}^T \mathbf{Z})_{ii}}{(\mathbf{D}_2)_{ii}}\\&- \max_{1 \leq i \leq n}\left(\frac{(\mathbf{J}_1)_{ii}}{(\mathbf{D}_1)_{ii}}+\frac{(\mathbf{J}_2)_{ii}}{(\mathbf{D}_2)_{ii}}\right) -\min_{1 \leq i \leq n}\left(\frac{(\mathbf{J}_1)_{ii}}{(\mathbf{D}_1)_{ii}}+\frac{(\mathbf{J}_2)_{ii}}{(\mathbf{D}_2)_{ii}}\right). 
    \end{align*}
 If $\delta>0$, then the columns of $\mathbf{\Gamma}$ are the $q$ leading eigenvectors of $\mathbf{L}$.
 Here, $\lambda_{\min }(\boldsymbol{\Omega})$ refers to the smallest eigenvalue of $\boldsymbol{\Omega}$. 
\end{proposition}
\begin{remark}
    (i) In the above result, $\mathbf{Z}^T \mathbf{Z}$ is a diagonal matrix with entries being the sizes of the $q$ communities. Hence, both $\mathbf{Z}^T \mathbf{Z}$ and $\mathbf{U}$ are of the rank $q$. This observation leads to the conclusion that the matrix $\mathbf{\Gamma}$ contains exactly $q$ distinct rows, each corresponding to one community. 


    (ii) The term $\delta$ arises from the lower bound on the eigen-gap between the $q$th and ($q+1$)th smallest eigenvalues of $\mathbf{L}$.  Without $\delta >0$, $\mathbf{\Gamma}$ may not correspond to the leading eigenvectors of $\mathbf{L}$. This point is particularly crucial in hypergraph analysis due to the large number of parameters involved.
    
\end{remark}
 Now we state our new spectral clustering algorithm based in Proposition 6.

\begin{algorithm}[H]
\caption{Spectral Clustering for Dynamic Stochastic Block Model}
\label{alg:dsb-spectral}
\begin{algorithmic}[1]
\Require Dynamic hypergraph sequence $\mathbf{X}_t$, known number of communities $q\le p$
\Ensure Community labels $\{\widehat\psi(i)\}_{i=1}^p$
\State Estimate edge parameters $\{\hat\alpha_\xi,\hat\beta_\xi\}$ from $\mathbf{X}_t$.

\State Define $p \times p$ matrices:
$$
\begin{gathered}
\widehat{\mathbf{L}}_1=\mathbf{I} - \widehat{\mathbf{D}}_1^{-1 / 2} \widehat{\mathbf{A}}_1 \widehat{\mathbf{D}}_1^{-1 / 2}, \quad \widehat{\mathbf{L}}_2=\mathbf{I} - \widehat{\mathbf{D}}_2^{-1 / 2} \widehat{\mathbf{A}}_2 \widehat{\mathbf{D}}_2^{-1 / 2}, \quad 
\widehat{\mathbf{L}} = \widehat{\mathbf{L}}_1+\widehat{\mathbf{L}}_2
\end{gathered}
$$
where $\widehat{\mathbf{A}}_1, \widehat{\mathbf{A}}_2, \widehat{\mathbf{D}}_1, \widehat{\mathbf{D}}_2$ can be estimated via $\hat{\alpha}_{\xi}$ and $\hat{\beta}_{\xi}$.

\State Compute eigen-decompostion for $\widehat{\mathbf{L}}$:
$$\widehat{\mathbf{L}} \equiv \widehat{\mathbf{L}}_1+\widehat{\mathbf{L}}_2=\widehat{\boldsymbol{\Gamma}} \operatorname{diag}\left(\widehat{\lambda}_1, \cdots, \widehat{\lambda}_p\right) \widehat{\boldsymbol{\Gamma}}^{\top}, \quad \text { where } \widehat{\lambda}_1 \leq \ldots \leq \widehat{\lambda}_p.$$ Let $\widehat{\boldsymbol{\Gamma}}_q$ the $p \times q$ matrix consisting of the first $q$ columns of $\widehat{\boldsymbol{\Gamma}}$.

\State Apply the $k$-means clustering algorithm to the $p$ rows of $\widehat{\boldsymbol{\Gamma}}_q$. 
\State \textbf{Return:}
Label of each node $i$ 
\[
\widehat\psi(i)\in\{1,\dots,q\}
\]
from the $k$‐means result.
\end{algorithmic}
\end{algorithm}

The following theorem justifies the validity of using $\widehat{\mathbf{L}}$ for spectral clustering. Note that $\|\cdot\|_2$ and $\|\cdot\|_F$ denote, respectively, the $L_2$ and the Frobenius norm of the matrices.

\begin{theorem}
    Let conditions (1), C1 and C2 hold, and $\frac{\sqrt{q}}{\delta}(\sqrt{\frac{\log np}{p^{K-1}n}} + \frac{1}{n})) \rightarrow 0$, as $n, p \rightarrow \infty$. Then it holds that
\begin{align}
    \max _{i=1, \ldots, p}\left|\lambda_i-\widehat{\lambda}_i\right| \leq\|\widehat{\mathbf{L}}-\mathbf{L}\|_2 \leq O_p\left(\sqrt{\frac{\log (p n)}{n p^{K-1}}}+\frac{1}{n}\right) .
\end{align}
Moreover, for any constant $B>0$, there exists a constant $C>0$ such that the inequality
\begin{align}
\left\|\widehat{\boldsymbol{\Gamma}}_q-\mathbf{\Gamma}_q \mathbf{O}_q\right\|_F \leq \frac{\sqrt{8q}C}{\delta}(\sqrt{\frac{\log np}{p^{K-1}n}} + \frac{1}{n}))
\end{align}
holds with probability greater than $1-8p[(np)^{-(B+1)}+\exp(-Bp^{(K-1)/2})]$, where $\mathbf{O}_q$ is a $q \times q$ orthogonal matrix.
\end{theorem}
It follows that the leading eigenvalues of $\mathbf{L}$ can be consistently recovered by the leading eigenvalues of $\widehat{\mathbf{L}}$.  The leading eigenvectors of $\mathbf{L}$ can also be consistently estimated, subject to a rotation.

\begin{remark}
    By Theorem 10 of \cite{jiang_autoregressive_2023}, the error bound for the spectral clustering algorithm is $\max _{i=1, \ldots, p}\left|\lambda_i^2-\widehat{\lambda}_i^2\right| \leq\|\widehat{\mathbf{L}} \widehat{\mathbf{L}}-\mathbf{L} \mathbf{L}\|_2 \leq\|\widehat{\mathbf{L}} \widehat{\mathbf{L}}-\mathbf{L L}\|_F=O_p\left(\sqrt{\frac{\log (p n)}{n p}}+\frac{1}{n}+\frac{1}{p}\right) .$ the term $\frac{1}{p}$ comes from the bias caused by the inconsistent estimation of diagonal terms of $\widehat{\mathbf{L}} \widehat{\mathbf{L}}$, while in our case we only consider $\widehat{\mathbf{L}}$ without removal of diagonal terms.

\end{remark}

Proposition 6 indicates that there are only $q$ distinct rows in $\boldsymbol{\Gamma}_q$, and, therefore, also $q$ distinct rows in $\boldsymbol{\Gamma}_q \mathbf{O}_q$, corresponding to the $q$ latent communities for the $p$ nodes. This paves the way for the $k$-means algorithm stated below. Put
$$
\mathcal{M}_{p, q}=\left\{\mathbf{M} \in \mathcal{R}^{p \times q}: \mathbf{M} \text { has } q \text { distinct rows }\right\} .
$$

The $k$-means clustering algorithm: Let
$$
\left(\widehat{\mathbf{c}}_1, \cdots, \widehat{\mathbf{c}}_p\right)^{\top}=\arg \min _{\mathbf{M} \in \mathcal{M}_{p, q}}\left\|\widehat{\boldsymbol{\Gamma}}_q-\mathbf{M}\right\|_F^2 .
$$

There are only $q$ distinct vectors among $\widehat{\mathbf{c}}_1, \cdots, \widehat{\mathbf{c}}_p$, forming the $q$ communities. Theorem 11 below shows that they are identical to the latent communities of the $p$ nodes under (6) and (7). The latter holds if $\frac{\sqrt{24qs_{\max }}C}{\delta}(\sqrt{\frac{\log np}{p^{K-1}n}} + \frac{1}{n}) \rightarrow 0$, where $s_{\max }=$ $\max \left\{s_1, \ldots, s_q\right\}$ is the size of the largest community.

\begin{theorem}
    For (2) holds and
\begin{align}
    \sqrt{\frac{1}{s_{\max }}}>\frac{\sqrt{24q}C}{\delta}(\sqrt{\frac{\log np}{p^{K-1}n}} + \frac{1}{n}) .
\end{align}

Then $\widehat{\mathbf{c}}_i=\widehat{\mathbf{c}}_j$ if and only if $\psi(i)=\psi(j), 1 \leq i, j \leq p$.
\end{theorem}

\subsection{Estimation for $\theta_{\boldsymbol{c}}$ and $\eta_{\boldsymbol{c}}$}
We can use the estimation of communities to estimate the transition probabilities $\theta_{\boldsymbol{c}}$ and $\eta_{\boldsymbol{c}}$. 
Define $\mathbf{\Psi} =  \bigcup_{k=2}^K\mathbf{\Psi}^k = \bigcup_{k=2}^K\{(c_1,c_2, \ldots, c_k):1 \leq c_1 \leq \dots\leq c_k \leq q\}$ and 
$$
S_{\boldsymbol{c}}= \{(i_1,\ldots,i_k):i_j \in V_{c_j}, i_j\neq i_l, \quad \forall 1 \leq j,l \leq k\}
$$where $V_{c_j}$ is the set of nodes that are in community $c_j$ and  $\boldsymbol{c} = (c_1, c_2, \ldots, c_k) \in \mathbf{\Psi}$. 

For any $\boldsymbol{c} = (c_1, c_2, \ldots, c_k) \in \mathbf{\Psi}$, based on the procedure presented in Section 3.2, we can obtain an estimated membership function $\boldsymbol{\widehat{\psi}}(\cdot)$. Consequently, the MLEs for $\left(\theta_{\boldsymbol{c}}, \eta_{\boldsymbol{c}}\right)$ admit the form
$$
\begin{aligned}
\widehat{\theta}_{\boldsymbol{c}} & =\sum_{\xi \in \hat{S}_{\boldsymbol{c}}} \sum_{t=1}^n X_{\xi}^t\left(1-X_{\xi}^{t-1}\right) / \sum_{\xi \in \hat{S}_{\boldsymbol{c}}} \sum_{t=1}^n\left(1-X_{\xi}^{t-1}\right), \\
\widehat{\eta}_{\boldsymbol{c}} & =\sum_{\xi \in \hat{S}_{\boldsymbol{c}}} \sum_{t=1}^n\left(1-X_{\xi}^t\right) X_{\xi}^{t-1} / \sum_{\xi \in \hat{S}_{\boldsymbol{c}}} \sum_{t=1}^n X_{\xi}^{t-1},
\end{aligned}
$$
where
$$
\hat{S}_{\boldsymbol{c}}= \{(i_1,\ldots,i_k):i_j \in \hat{V}_{c_j}, i_j\neq i_l, \quad \forall 1 \leq j,l \leq k\}
$$
is the set of size-$k$ hyperedges whose nodes are in a given community set $\boldsymbol{c}$ and $\hat{V}_{c_j}$ is the set of nodes that are estimated in community $c_j$.  

Theorem 2 implies that the memberships of the nodes can be consistently recovered. Consequently, the consistency and the asymptotic normality of the MLEs $\widehat{\theta}_{\boldsymbol{c}}$ and $\widehat{\eta}_{\boldsymbol{c}}$ can be established in the same manner as for Propositions 4 and 5. We state the results below.

Let $\mathbf{\Psi}_1=\left\{\zeta_1,\zeta_2,\cdots, \zeta_{m_1}\right\}, \mathbf{\Psi}_2=$ $\left\{k_1, k_2, \ldots,k_{m_2}\right\}$ be two arbitrary subsets of $\mathbf{\Psi}$ with $m_1, m_2 \geq 1$ fixed. Denote $\boldsymbol{\Phi}_{\mathbf{\Psi}_1, \mathbf{\Psi}_2}=\left(\theta_{\zeta_1}, \ldots, \theta_{\zeta_{m_1}}, \eta_{k_1}, \ldots, \eta_{k_{m_2}}\right)^{\top}$, and correspondingly denote the MLEs as $\boldsymbol{\Phi}_{\mathbf{\Psi}_1, \mathbf{\Psi}_2}=\left(\widehat{\theta}_{\zeta_1}, \ldots, \widehat{\theta}_{\zeta_{m_1}}, \widehat{\eta}_{k_1}, \ldots, \widehat{\eta}_{k_{m_2}}\right)^{\top}$. Put $\mathbf{N}_{\mathbf{\Psi}_1, \mathbf{\Psi}_2}=\operatorname{diag}\left(n_{\zeta_1}, \ldots, n_{\zeta_{m_1}}, n_{k_1}, \ldots, n_{k_{m_2}}\right)$ where $n_{\zeta}$ is the cardinality of $S_{\zeta}$.

\begin{theorem}
    Let $K$ be the maximum size of hyperedges and conditions of (1), C1 and C2 hold, and $\frac{\sqrt{24qs_{\max }}C}{\delta}(\sqrt{\frac{\log np}{p^{K-1}n}} + \frac{1}{n}) \rightarrow 0$. Then it holds that
$$
\max _{\zeta \in \mathbf{\Psi}}\left|\widehat{\theta}_{\zeta}-\theta_{\zeta}\right|=O_p\left(\sqrt{\frac{\log q}{n s_{\min }^{2}}}\right) \quad \text { and } \max _{\zeta \in \mathbf{\Psi}}\left|\widehat{\eta}_{\zeta}-\eta_{\zeta}\right|=O_p\left(\sqrt{\frac{\log q}{n s_{\min }^{2}}}\right) \text {, }
$$
where $s_{\min }=\min \left\{s_1, \ldots, s_q\right\}$.
\end{theorem} 

\begin{theorem}
    Let the condition of Theorem 3 hold. Then
$$
\sqrt{n} \mathbf{N}_{\mathbf{\Psi}_1, \mathbf{\Psi}_2}^{\frac{1}{2}}\left(\widehat{\boldsymbol{\Phi}}_{\mathbf{\Psi}_1, \mathbf{\Psi}_2}-\boldsymbol{\Phi}_{\mathbf{\Psi}_1, \mathbf{\Psi}_2}\right) \rightarrow N\left(\mathbf{0}, \widetilde{\boldsymbol{\Sigma}}_{\mathbf{\Psi}_1, \mathbf{\Psi}_2}\right),
$$
where $\widetilde{\boldsymbol{\Sigma}}_{\mathbf{\Psi}_1, \mathbf{\Psi}_2}=\operatorname{diag}\left(\widetilde{\sigma}_{11}, \ldots, \widetilde{\sigma}_{m_1+m_2, m_1+m_2}\right)$ with
$$
\begin{aligned}
& \widetilde{\sigma}_{r r}=\frac{\theta_{\zeta_r}\left(1-\theta_{\zeta_r}\right)\left(\theta_{\zeta_r}+\eta_{\zeta_r}\right)}{\eta_{\zeta_r}}, \quad 1 \leq r \leq m_1, \\
& \widetilde{\sigma}_{r r}=\frac{\eta_{k_r}\left(1-\eta_{k_r}\right)\left(\theta_{k_r}+\eta_{k_r}\right)}{\theta_{k_r}}, \quad m_1+1 \leq r \leq m_1+m_2 .
\end{aligned}
$$
\end{theorem}

\subsection{Estimation for Change-point}
    Now, consider the case that there is a change point $\tau_0$ at which the membership of nodes and the transition probability $\{\theta_{\boldsymbol{\psi}}, \eta_{\boldsymbol{\psi}}\}$ change. Assume that $\tau_0 \in [n_0, n-n_0]$ where $n_0
    \ge 1$ is an integer. 
We denote $\widehat{\boldsymbol{\psi}}^{1, n}$, to reflect the fact that the community clustering was carried out using the data $\mathbf{X}_1, \cdots, \mathbf{X}_n$ (conditionally on $\mathbf{X}_0$ ). Furthermore, the maximum log-likelihood is denoted by

$$
\widehat{l}\left(\widehat{\psi}^{1, n}\right)=l\left(\left\{\widehat{\theta}_{\boldsymbol{c}}, \widehat{\eta}_{\boldsymbol{c}}\right\} ; \widehat{\psi}^{1, n}\right)
$$

to emphasize that both node clustering and the estimation of transition probabilities are based on the data$\mathbf{X}_1, \cdots, \mathbf{X}_n$.

    We assume that prior to the change point, the hypergraph process is stationary with membership function$\boldsymbol{\psi}^{1,\tau_0}(\cdot)$ and transition probability parameter $\theta_{1, \boldsymbol{c}}, \eta_{1, \boldsymbol{c}}$. After the change point, the hypergraph process is assumed to be stationary with membership function $\boldsymbol{\psi}^{\tau_0+1,n}(\cdot)$ and transition probability parameter $\theta_{2, \boldsymbol{c}}, \eta_{2, \boldsymbol{c}}$. 

    Then, we estimate the change point by the maximum likelihood method:
    $$
\widehat{\tau}=\arg \max _{n_0 \leq \tau \leq n-n_0}\left\{\widehat{l}\left( \widehat{\boldsymbol{\psi}}^{1, \tau}\right)+\widehat{l}\left( \widehat{\boldsymbol{\psi}}^{\tau+1, n}\right)\right\}.
$$

To measure the difference between the two sets of transition probabilities before and after the change, we put

$$
\Delta_F^2=\frac{1}{p^{K-1}}\left(\left\|\mathbf{A}_{1,1}-\mathbf{A}_{2,1}\right\|_F^2+\left\|\mathbf{A}_{1,2}-\mathbf{A}_{2,2}\right\|_F^2\right)
$$

where the four $p \times p$ matrices are defined as
$$
\begin{aligned}
    & \left(\mathbf{A}_{1,1}\right)_{i,j}
=\begin{cases}\sum_{m=2}^K \sum_{i_3<\ldots<i_m, i,j \notin \{i_3,\ldots,i_m\} }\frac{1}{m}\theta_{1,\psi(i),\psi(j),\psi(i_3),\ldots, \psi(i_m)}& \text { if }i \neq j, \\
\sum_{m=2}^K \sum_{i_2<\ldots<i_m, i \notin \{i_2,\ldots,i_m\} }\frac{1}{m}\theta_{1,\psi(i),\psi(i_2),\psi(i_3),\ldots, \psi(i_m)} & \text { if } i=j\end{cases}\\
& \left(\mathbf{A}_{1,2}\right)_{i,j}
=\begin{cases}\sum_{m=2}^K \sum_{i_3<\ldots<i_m, i,j \notin \{i_3,\ldots,i_m\} }\frac{1}{m}\eta_{1,\psi(i),\psi(j),\psi(i_3),\ldots, \psi(i_m)}& \text { if }i \neq j, \\
\sum_{m=2}^K \sum_{i_2<\ldots<i_m, i \notin \{i_2,\ldots,i_m\} }\frac{1}{m}\eta_{1,\psi(i),\psi(i_2),\psi(i_3),\ldots, \psi(i_m)} & \text { if } i=j\end{cases}\\
& \left(\mathbf{A}_{2,1}\right)_{i,j}
=\begin{cases}\sum_{m=2}^K \sum_{i_3<\ldots<i_m, i,j \notin \{i_3,\ldots,i_m\} }\frac{1}{m}\theta_{2,\psi(i),\psi(j),\psi(i_3),\ldots, \psi(i_m)}& \text { if }i \neq j, \\
\sum_{m=2}^K \sum_{i_2<\ldots<i_m, i \notin \{i_2,\ldots,i_m\} }\frac{1}{m}\theta_{2,\psi(i),\psi(i_2),\psi(i_3),\ldots, \psi(i_m)} & \text { if } i=j\end{cases}\\
& \left(\mathbf{A}_{2,2}\right)_{i,j}
=\begin{cases}\sum_{m=2}^K \sum_{i_3<\ldots<i_m, i,j \notin \{i_3,\ldots,i_m\} }\frac{1}{m}\eta_{2,\psi(i),\psi(j),\psi(i_3),\ldots, \psi(i_m)}& \text { if }i \neq j, \\
\sum_{m=2}^K \sum_{i_2<\ldots<i_m, i \notin \{i_2,\ldots,i_m\} }\frac{1}{m}\eta_{2,\psi(i),\psi(i_2),\psi(i_3),\ldots, \psi(i_m)} & \text { if } i=j\end{cases}
\end{aligned}
$$

Note that $\Delta_F$ can be viewed as the signal strength for detecting the change point $\tau_0$. Let $s_{\max }, s_{\min }$ denote, respectively, the largest, and the smallest community size among all the communities before and after the change. We denote the normalized Laplacian matrices corresponding to $\mathbf{A}_{i,j}$ as $\mathbf{L}_{i,j}$ for $i,j=1,2$. Now some regularity conditions are in order.

C3. For some constant $l>0, \theta_{i, \boldsymbol{c}}, \eta_{i, \boldsymbol{c}}>l$, and $\theta_{i, \boldsymbol{c}}+\eta_{i, \boldsymbol{c}} \leq 1$ for all $i=1,2$ and $\boldsymbol{c}   \in \mathbf{\Psi}$.

C4. $\log (n p) / \sqrt{p^{K-1}} \rightarrow 0$, and $\frac{\sqrt{qs_{\max }}}{\delta}(\sqrt{\frac{\log np}{p^{K-1}n}} + \frac{1}{n}) \rightarrow 0$.

C5. $\frac{\Delta_F^2}{\log (n p) / n+\sqrt{\log (n p) /\left(n p^{K-1}\right)}} \rightarrow \infty$.

Condition C3 is similar to C1. The condition $\log (n p) / \sqrt{p^{K-1}} \rightarrow 0$ in C4 controls the misclassification rate of the $k$-means algorithm. Note that the length of the time interval for searching for the change point is of order $O(n)$; the $\log (n)$ term here in some sense reflects the effect of the difficulty in detecting the true change point when the searching interval is extended as $n$ increases. The second condition in C4 is similar to (7), which ensures that the true communities can be recovered. Condition C5 requires that the average signal strength $\Delta_F^2=p^{-(K-1)}\left[\left\|\mathbf{A}_{1,1}-\mathbf{A}_{2,1}\right\|_F^2+\left\|\mathbf{A}_{1,2}-\mathbf{A}_{2,2}\right\|_F^2\right]$ is of higher order than $\frac{\log (n p)}{n}+\sqrt{\frac{\log (n p)}{n p^{K-1}}}$ for change point detection.

\begin{theorem}
     Let conditions C2-C5 hold. Then the following results hold.
     
(i) When $\boldsymbol{\psi}^{1, \tau_0} \equiv \boldsymbol{\psi}^{\tau_0+1, n}$,
$$
\frac{\left|\tau_0-\widehat{\tau}\right|}{n}=O_p\left(\frac{\frac{\log (n p)}{n}+\sqrt{\frac{\log (n p)}{n p^{K-1}}}}{\Delta_F^2} \times \min \left\{1, \frac{\min \left\{1,\left(n^{-1} p^{K-1} \log (n p)\right)^{\frac{1}{4}}\right\}}{\Delta_F s_{\min }}\right\}\right) .
$$
(ii) When $\boldsymbol{\psi}^{1, \tau_0} \neq \boldsymbol{\psi}^{\tau_0+1, n}$,
$$
\frac{\left|\tau_0-\widehat{\tau}\right|}{n}=O_p\left(\frac{\frac{\log (n p)}{n}+\sqrt{\frac{\log (n p)}{n p^{K-1}}}}{\Delta_F^2} \times \min \left\{1, \frac{\min \left\{1,\left(n^{-1} p^{K-1} \log (n p)\right)^{\frac{1}{4}}\right\}}{\Delta_F s_{\min }}+\frac{1}{\Delta_F^2}\right\}\right) .
$$
\end{theorem}

Notice that for $\tau<\tau_0$, the observations in the time interval $[\tau+1, n]$ are a mixture of the two different network processes if $\psi^{1, \tau_0} \neq \psi^{\tau_0+1, n}$. In the worst case scenario then, all $q$ communities can be changed after the change point $\tau_0$. This causes the extra estimation error term $\frac{1}{\Delta_F^2}$ in Theorem 5(ii).

\section{Simulation}
\subsection{AR(1) Hypergraph}
We generate data according to model AR(1) hypergraph in which $|\xi| = 2,3$ and the parameters $\alpha_{\xi}$ and $\beta_{\xi}$ are drawn independently from $U[0.1,0.5]$. The initial value $\mathbf{X}_0$ was simulated with $\pi_{\xi}=0.5$. We calculate the estimates according to (3). For each setting (with different $p$ and $n$), we replicate the experiment 500 times. Furthermore we also calculate the $95 \%$ confidence intervals for $\alpha_{\xi}$ and $\beta_{\xi}$ based on the asymptotically normal distributions specified in Proposition 5, and report the relative frequencies of the intervals covering the true values of the parameters. There are a few cases with denominators being exactly zero when evaluating the asymptotic variance. In such cases we follow a traditional approach by adding a small number $n^{-1} \times 10^{-4}$ to the denominator. This small value is negligible when the denominators are non-zero. The results are summarized in Table \ref{tab:sim1}.

\begin{table}[htbp]
    \centering
    \caption{The mean squared errors (MSE) of the estimated parameters in $\operatorname{AR}(1)$ 3-uniform hypergraph model and the relative frequencies (coverage rates) of the event that the asymptotic $95 \%$ confidence intervals cover the true values in a simulation with 500 replications.}
\begin{tabular}{|c|c|c|c|c|c|}
\hline & & \multicolumn{2}{|c|}{$\widehat{\alpha}_{\xi}$} & \multicolumn{2}{|c|}{$\widehat{\beta}_{\xi}$} \\
\hline n & p & MSE & Coverage (\%) & MSE & Coverage (\%) \\
\hline 4 & 100 & 0.145 & 22.35 & 0.145 & 22.35 \\
4 & 200 & 0.145 & 22.35 & 0.145 & 22.35 \\
20 & 100 & 0.037& 86.08 & 0.037 & 86.08\\
20 & 200 & 0.037 & 86.07 & 0.037 & 86.07 \\
50 & 100 & 0.012 & 92.27 & 0.012 & 92.26 \\
50 & 200 & 0.012&  92.26&0.012 & 92.26\\
100 & 100 & 0.005& 93.75 &0.005 & 93.75\\
100 & 200 & 0.005& 93.75 & 0.005& 93.75\\
200 & 100 & 0.002 & 94.40 & 0.002 & 94.40\\
200 & 200 & 0.002 & 94.40 & 0.002 & 94.40 \\
\hline
\end{tabular}
\label{tab:sim1}
\end{table}
The MSE decreases as $n$ increases, showing steadily improvement in performance. The coverage rates of the asymptotic confidence intervals are very close to the nominal level when $n \geq 50$. The results hardly change between $p=100$ and 200.

\subsection{Hypergraph SBM}

We now consider Hypergraph Stochastic blockmodel with $q=2$ or 3 clusters, in which $|\xi| = 2 \text{ or } 3, \theta_{i, i}=\theta_{i, i, i}=0.6, \eta_{i, i}=\eta_{i, i, i}=0.4$ for $i=1, \cdots, q$, and $\theta_{\xi}$ and $\eta_{\xi}$ otherwise, are drawn independently from $U[0.05,0.25]$ and $U[0.75,0.95]$ to meet the phenomenon of the community structure. The initial value $\mathbf{X}_0$ was simulated with $\pi_{\xi}=0.5$. For each setting, we replicate the experiment 500 times.

We identify the $q$ latent communities using the newly proposed spectral clustering algorithm based on matrix $\widehat{\mathbf{L}}=\widehat{\mathbf{L}}_1+\widehat{\mathbf{L}}_2$. For the comparison purpose, we also implement the standard spectral clustering method for static hypergraph by \cite{ghoshdastidar_consistency_2017} but using the average

$$
\overline{\mathbf{X}}=\frac{1}{n} \sum_{t=1}^n \mathbf{X}_t
$$

in place of the single observed adjacency matrix. This idea has been frequently used in spectral clustering for dynamic networks.  We report the normalized mutual information (NMI) and the adjusted Rand index (ARI): Both metrics take values between 0 and 1, and both measure the closeness between the true communities and the estimated communities in the sense that the larger the values of NMI and ARI are, the closer the two sets of communities are; see \cite{JMLR:v11:vinh10a}. The results are summarized in Table \ref{tab:sim2}. When $n$ is small, the algorithm based on $\overline{\mathbf{X}}$ has better performance than the newly proposed algorithm based on $\widehat{\mathbf{L}}$, while for larger $n$, the newly proposed algorithm always outperforms the algorithm based on $\overline{\mathbf{X}}$. The differences between the two methods are substantial in terms of the scores of both NMI and ARI. This is due to the fact that the standard method uses only the information on $\pi_{\xi}=\frac{\alpha_{\xi}}{\alpha_{\xi}+\beta_{\xi}}$. It could shows patterns when $n$ is small, but it fails to take the advantage of the $\operatorname{AR}(1)$ structure in which the information on both $\alpha_{\xi}$ and $\beta_{\xi}$ is available as $n$ increases. We cannot include a comparison with the AR(1) network obtained by transforming hyperedges into cliques, since we consider all possible 3-hyperedges, causing the resulting AR(1) projection network to become extremely dense. Instead, the higher NMI and ARI scores, compared to the AR(1) network scenario, can be attributed to the additional information contained in the large number of edges.

After the communities were identified, we estimate $\theta_{\boldsymbol{c}}$ and $\eta_{\boldsymbol{c}}$ in Section 3.3, respectively. The mean squared errors (MSE) are evaluated for all the parameters. The results are summarized in Table \ref{tab:sim2}. For the comparison purpose, we also report the estimates based on the identified communities by the $\overline{\mathbf{X}}$-based clustering. The MSE values of the estimates based on the communities identified by the new clustering method are always smaller than those of based on $\overline{\mathbf{X}}$ for $n$ larger than 4. 
\begin{table}[htbp]
    \centering
    \caption{Normalized mutual information (NMI) and adjusted Rand index (ARI) of the true communities and the estimated communities and the mean squared errors (MSE) of the estimated parameters in $\mathrm{AR}(1)$ hypergraph stochastic block models with $q$ communities  in the simulation with 500 replications.. The communities are estimated by the spectral clustering algorithm (SCA) based on either matrix $\widehat{\mathbf{L}}$ or matrix $\overline{\mathbf{X}}$.}
\begin{tabular}{|c|c|c|c|c|c|c|c|c|c|c|}
\hline & & & \multicolumn{4}{|c|}{SCA based on $\widehat{\mathbf{L}}$} & \multicolumn{4}{|c|}{SCA based on $\overline{\mathbf{X}}$} \\
\hline q & p & n  & ARI & NMI & MSE($\widehat{\theta}_{\boldsymbol{c}}$) & MSE($\widehat{\eta}_{\boldsymbol{c}}) $ & ARI & NMI & MSE($\widehat{\theta}_{\boldsymbol{c}}$) & MSE($\widehat{\eta}_{\boldsymbol{c}} $) \\

\hline \multirow[t]{4}{*}{6} & \multirow[t]{4}{*}{80} & 4 &0.835 &0.879& 0.0105 & 0.0098&0.918&0.939& 0.0083 & 0.0079\\
 & & 10 & 0.995&0.996& 0.0057 & 0.0057& 0.960& 0.971& 0.0071 & 0.067\\
& & 40 & 0.987&0.990& 0.0060 & 0.0058& 0.965& 0.976& 0.0068 & 0.0066\\
& & 100 & 0.985&0.989& 0.0061& 0.0059& 0.964& 0.976& 0.0070 & 0.066\\
\hline \multirow[t]{4}{*}{6} & \multirow[t]{4}{*}{120} & 4 &0.931&0.948& 0.0076 & 0.0069& 0.968& 0.976& 0.0066 & 0.0062\\
 & & 10 & 1.000&1.000& 0.0054 & 0.0054& 0.977& 0.983& 0.0063 & 0.0060\\
& & 40 & 0.997&0.998& 0.0056 & 0.0055& 0.977& 0.984& 0.0064 & 0.0060\\
& & 100 & 0.994&0.995& 0.0056 & 0.0056& 0.976& 0.984& 0.0063 & 0.0061\\
 \hline \multirow[t]{4}{*}{10} & \multirow[t]{4}{*}{80} & 4 &0.429&0.639& 0.0122 & 0.0130& 0.572& 0.733& 0.0114 & 0.0116\\
 & & 10 & 0.843&0.902& 0.0093 & 0.0090& 0.814& 0.888& 0.0091 & 0.0091\\
& & 40 & 0.909&0.945& 0.0081 & 0.0078& 0.899& 0.939& 0.0080 & 0.0079\\
& & 100 & 0.931&0.958& 0.0076 & 0.0074& 0.899& 0.940& 0.0081 & 0.0080\\
\hline \multirow[t]{4}{*}{10} & \multirow[t]{4}{*}{120} & 4 &0.652&0.774& 0.0097 & 0.0093& 0.774& 0.857& 0.0089 & 0.0085\\
 & & 10 & 0.961&0.974& 0.0071 & 0.0068& 0.906& 0.944& 0.0076 & 0.0073\\
& & 40 & 0.964&0.978& 0.0068 & 0.0066& 0.936& 0.962& 0.0072 & 0.0070\\
& & 100 & 0.960&0.976& 0.0069 & 0.0067& 0.924& 0.957& 0.0075 & 0.0072\\
\hline
\end{tabular}
\label{tab:sim2}
\end{table}

\section{Real-world Data Application}
\subsection{Primary School Contact Data}

Now, we consider a contact network dataset collected in a primary school in Lyon, France \cite{stehle2011high}. The data consist of recorded face-to-face interactions among 232 students and 10 teachers across 10 classes (one teacher per class), measured over n=4 half-days in October 2009 using the SocioPatterns RFID infrastructure. The school consists of five grade levels, with each grade divided into two classes (e.g., 1A, 1B, 2A, 2B, etc.). Contact events were recorded at 20-second intervals, capturing a total of 77,602 interactions, with each student having, on average, 323 contacts per day and interacting with 47 peers. The data are publicly available at {\tt www.sociopatterns.org} and have been widely used for social network analysis and community detection.  

In our analysis, we focus on the higher-order interaction structure formed by these interactions to identify latent communities using clustering techniques. The original dataset includes only pairwise interactions. These interactions are recorded when two students face each other within close proximity (1 to 1.5 meters) during a short time window of 20 seconds. Therefore, if three interactions—A-B, B-C, and A-C—occur at the same timestamp, it is highly likely that students A, B, and C were engaged in a group interaction, given the tight spatial and temporal constraints. It is less plausible that these represent independent pairwise interactions occurring simultaneously within such a restricted context. Therefore,  to construct a hypergraph structure, if a set of student contacts forms a $k$-clique at a given timestamp, we represent these contacts as $k$-hyperedges, while 2-hyperedges are removed. The distribution of $k$-hyperedges for different values of $k$ is shown in Table \ref{tab:pri_size}. As the proportions of 4-hyperedges and 5-hyperedges are small, we consider only 2- and 3-hyperedges in our analysis. Meanwhile, we compare our results with AR(1) network. 

\begin{table}[htbp]
    \centering
    \caption{Distribution for size of hyperedges in primary school contact data.}
\begin{tabular}{|c|c|c|c|c|}
\hline & $k$=2 & $k$=3 & $k$=4 & $k$=5   \\
\hline Proportion & 90.88\%&8.67\%&0.44\%&0.01\%\\
\hline
\end{tabular}
\label{tab:pri_size}
\end{table}
\begin{table}[htbp]
    \centering
    \caption{Community distribution for primary school contact data with $q=5$.}
\begin{tabular}{|c|c|c|c|c|c|}
\hline Class & Cluster 1 & Cluster 2& Cluster 3 & Cluster 4& Cluster 5  \\
\hline  \multicolumn{6}{|c|}{ Hypergraph} \\
\hline 1A & 22 & 0 & 1 & 0&0  \\
1B & 25 &0 & 0 & 0& 0  \\
2A & 0 & 23& 0 & 0& 0 \\
2B & 0 & 26& 0 & 0& 0 \\
3A & 0 & 0& 23 & 0& 0  \\
3B & 0 & 0& 22 & 0& 0   \\
4A & 21 & 0 & 0 & 0& 0 \\
4B & 1 & 0& 0 & 0& 22 \\
5A & 22 & 0 & 0 & 0& 0  \\
5B & 0 & 0 & 0 & 24& 0 \\
\hline  \multicolumn{6}{|c|}{Network} \\
\hline 1A & 22 & 0 & 1 & 0&0  \\
1B & 25 &0 & 0 & 0& 0  \\
2A & 23 & 0& 0 & 0& 0 \\
2B & 0 & 26& 0 & 0& 0 \\
3A & 0 & 0& 23 & 0& 0  \\
3B & 0 & 0& 22 & 0& 0   \\
4A & 21 & 0 & 0 & 0& 0 \\
4B & 1 & 0& 0 & 0& 22 \\
5A & 22 & 0 & 0 & 0& 0  \\
5B & 0 & 0 & 0 & 24& 0 \\
\hline
\end{tabular}
\label{tab:pri_comm5}
\end{table}
\begin{table}[htbp]
    \centering
    \caption{Community distribution for primary school contact data with $q=2$.}
\begin{tabular}{|c|c|c|c|c|}
\hline & \multicolumn{2}{|c|}{ Hypergraph} & \multicolumn{2}{|c|}{ Network } \\
\hline Class & Cluster 1 & Cluster 2 & Cluster 1 & Cluster 2\\
\hline 1A & 23 & 0 & 23 & 0\\
1B & 0 & 25 & 6 & 16 \\
2A & 23 & 0 & 23 & 0\\
2B & 26 & 0& 26 & 0 \\
3A & 23 & 0& 23 & 0 \\
3B & 22 & 0& 22 & 0  \\
4A & 0 & 21& 0 & 21 \\
4B & 0 & 23& 0 & 23\\
5A & 1 & 21 & 1 & 21 \\
5B & 0 & 24 & 0 & 24\\
\hline
\end{tabular}
\label{tab:pri_comm2}
\end{table}

We validate the suitability of the stationary AR(1) random block model for this dataset using permutation test, with $p$-value 0.344.  First, we set \( q=2 \), dividing all individuals into two communities. The results in Table \ref{tab:pri_comm2} indicate that clusters are formed such that individuals from 1A, 2A, 2B, 3A, and 3B are grouped together, while the remaining individuals form another cluster. Only one student is misclassified. The specific distribution details are presented in the table below.  The phenomenon where 1A and 1B are divided into different groups may be explained by Table 2 and Table 3 in the dataset description by Stehlé et al. These tables indicate that the total number of contacts between 1A and individuals from 4A, 4B, 5A, and 5B is relatively small, whereas 1B has significantly more interactions with these groups.  The community structure obtained using the AR(1) network has 10 misclassified students, one from 5A and the remaining from 1B.

To determine the optimal \( q \), we use the Bayesian Information Criterion (BIC), defined as:  
\[
\operatorname{BIC}(q)=-2 \max \log (\text{likelihood}) + \log \left\{n(p/q)^2+ n(p/q)^3\right\} (q(q+1)+q(q+1)(q+2)/3)
\]
For each fixed \( q \) and \(K =3\) in our analysis, we effectively construct \( q(q+1)/2 + q(q+1)(q+2)/6 \) independent models, where the first part represent the model for 2-hyperedges with two parameters, \( \theta_{k, \ell} \) and \( \eta_{k, \ell} \), for \( 1 \leq k \leq \ell \leq q \) and the second part represent the model for 3-hyperedges with three parameters, \( \theta_{k, \ell, m} \) and \( \eta_{k, \ell, m} \), for \( 1 \leq k \leq \ell \leq m \leq q \). The number of available observations for each model is approximately \( n(p/q)^2 + n(p/q)^3\), under the assumption that the node counts across all \( q \) clusters are approximately equal, i.e., \( p/q \). Consequently, the penalty term in the BIC formulation is:
\[
 \log \left\{n(p/q)^2+ n(p/q)^3\right\} (q(q+1)+q(q+1)(q+2)/3).
\]

Meanwhile, we can consider Akaike Information Criterion(AIC), defined as 
\[
\operatorname{AIC}(q)=-2 \max \log (\text{likelihood}) +q(q+1)+q(q+1)(q+2)/3. 
\]
BIC and AIC for different \(K\) can be construct with similar way. 
\begin{figure}[t!]
    \centering
    \subfigure[
    BIC($q$)]{\includegraphics[width=0.48\textwidth]{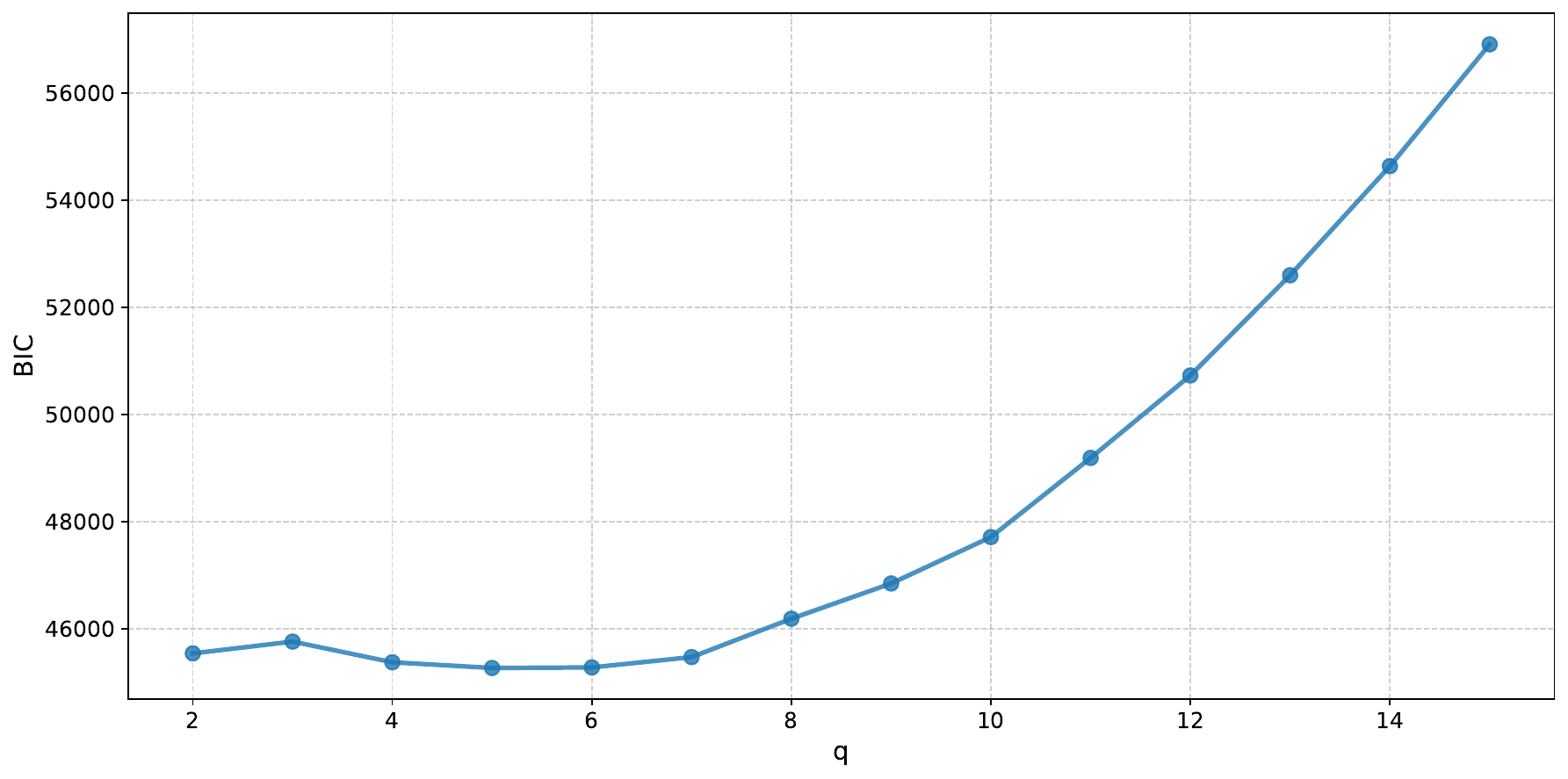}\label{tab:pri_bic}}
    \subfigure[AIC($q$)]{\includegraphics[width=0.48\textwidth]{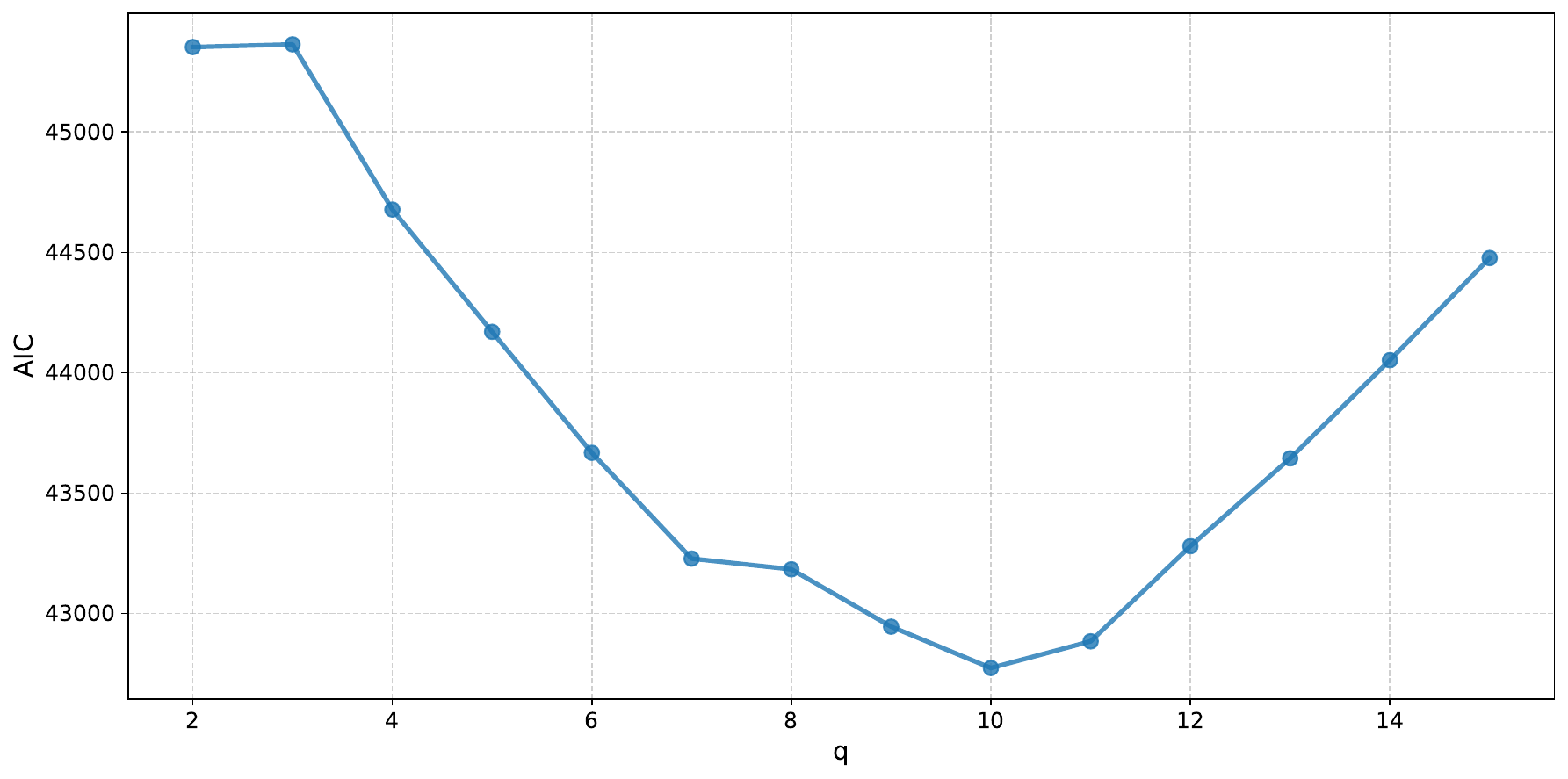}\label{tab:pri_aic}}
    \caption{Optimal $q$ selection criteria corresponding to different number of communities for primary school contact data.}
    \label{fig:pri}
\end{figure}

The minimum BIC value is achieved at \( q=5 \) shown in \ref{tab:pri_bic}, at which point the clustering results are as follows shown in Table \ref{tab:pri_comm5}: 1A, 1B, 4A, and 5A are grouped together; 2A and 2B form another group; 3A and 3B constitute a separate group; while 4B and 5B are each assigned to their own respective groups.  For the results with $q = 5$ from AR(1) network, the primary difference is that 2A is moved from Cluster 2 to Cluster 1. 

Next, the minimum AIC value is achieved at \( q=10 \) shown in \ref{tab:pri_bic}, corresponding to the actual number of classes in the dataset. The results show that only two students were misclassified: one student from 4A was assigned to 4B, and another from 5A was assigned to 2B. Notably, the misclassified individual in the \( q=2 \) case is also included in the \( q=10 \) case. Furthermore, all ten teachers were correctly assigned to their respective classes.   The results from the AR(1) network show that one student in 5A and four students in 3B were reassigned to 3A. We also used spectral clustering algorithm based on $\overline{\mathbf{X}}$. It perfectly recovers every student’s class, which mirroring our simulation findings.
\subsection{Enron Email Data}

The Enron email dataset comprises approximately 500,000 emails exchanged among employees of the Enron Corporation. This dataset was originally collected by the Federal Energy Regulatory Commission during its investigation into Enron’s collapse and was subsequently made publicly available at (https://www.cs.cmu.edu/~./enron/). We utilized data covering all email exchanges among Enron employees from January 2000 to March 2002 (\( n=27 \)) with 143 nodes.  

To construct the hypergraph structure, we consider each email as forming a hyperedge, where the sender and all recipients collectively constitute the hyperedge. In total, there are 10,885 hyperedges, with their distribution provided in Table 4. To simplify our analysis, we transform all hyperedges of size \( k \geq 4 \) by decomposing each into \( {k \choose 3} \) hyperedges of size 3 so that the maximum size of hyperedges in our analysis is $K=3$.  

\begin{table}[htbp]
    \centering
    \caption{Distribution for size of hyperedges in Enron email data.}
\begin{tabular}{|c|c|c|c|c|c|c|c|c|c|}
\hline & $k$=1& $k$=2 & $k$=3 & $k$=4 & $k$=5 & $k$=6 & $k$=7 & $k$=8 & $k$=9   \\
\hline Cluster Size & 431 & 7940& 1231 & 567 & 364 &  91 & 123 &  50 &  25 \\
\hline & $k$=10& $k$=11 & $k$=12 & $k$=12 & $k$=15 & $k$=16 & $k$=18 & $k$=36 & $k$=37   \\
\hline Cluster Size & 12 &  17 &  24  &  3 &   1 &   2  &  2 &1&1 \\
\hline
\end{tabular}
\label{tab:Enron_size}
\end{table}

\begin{table}[htbp]
    \centering
    \caption{Community distribution for Enron email data with $q=7$.}
\begin{tabular}{|c|c|c|c|c|c|c|c|}
\hline Cluster &   1&  2 & 3& 4 &  5 & 6 &  7  \\
\hline Cluster Size & 57 & 27 & 20 & 11 & 11 &  9 & 8 \\
\hline
\end{tabular}
\label{tab:Enron_comm_BIC}
\end{table}
\begin{table}[htbp]
    \centering
    \caption{Community distribution for Enron email data with $q=11$.}
\begin{tabular}{|c|c|c|c|c|c|c|c|c|c|c|c|}
\hline Cluster  & 1&  2 &  3&  4 & 5 & 6 & 7  & 8 &  9 & 10 & 11 \\
\hline Cluster Size & 35 & 20 & 15 & 13 & 9 &  9 & 9 &9 & 9 &  8 & 7\\
\hline
\end{tabular}
\label{tab:Enron_comm_AIC}
\end{table}
We first determine the optimal \( q \) using the dataset with full timesteps, yielding \( q=7 \) according to BIC, and \( q=11 \) according to AIC , with the corresponding distribution presented in Table \ref{tab:Enron_comm_BIC} and Table \ref{tab:Enron_comm_AIC}. However, the $p$-value of the permutation test for the residuals from the fitted model is 0, indicating strong statistical evidence that stationarity does not hold over the entire period. Applying our proposed change-point detection method with \( q=7 \), we identify a change point in August 2001, four month earlier than Enron's disclosure due to  bankruptcy.  If we change the the maximum size of hyperedges to $K=4$ we have the same result.  
This can be explained by the fact that the accounting scandal was first raised within the company in this month, and financial problems were the main reason for the company's bankruptcy four months later.

\begin{figure}[!htb]
    \centering
    \includegraphics[scale = 0.5]{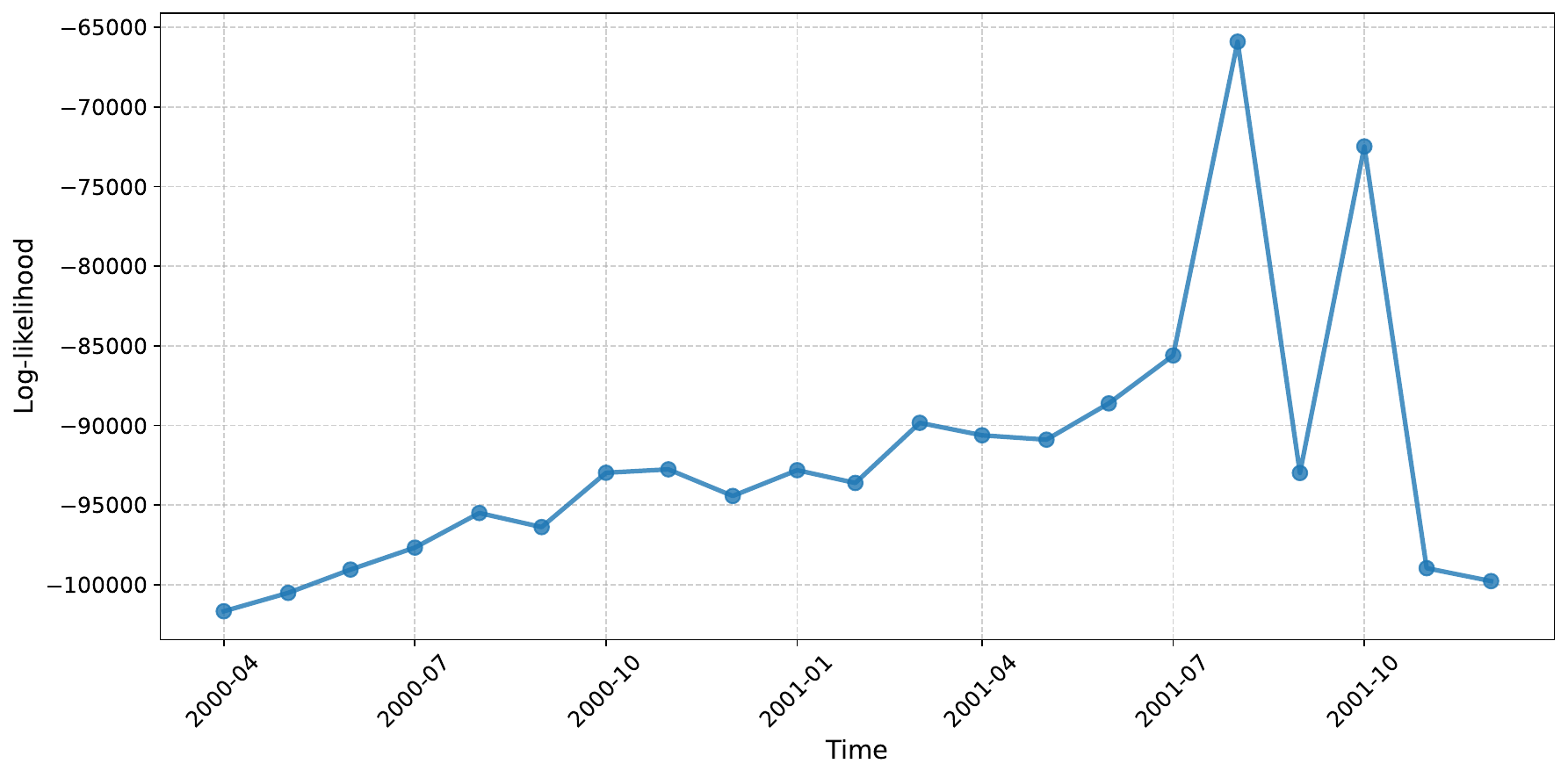}
     \caption{Log-likelihood functions corresponding to different change points in time for Enron email data.}
\end{figure}

\section*{Acknowledgement}
We gratefully acknowledge the support from EPSRC NeST Programme
grant EP/X002195/1 and EPSRC DASS Programme grant EP/Z531327/1. 
\bibliographystyle{plainnat}
\bibliography{main}

\begin{thebibliography}{38}
\providecommand{\natexlab}[1]{#1}
\providecommand{\url}[1]{\texttt{#1}}
\expandafter\ifx\csname urlstyle\endcsname\relax
  \providecommand{\doi}[1]{doi: #1}\else
  \providecommand{\doi}{doi: \begingroup \urlstyle{rm}\Url}\fi

\bibitem[Benson et~al.(2018)Benson, Abebe, Schaub, Jadbabaie, and Kleinberg]{benson_simplicial_2018}
Austin~R. Benson, Rediet Abebe, Michael~T. Schaub, Ali Jadbabaie, and Jon Kleinberg.
\newblock Simplicial closure and higher-order link prediction.
\newblock \emph{Proceedings of the National Academy of Sciences}, 115\penalty0 (48):\penalty0 E11221--E11230, November 2018.
\newblock \doi{10.1073/pnas.1800683115}.
\newblock URL \url{https://www.pnas.org/doi/abs/10.1073/pnas.1800683115}.
\newblock Publisher: Proceedings of the National Academy of Sciences.

\bibitem[Benson et~al.(2022)Benson, Veldt, and Gleich]{benson2022fauci}
Austin~R Benson, Nate Veldt, and David~F Gleich.
\newblock fauci-email: a json digest of anthony fauci’s released emails.
\newblock In \emph{Proceedings of the International AAAI Conference on Web and Social Media}, volume~16, pages 1208--1217, 2022.

\bibitem[Brusa and Matias(2024)]{brusa2024model}
Luca Brusa and Catherine Matias.
\newblock Model-based clustering in simple hypergraphs through a stochastic blockmodel.
\newblock \emph{Scandinavian Journal of Statistics}, 51\penalty0 (4):\penalty0 1661--1684, 2024.

\bibitem[Chang et~al.(2024)Chang, Fang, Kolaczyk, MacDonald, and Yao]{chang2024autoregressive}
Jinyuan Chang, Qin Fang, Eric~D Kolaczyk, Peter~W MacDonald, and Qiwei Yao.
\newblock Autoregressive networks with dependent edges.
\newblock \emph{arXiv preprint arXiv:2404.15654}, 2024.

\bibitem[Chen et~al.(2023)Chen, Fan, and Zhu]{chen2023community}
Elynn~Y Chen, Jianqing Fan, and Xuening Zhu.
\newblock Community network auto-regression for high-dimensional time series.
\newblock \emph{Journal of Econometrics}, 235\penalty0 (2):\penalty0 1239--1256, 2023.

\bibitem[Chen et~al.(2025)Chen, Arroyo, Athreya, Cape, Vogelstein, Park, White, Larson, Yang, and Priebe]{chen_multiple_2025}
Guodong Chen, Jes{\'u}s Arroyo, Avanti Athreya, Joshua Cape, Joshua~T Vogelstein, Youngser Park, Chris White, Jonathan Larson, Weiwei Yang, and Carey~E Priebe.
\newblock Multiple network embedding for anomaly detection in time series of graphs.
\newblock \emph{Computational Statistics \& Data Analysis}, 203:\penalty0 108070, 2025.

\bibitem[Corneck et~al.(2025)Corneck, Cohen, Martin, and Sanna~Passino]{corneck2025online}
Joshua Corneck, Edward~AK Cohen, James~S Martin, and Francesco Sanna~Passino.
\newblock Online bayesian changepoint detection for network poisson processes with community structure.
\newblock \emph{Statistics and Computing}, 35\penalty0 (3):\penalty0 1--29, 2025.

\bibitem[Ghoshdastidar and Dukkipati(2015)]{ghoshdastidar_provable_2015}
Debarghya Ghoshdastidar and Ambedkar Dukkipati.
\newblock A provable generalized tensor spectral method for uniform hypergraph partitioning.
\newblock In \emph{International Conference on Machine Learning}, pages 400--409. PMLR, 2015.

\bibitem[Ghoshdastidar and Dukkipati(2017{\natexlab{a}})]{JMLR:v18:16-100}
Debarghya Ghoshdastidar and Ambedkar Dukkipati.
\newblock Uniform hypergraph partitioning: Provable tensor methods and sampling techniques.
\newblock \emph{Journal of Machine Learning Research}, 18\penalty0 (50):\penalty0 1--41, 2017{\natexlab{a}}.

\bibitem[Ghoshdastidar and Dukkipati(2017{\natexlab{b}})]{ghoshdastidar_consistency_2017}
Debarghya Ghoshdastidar and Ambedkar Dukkipati.
\newblock {Consistency of spectral hypergraph partitioning under planted partition model}.
\newblock \emph{The Annals of Statistics}, 45\penalty0 (1):\penalty0 289 -- 315, 2017{\natexlab{b}}.

\bibitem[Hallgren et~al.(2024)Hallgren, Heard, and Turcotte]{hallgren2024changepoint}
Karl~L Hallgren, Nicholas~A Heard, and Melissa~JM Turcotte.
\newblock Changepoint detection on a graph of time series.
\newblock \emph{Bayesian Analysis}, 19\penalty0 (2):\penalty0 649--676, 2024.

\bibitem[Hung et~al.(2024)Hung, Mantziou, and Reinert]{hung2024bayesian}
Elly Hung, Anastasia Mantziou, and Gesine Reinert.
\newblock A bayesian mixture model for poisson network autoregression.
\newblock \emph{arXiv preprint arXiv:2411.14265}, 2024.

\bibitem[Jiang et~al.(2023{\natexlab{a}})Jiang, Leng, Yan, Yao, and Yu]{jiang2023two}
Binyan Jiang, Chenlei Leng, Ting Yan, Qiwei Yao, and Xinyang Yu.
\newblock A two-way heterogeneity model for dynamic networks.
\newblock \emph{arXiv preprint arXiv:2305.12643}, 2023{\natexlab{a}}.

\bibitem[Jiang et~al.(2023{\natexlab{b}})Jiang, Li, and Yao]{jiang_autoregressive_2023}
Binyan Jiang, Jialiang Li, and Qiwei Yao.
\newblock Autoregressive networks.
\newblock \emph{Journal of Machine Learning Research}, 24\penalty0 (227):\penalty0 1--69, 2023{\natexlab{b}}.

\bibitem[Josephs and Upton(2024)]{josephs_hypergraph_2024}
Nathaniel Josephs and Elizabeth Upton.
\newblock Hypergraph adjusted plus-minus.
\newblock \emph{arXiv preprint arXiv:2403.20214}, 2024.

\bibitem[Ke et~al.(2019)Ke, Shi, and Xia]{ke_community_2020}
Zheng~Tracy Ke, Feng Shi, and Dong Xia.
\newblock Community detection for hypergraph networks via regularized tensor power iteration.
\newblock \emph{arXiv preprint arXiv:1909.06503}, 2019.

\bibitem[Khabou et~al.(2025)Khabou, Cohen, and Veraart]{khabou2025markov}
Mahmoud Khabou, Edward~AK Cohen, and Almut~ED Veraart.
\newblock The markov approximation of the periodic multivariate poisson autoregression.
\newblock \emph{arXiv preprint arXiv:2504.02649}, 2025.

\bibitem[Klamt et~al.(2009)Klamt, Haus, and Theis]{klamt2009hypergraphs}
Steffen Klamt, Utz-Uwe Haus, and Fabian Theis.
\newblock Hypergraphs and cellular networks.
\newblock \emph{PLoS computational biology}, 5\penalty0 (5):\penalty0 e1000385, 2009.

\bibitem[Klimm et~al.(2021)Klimm, Deane, and Reinert]{klimm_hypergraphs_2021}
Florian Klimm, Charlotte~M Deane, and Gesine Reinert.
\newblock Hypergraphs for predicting essential genes using multiprotein complex data.
\newblock \emph{Journal of Complex Networks}, 9\penalty0 (2):\penalty0 cnaa028, April 2021.
\newblock ISSN 2051-1329.
\newblock \doi{10.1093/comnet/cnaa028}.
\newblock URL \url{https://doi.org/10.1093/comnet/cnaa028}.

\bibitem[Knight et~al.(2016)Knight, Nunes, and Nason]{knight2016modelling}
Marina~Iuliana Knight, MA~Nunes, and GP~Nason.
\newblock Modelling, detrending and decorrelation of network time series.
\newblock \emph{arXiv preprint arXiv:1603.03221}, 2016.

\bibitem[Liu and Nason(2023)]{liu2023new}
Hengxu Liu and Guy Nason.
\newblock New methods for network count time series.
\newblock \emph{arXiv preprint arXiv:2312.01944}, 2023.

\bibitem[Lunag{\'o}mez et~al.(2017)Lunag{\'o}mez, Mukherjee, Wolpert, and Airoldi]{lunagomez2017geometric}
Sim{\'o}n Lunag{\'o}mez, Sayan Mukherjee, Robert~L Wolpert, and Edoardo~M Airoldi.
\newblock Geometric representations of random hypergraphs.
\newblock \emph{Journal of the American Statistical Association}, 112\penalty0 (517):\penalty0 363--383, 2017.

\bibitem[Lyu et~al.(2023)Lyu, Xia, and Zhang]{lyu_latent_2023}
Zhongyuan Lyu, Dong Xia, and Yuan Zhang.
\newblock Latent space model for higher-order networks and generalized tensor decomposition.
\newblock \emph{Journal of Computational and Graphical Statistics}, 32\penalty0 (4):\penalty0 1320--1336, 2023.

\bibitem[Mantziou et~al.(2023)Mantziou, Cucuringu, Meirinhos, and Reinert]{mantziou2023gnar}
Anastasia Mantziou, Mihai Cucuringu, Victor Meirinhos, and Gesine Reinert.
\newblock The gnar-edge model: a network autoregressive model for networks with time-varying edge weights.
\newblock \emph{Journal of Complex Networks}, 11\penalty0 (6):\penalty0 cnad039, 2023.

\bibitem[Merlev{\`e}de et~al.(2009)Merlev{\`e}de, Peligrad, and Rio]{merlevede2009bernstein}
Florence Merlev{\`e}de, Magda Peligrad, and Emmanuel Rio.
\newblock Bernstein inequality and moderate deviations under strong mixing conditions.
\newblock In \emph{High dimensional probability V: the Luminy volume}, volume~5, pages 273--293. Institute of Mathematical Statistics, 2009.

\bibitem[Padilla et~al.(2022)Padilla, Yu, and Priebe]{padilla2022change}
Oscar Hernan~Madrid Padilla, Yi~Yu, and Carey~E Priebe.
\newblock Change point localization in dependent dynamic nonparametric random dot product graphs.
\newblock \emph{Journal of Machine Learning Research}, 23\penalty0 (234):\penalty0 1--59, 2022.

\bibitem[Passino and Heard(2023)]{passino2023mutually}
Francesco~Sanna Passino and Nicholas~A Heard.
\newblock Mutually exciting point process graphs for modeling dynamic networks.
\newblock \emph{Journal of Computational and Graphical Statistics}, 32\penalty0 (1):\penalty0 116--130, 2023.

\bibitem[Stehl{\'e} et~al.(2011)Stehl{\'e}, Voirin, Barrat, Cattuto, Isella, Pinton, Quaggiotto, Van~den Broeck, R{\'e}gis, Lina, et~al.]{stehle2011high}
Juliette Stehl{\'e}, Nicolas Voirin, Alain Barrat, Ciro Cattuto, Lorenzo Isella, Jean-Fran{\c{c}}ois Pinton, Marco Quaggiotto, Wouter Van~den Broeck, Corinne R{\'e}gis, Bruno Lina, et~al.
\newblock High-resolution measurements of face-to-face contact patterns in a primary school.
\newblock \emph{PloS one}, 6\penalty0 (8):\penalty0 e23176, 2011.

\bibitem[Stewart and Sun(1990)]{zbMATH00047363}
Gilbert~W. Stewart and Ji-guang Sun.
\newblock \emph{Matrix perturbation theory}.
\newblock Boston etc.: Academic Press, Inc., 1990.
\newblock ISBN 0-12-670230-6.

\bibitem[Vinh et~al.(2010)Vinh, Epps, and Bailey]{JMLR:v11:vinh10a}
Nguyen~Xuan Vinh, Julien Epps, and James Bailey.
\newblock Information theoretic measures for clusterings comparison: Variants, properties, normalization and correction for chance.
\newblock \emph{Journal of Machine Learning Research}, 11\penalty0 (95):\penalty0 2837--2854, 2010.

\bibitem[Wu et~al.(2024)Wu, Xu, and Zhu]{wu_general_2024}
Shihao Wu, Gongjun Xu, and Ji~Zhu.
\newblock A general latent embedding approach for modeling non-uniform high-dimensional sparse hypergraphs with multiplicity.
\newblock \emph{arXiv preprint arXiv:2410.12108}, 2024.

\bibitem[Xie(2021)]{xie2021distributed}
Zheng Xie.
\newblock A distributed hypergraph model for simulating the evolution of large coauthorship networks.
\newblock \emph{Scientometrics}, 126\penalty0 (6):\penalty0 4609--4638, 2021.

\bibitem[Yi and Lee(2022)]{yi2022structure}
Sudo Yi and Deok-Sun Lee.
\newblock Structure of international trade hypergraphs.
\newblock \emph{Journal of Statistical Mechanics: Theory and Experiment}, 2022\penalty0 (10):\penalty0 103402, 2022.

\bibitem[Yu and Zhu(2025)]{yu_modeling_2025}
Xianshi Yu and Ji~Zhu.
\newblock Modeling hypergraphs with diversity and heterogeneous popularity.
\newblock \emph{Journal of the American Statistical Association}, \penalty0 (just-accepted):\penalty0 1--20, 2025.

\bibitem[Yu et~al.(2015)Yu, Wang, and Samworth]{yu2015useful}
Yi~Yu, Tengyao Wang, and Richard~J Samworth.
\newblock A useful variant of the davis--kahan theorem for statisticians.
\newblock \emph{Biometrika}, 102\penalty0 (2):\penalty0 315--323, 2015.

\bibitem[Yuan and Shang(2022)]{yuan2022statistical}
Mingao Yuan and Zuofeng Shang.
\newblock Statistical limits for testing correlation of hypergraphs.
\newblock \emph{arXiv preprint arXiv:2202.05888}, 2022.

\bibitem[Zhen and Wang(2023)]{zhen2023community}
Yaoming Zhen and Junhui Wang.
\newblock Community detection in general hypergraph via graph embedding.
\newblock \emph{Journal of the American Statistical Association}, 118\penalty0 (543):\penalty0 1620--1629, 2023.

\bibitem[Zhou et~al.(2006)Zhou, Huang, and Sch{\"o}lkopf]{zhou2006learning}
Dengyong Zhou, Jiayuan Huang, and Bernhard Sch{\"o}lkopf.
\newblock Learning with hypergraphs: Clustering, classification, and embedding.
\newblock \emph{Advances in neural information processing systems}, 19, 2006.

\end{thebibliography}

\section*{Appendix}
\subsection*{A.1 Proof of Proposition 1}
To show the stationary of $\{\textbf{X}_t\}$,  we need to show that $\mathbb{P}\left(X_\xi^t=1\right)=\mathbb{P}\left(X_\xi^s=1\right)$ for distinct $s, t \geq 0$. and every $\xi \in \mathcal{E}$. We have
$$
\begin{aligned}
\begin{aligned}
\mathbb{P}\left(X_\xi^t=1\right) & =\mathbb{P}\left(X_\xi^t=1 \mid X_\xi^{t-1}=0\right) \mathbb{P}\left(X_\xi^{t-1}=0\right) +\mathbb{P}\left(X_\xi^t=1 \mid X_\xi^{t-1}=1\right) \mathbb{P}\left(X_\xi^{t-1}=1\right) \\
& =\alpha_\xi\left(1-\mathbb{P}\left(X_\xi^{t-1}=1\right)\right)+\left(1-\beta_\xi\right) \mathbb{P}\left(X_\xi^{t-1} =1\right) \\
& =\alpha_\xi+\left(1-\alpha_\xi-\beta_\xi\right) \mathbb{P}\left(X_\xi^{t-1}=1\right)
\end{aligned}
\end{aligned}
$$
sine $\mathbb{P}\left(X_\xi^0=1\right)=\pi_\xi=\frac{\alpha_\xi}{\alpha_\xi+\beta_\xi}$, setting $t=1$, we have
$$
\mathbb{P}\left(X_\xi^{1}=1\right)=\alpha_\xi+\left(1-\alpha_\xi-\beta_\xi\right) \pi_\xi=\frac{\alpha_\xi}{\alpha_\xi+\beta_\xi}=\pi_\xi
$$

Hence, $\mathbb{P}\left(X_\xi{ }^t=1\right)=\mathbb{P}\left(X_\xi^s=1\right)=\pi_\xi$ for all $s, t \geq 0$.

Then $\mathbb{E}\left(X_\xi^t\right)=\mathbb{P}\left(X_\xi^t=1\right)=\frac{\alpha_\xi}{\alpha_\xi+\beta_\xi}$ and $\operatorname{Var}\left(X_\xi^t\right)=\mathbb{E}\left(\left(X_\xi^{t}\right)^2\right)-\mathbb{E}\left(X_\xi^t\right)^2=\frac{\alpha_\xi \beta_\xi}{\left(\alpha_\xi+\beta_\xi\right)^2}$

For $ \xi \neq \xi'$, as the present of each edges are independence, the covariance should be 0.
For $\xi=\xi^{\prime}$, we have for $t>s$
$$
\begin{aligned}
 \operatorname{Cov}\left(X_\xi^t, X_\xi^s\right)=&\mathbb{E}\left(X_\xi^t X_\xi^s\right)-\mathbb{E}\left(X_\xi^t\right) \mathbb{E}\left(X_\xi^s\right)
=  \mathbb{P}\left(X_\xi^t=X_\xi^s=1\right)-\left(\frac{\alpha_\xi}{\alpha_\xi+\beta_\xi}\right)^2 \\
= & \alpha_\xi\mathbb{P}(X_\xi^{t-1}=0, X^s_\xi = 1) +(1-\beta)\mathbb{P}(X_\xi^{t-1}=1, X^s_\xi = 1)-\left(\frac{\alpha_\xi}{\alpha_\xi+\beta_\xi}\right)^2 \\
= & \alpha_\xi \mathbb{P}(X_\xi^s = 1)-\alpha_\xi \mathbb{E}(X_\xi^{t-1}X_\xi^s)+(1-\beta_\xi) \mathbb{E}(X_\xi^{t-1}X_\xi^s)-\left(\frac{\alpha_\xi}{\alpha_\xi+\beta_\xi}\right)^2  \\
= & (1-\alpha_\xi-\beta_\xi) \left(\mathbb{E}\left(X_\xi^{t-1} X_\xi^s\right)-\mathbb{E}\left(X_\xi^{t-1}\right) \mathbb{E}\left(X_\xi^s\right)\right)
\\ = &(1-\alpha_\xi-\beta_\xi) \operatorname{Cov}\left(X_\xi^{t-1}, X_\xi^s\right)
\end{aligned}
$$

So we have 
$$
\rho_{\xi,\xi'}(|t-s|)=\left(1-\alpha_\xi-\beta_\xi\right)^{|t-s|} \text { for } \xi=\xi'.
$$
\subsection*{A.2 Proof of Proposition 2}
We have
$$
\begin{aligned}
 d_H(k)=&\mathbb{E}\left\{D_H\left(\textbf{X}_{s+k}, \textbf{X}_s\right)\right\}=\sum_{\xi \in \mathcal{E}} \mathbb{P}\left(X_\xi^{s+k} \neq X_\xi^s\right) \\=&\sum_{\xi \in \mathcal{E}} \mathbb{P}\left(X_\xi^{s+k}=0, X_\xi^s=1\right)+\mathbb{P}\left(X_\xi^{s+k}=1, X_\xi^s=0\right) .
\end{aligned}
$$
For the first term,
$$
\begin{aligned}
 \mathbb{P}\left(X_\xi^{s+k}=0, X_\xi^s=1\right)&=\mathbb{P}\left(X_\xi^{s+k-1}=0, X_\xi^s=1\right) \mathbb{P}\left(X_\xi^{s+k}=0 \mid X_\xi^{s+k-1}=0\right)\\&+\mathbb{P}\left(X_\xi^{s+k-1}=1, X_\xi^s=1\right) \mathbb{P}\left(X_\xi^{s+k}=0 \mid X_\xi^{s+k-1}=1\right) \\
& =\left(1-\alpha_\xi\right) \mathbb{P}\left(X_\xi^{s+k-1}=0, X_\xi^s=1\right)\\&+\beta_\xi\left(\mathbb{P}\left(X_\xi^s=1\right)-\mathbb{P}\left(X_\xi^{s+k-1}=0,\right. \left.X_\xi^s=1\right)\right) 
\\&=\left(1-\alpha_\xi-\beta_\xi\right) \mathbb{P}\left(X_\xi^{s+k-1}=0, X_\xi^s=1\right)+\frac{\alpha_\xi \beta_\xi}{\alpha_\xi+\beta_\xi} .
\end{aligned}
$$

Similarly, for the second term,
$$
\begin{aligned}
 \mathbb{P}\left(X_\xi^{s+k}=1, X_\xi^s=0\right)=\left(1-\alpha_\xi-\beta_\xi\right) \mathbb{P}\left(X_\xi^{s+k-1}=1, X_\xi^s=0\right)+\frac{\alpha_\xi \beta_\xi}{\alpha_\xi+\beta_\xi} .
\end{aligned}
$$
So, 
$$
\begin{aligned}
 \mathbb{P}\left(X_{\xi}^{s+k} \neq X_\xi^s\right)&=\left(1-\alpha_\xi-\beta_\xi\right) \mathbb{P}\left(X_\xi^{s+k-1} \neq X_\xi^s\right)+ \frac{2 \alpha_\xi \beta_\xi}{\alpha_\xi+\beta_\xi}  \\
& =\left(1-\alpha_\xi-\beta_\xi\right) \mathbb{P}\left(X_\xi^{s+k-2} \neq X_\xi^s\right)+\left(1+1-\alpha_\xi-\beta_\xi\right) \frac{2 \alpha_\xi \beta_\xi}{\alpha_\xi+\beta_\xi} \\
& =\sum_{i=1}^k\left(1-\alpha_\xi-\beta_\xi\right)^{i-1} \frac{2 \alpha_\xi \beta_\xi}{\alpha_\xi+\beta_\xi}=\frac{2 \alpha_\xi \beta_\xi}{\left(\alpha_\xi+\beta_\xi\right)^2}\left(1-\left(1-\alpha_\xi-\beta_\xi\right)^k\right) \\
\end{aligned}
$$
and the result holds. 
\subsection*{A.3 Proof of Proposition 3}
In this case, for $A=\mathcal{F}_0^k, B=\mathcal{F}_{k+r}^{\infty}$, consider $A=A_0 \times\left\{x_k\right\}$ and $B=B_0 \times\left\{x_{k+r}\right\}$ with $x_k=0$ or 1 and $x_{k+r}=0$ or 1 . Then
$$
\begin{aligned}
|\mathbb{P}(A \cap B)-\mathbb{P}(A) \mathbb{P}(B)|
&=\mid \mathbb{P}\left(A_0, X_k ; B_0 X_{k+r}\right)-\mathbb{P}\left(A_0, X_k\right) \mathbb{P}\left(B_0, X_{k+r}\right) \mid \\
&= \mid\mathbb{P}\left(B_0 \mid X_{k+r}\right) \mathbb{P}\left(X_{k+r} \mid X_k\right) \mathbb{P}\left(X_k, A_0\right)-\mathbb{P}\left(A_0, X_k\right) \mathbb{P}\left(B_0, X_{k+r}\right)\mid \\
&= \mid \frac{\mathbb{P}\left(B_0, X_{k+r}\right)}{\mathbb{P}\left(X_{k+r}\right)} \cdot \frac{\mathbb{P}\left(X_{k+r}, X_k\right)}{\mathbb{P}\left(X_k\right)} \mathbb{P}\left(X_k, A_0\right)-\mathbb{P}\left(A_0, X_k\right) \mathbb{P}\left(B_0, X_{k+r}\right) \mid\\
&=\mid \mathbb{P}\left(A_0, X_k\right) \mathbb{P}\left(B_0, X_{k+r}\right) \frac{\mathbb{P}(X_{k+r}), \mathbb{P}(X_k) - \mathbb{P}\left(X_{k+r}\right) \mathbb{P}\left(X_k\right)}{\mathbb{P}\left(X_{k+r}\right) \mathbb{P}\left(X_k\right)}\mid\\
&=\frac{\mathbb{P}\left(A_0, X_k\right) \mathbb{P}\left(B_0, X_{k+r}\right)}{\mathbb{P}\left(X_{k+r}\right) \mathbb{P}\left(X_k\right)} \rho_\xi(r) \leqslant \rho_\xi(r),
\end{aligned}$$ then completes the proof.
\subsection*{A.4 Proof of Proposition 4}
\begin{lemma}
    For any $\xi \in \mathcal{E}$, denote $Y_{\xi}^t:=X_{\xi}^t\left(1-X_{\xi}^{t-1}\right)$, and let $\mathbf{Y}_t=\left(Y_{\xi}^t\right)_{\xi \in \mathcal{E}}$ be the $p \times p$ matrix at time t. Under the assumptions of Proposition 1, we have $\left\{\mathbf{Y}_t, t=1,2 \ldots\right\}$ is stationary such that for any $\xi_1,\xi_2 \in \mathcal{E}$, and $t, s \geq 1, t \neq s$,
$$
\begin{gathered}
E Y_{\xi}^t=\frac{\alpha_{\xi} \beta_{\xi}}{\alpha_{\xi}+\beta_{\xi}}, \quad \operatorname{Var}\left(Y_{\xi}^t\right)=\frac{\alpha_{\xi} \beta_{\xi}\left(\alpha_{\xi}+\beta_{\xi}-\alpha_{\xi} \beta_{\xi}\right)}{\left(\alpha_{\xi}+\beta_{\xi}\right)^2}, \\
\rho_{Y_{\xi}}(|t-s|) \equiv \operatorname{Corr}\left(Y_{\xi_1}^t, Y_{\xi_2}^s\right)= \begin{cases}-\frac{\alpha_{\xi} \beta_{\xi}\left(1-\alpha_{\xi}-\beta_{\xi}\right)^{|t-s|-1}}{\alpha_{\xi}+\beta_{\xi}-\alpha_{\xi} \beta_{\xi}} & \text { if }\xi_1=\xi_2, \\
0 & \text { otherwise. }\end{cases}
\end{gathered}
$$
\end{lemma}
\begin{proof}
    Notice that $Y_\xi^t=X_\xi^t\left(1-X_\xi^{t-1}\right)=\left(1-X_\xi^{t-1}\right)I\left(\varepsilon_\xi^t=1\right)$, then
$$
\begin{aligned}
 \mathbb{E}\left(Y_\xi^t\right)&=\mathbb{E}\left(1-X_\xi^{t-1}\right) \mathbb{P}\left(\varepsilon_\xi^t=1\right)=\frac{\alpha_\xi \beta_\xi}{\alpha_\xi+\beta_\xi} \\
 \mathbb{E}\left[\left(Y_\xi^{t}\right)^2\right]&=\mathbb{P}\left(\varepsilon_t^t=1\right)\left[1-2 \mathbb{E}\left(X_\xi^t\right)+\mathbb{E}\left[\left(Y_s^{t}\right)^2\right]\right]  =\left(1-2 \frac{\alpha_\xi+\beta_\xi}{\alpha_\xi}+\frac{\alpha_\xi \beta_\xi+\alpha_\xi^2}{\left(\alpha_\xi+\beta_\xi\right)^2}\right) \alpha_\xi \\
& =\frac{\alpha_\xi \beta_\xi\left(\alpha_\xi+\beta_\xi\right)}{\left(\alpha_\xi+\beta_\xi\right)^2} 
\end{aligned}
$$
    So $\operatorname{Var}\left(Y_\xi^t\right)=\mathbb{E}\left[\left(Y_\xi^{t}\right)^2\right]-\mathbb{E}\left(Y_\xi^t\right)^2=\frac{\alpha_\xi \beta_\xi\left(\alpha_\xi+\beta_\xi-\alpha_\xi \beta_\xi\right)}{\left(\alpha_\xi+\beta_\xi\right)^2}$.

Using $\mathbb{E}\left(X_\xi^t X_\xi^{t+k}\right)=\frac{\alpha_\xi}{\left(\alpha_\xi+\beta_\xi\right)^2}\left(\beta_\xi\left(1-\alpha_\xi-\beta_\xi\right)^k+\alpha_\xi\right)$ from Proposition 1, we have
$$
\begin{aligned}
E\left(Y_{\xi}^t Y_{\xi}^{t+k}\right) & =E\left[X_{\xi}^t\left(1-X_{\xi}^{t-1}\right)X_{\xi}^{t+k}\left(1-X_{\xi}^{t+k-1}\right) \right] \\
& = \mathbb{P}\left(X_\xi^t = X_\xi^{t+k}=1, X_\xi^{t-1} = X_\xi^{t+k-1}=0\right)\\
& =P\left(X_{\xi}^{t+k}=1 \mid X_{\xi}^{t+k-1}=0\right) P\left(X_{\xi}^{t+k-1}=0 \mid X_{\xi}^t=1\right) P\left(X_{\xi}^t=1 \mid X_{\xi}^{t-1}=0\right) P\left(X_{\xi}^{t-1}=0\right) \\
& =\frac{\alpha_{\xi}^2 \beta_{\xi}}{\alpha_{\xi}+\beta_{\xi}}\left[1-P\left(X_{\xi}^{t+k-1}=1 \mid X_{\xi}^t=1\right)\right] \\
& =\frac{\alpha_{\xi}^2 \beta_{\xi}}{\alpha_{\xi}+\beta_{\xi}}\left[1-\frac{E\left(X_{\xi}^{t+k-1} X_{\xi}^t\right)}{E X_{\xi}^t}\right] \\
& =\frac{\alpha_{\xi}^2 \beta_{\xi}}{\alpha_{\xi}+\beta_{\xi}}\left[1-\frac{\beta_{\xi}\left(1-\alpha_{\xi}-\beta_{\xi}\right)^{k-1}+\alpha_{\xi}}{\alpha_{\xi}+\beta_{\xi}}\right] \\
& =\frac{\alpha_{\xi}^2 \beta_{\xi}^2\left[1-\left(1-\alpha_{\xi}-\beta_{\xi}\right)^{k-1}\right]}{\left(\alpha_{\xi}+\beta_{\xi}\right)^2} .
\end{aligned}
$$
Therefore
$\operatorname{Cov}\left(Y_\xi^t, Y_\xi^{t+k}\right)=\mathbb{E}\left(Y_\xi^t Y_\xi^{t+k}\right)-\mathbb{E}\left(Y_\xi^t\right) \mathbb{E}\left(Y_\xi^{t+k}\right)=-\frac{\alpha_\xi^2 \beta^2_\xi\left(1-\alpha_\xi-\beta_\xi\right)^{k-1}}{\left(\alpha_\xi+\beta_\xi\right)^2}$ for any $k \geq 1$ and the ACF of process $\left\{Y^t_\xi\right\}$ is given as
$$\rho_{Y_\xi}\left(|t-s|\right) = \frac{\operatorname{Cov}\left(Y_\xi^t, Y_\xi^{s}\right)}{\sqrt{\operatorname{Var}\left(Y_\xi^t\right)\operatorname{Var}\left(Y_\xi^s\right)}} = -\frac{\alpha_\xi \beta_\xi \left(1-\alpha_\xi-\beta_\xi\right)^{\mid t-s\mid -1}}{\alpha_\xi+\beta_\xi-\alpha_\xi\beta_\xi}$$
\end{proof}

Notice that $Y_\xi^t$ is also $\alpha$-mixing. By Theorem 1 of \cite{merlevede2009bernstein}, we obtain gives the following result by considering $S_n = \sum_i X_i-\mathbb{E}(X_i)$, $x = n\epsilon$: 

\begin{lemma}
   Let conditions in Proposition 1 and $C 1$ hold. There exist positive constants $C_1$ and $C_2$ such that for all $n \geq 4$ and $0<\varepsilon<\frac{1}{(\log n)(\log \log n)}$,
$$
\begin{gathered}
P\left(\left|n^{-1} \sum_{t=1}^n X_{\xi}^t-E X_{\xi}^t\right|>\varepsilon\right) \leq \exp \left\{-C_1 n \varepsilon^2\right\} \\
P\left(\left|n^{-1} \sum_{t=1}^n Y_{\xi}^t-E Y_{\xi}^t\right|>\varepsilon\right) \leq \exp \left\{-C_2 n \varepsilon^2\right\}
\end{gathered}
$$
\end{lemma}

\textbf{Proof of Proposition 4}

Let $\varepsilon=C_0 \sqrt{\frac{\log p}{n}}$ with $C_0^2 C_1>K$ and $C_0^2 C_2>K$. Note that under condition C2 we have $\varepsilon=o\left(\frac{1}{(\log n)(\log \log n)}\right)$. Consequently by Lemma 1, Proposition 1 and Lemma 2, we have
$$
\begin{aligned}
& P\left(\left|n^{-1} \sum_{t=1}^n X_{\xi}^t-\frac{\alpha_{\xi}}{\alpha_{\xi}+\beta_{\xi}}\right|>C_0 \sqrt{\frac{\log p}{n}}\right) \leq \exp \left\{-C_0^2 C_1 \log p\right\}, \\
& P\left(\left|n^{-1} \sum_{t=1}^n Y_{\xi}^t-\frac{\alpha_{\xi} \beta_{\xi}}{\alpha_{\xi}+\beta_{\xi}}\right|>C_0 \sqrt{\frac{\log p}{n}}\right) \leq \exp \left\{-C_0^2 C_2 \log p\right\} .
\end{aligned}
$$

Consequently, with probability greater than $1-\exp \left\{-C_0^2 C_2 \log p\right\}$,

$$\frac{\alpha_{\xi} \beta_{\xi}}{\alpha_{\xi}+\beta_{\xi}}-C_0\sqrt{\frac{\log p}{n}} < n^{-1} \sum_{t=1}^n Y_{\xi}^t < \frac{\alpha_{\xi} \beta_{\xi}}{\alpha_{\xi}+\beta_{\xi}}+C_0\sqrt{\frac{\log p}{n}} $$

and similarly, with probability greater than $1-\exp \left\{-C_0^2 C_1 \log p\right\}$,

$$\frac{\beta_{\xi}}{\alpha_{\xi}+\beta_{\xi}}-\frac{1}{n}-C_0\sqrt{\frac{\log p}{n}} < \frac{\beta_{\xi}}{\alpha_{\xi}+\beta_{\xi}}-C_0\sqrt{\frac{\log p}{n}} < n^{-1} \sum_{t=1}^n (1-X_{\xi}^{t-1})  $$$$< \frac{\beta_{\xi}}{\alpha_{\xi}+\beta_{\xi}}+C_0\sqrt{\frac{\log p}{n}} <\frac{\beta_{\xi}}{\alpha_{\xi}+\beta_{\xi}}+\frac{1}{n}+C_0\sqrt{\frac{\log p}{n}} $$

Then, with probability greater than $(1-\exp \left\{-C_0^2 C_1 p\right\})(1-\exp \left\{-C_0^2 C_2 p\right\})$,
$$
\frac{\frac{\alpha_{\xi} \beta_{\xi}}{\alpha_{\xi}+\beta_{\xi}}-C_0 \sqrt{\frac{\log p}{n}}}{\frac{\beta_{\xi}}{\alpha_{\xi}+\beta_{\xi}} +\frac{1}{n}+C_0 \sqrt{\frac{\log p}{n}}}\leq \widehat{\alpha}_{\xi} =\frac{\sum_{t=1}^n Y_\xi^t}{\sum_{t=1}^n\left(1-X_\xi^{t-1}\right)}\leq \frac{\frac{\alpha_{\xi} \beta_{\xi}}{\alpha_{\xi}+\beta_{\xi}}+C_0 \sqrt{\frac{\log p}{n}}}{\frac{\beta_{\xi}}{\alpha_{\xi}+\beta_{\xi}}-\frac{1}{n}-C_0 \sqrt{\frac{\log p}{n}}} .
$$

Note that when $n$ and $\frac{n}{\log p}$ are large enough such that, $\frac{1}{n} \leq C_0 \sqrt{\frac{\log p}{n}} \leq l / 4$, by condition C1, we have
$$
\alpha_{\xi}-\frac{\frac{\alpha_{\xi} \beta_{\xi}}{\alpha_{\xi}+\beta_{\xi}}-C_0 \sqrt{\frac{\log p}{n}}}{\frac{\beta_{\xi}}{\alpha_{\xi}+\beta_{\xi}}+\frac{1}{n}+C_0 \sqrt{\frac{\log p}{n}}} \leq \frac{2 C_0 \alpha_{\xi} \sqrt{\frac{\log p}{n}}+C_0 \sqrt{\frac{\log p}{n}}}{\frac{\beta_{\xi}}{\alpha_{\xi}+\beta_{\xi}}} = 
\frac{2\alpha_\xi+1}{\frac{\beta_\xi}{\alpha_\xi+\beta_\xi}}C_0 \sqrt{\frac{\log p}{n}}$$$$\leq \frac{3}{l}C_0 \sqrt{\frac{\log p}{n}} \leq 3 l^{-1} C_0 \sqrt{\frac{\log p}{n}},
$$
and
$$
\frac{\frac{\alpha_{\xi} \beta_{\xi}}{\alpha_{\xi}+\beta_{\xi}}+C_0 \sqrt{\frac{\log p}{n}}}{\frac{\beta_{\xi}}{\alpha_{\xi}+\beta_{\xi}}-\frac{1}{n}-C_0 \sqrt{\frac{\log p}{n}}}-\alpha_{\xi} \leq \frac{2 C_0 \alpha_{\xi} \sqrt{\frac{\log p}{n}}+C_0 \sqrt{\frac{\log p}{n}}}{\frac{\beta_{\xi}}{\alpha_{\xi}+\beta_{\xi}}-\frac{l}{2}} \leq 6 l^{-1} C_0 \sqrt{\frac{\log p}{n}},
$$

Therefore we conclude that when when $n$ and $\frac{n}{\log p}$ are large enough,
$$
P\left(\left|\widehat{\alpha}_{\xi}-\alpha_{\xi}\right| \geq 6 l^{-1} C_0 \sqrt{\frac{\log p}{n}}\right) \leq \exp \left\{-C_0^2 C_1 \log p\right\}+\exp \left\{-C_0^2 C_2 \log p\right\} $$$$- \exp\left\{-C_0^2 C_1 \log p-C_0^2 C_2 \log p\right\}\leq \exp \left\{-C_0^2 C_1 \log p\right\}+\exp \left\{-C_0^2 C_2 \log p\right\}
$$

For any $c>K$, the concentration inequalities in Proposition 4 can then be concluded by setting $C_0=\max \left\{\sqrt{c / C_1}, \sqrt{c / C_2}\right\}$. Further, as $n, p \rightarrow \infty$, we immediately have $\max _{\xi \in \mathcal{E}}\left|\widehat{\alpha}_{\xi}-\alpha_{\xi}\right|=O_p\left(\sqrt{\frac{\log p}{n}}\right)$. Convergence of $\widehat{\beta}_{\xi}$ can be proved similarly.
\subsection*{A.5 Proof of Proposition 5}
Note that the likelihood function is 
$$
\begin{aligned}
    \mathbb{P}(X_\xi^n, X_\xi^{n-1}, \cdots, X_\xi^1) &= \prod_{t=1}^n \mathbb{P}(X_\xi^t \mid X_\xi^{t-1}) \\
    &= \prod_{t=1}^n \alpha_\xi^{X_\xi^t (1-X_\xi^{t-1})} (1-\alpha_\xi)^{(1-X_\xi^t)(1-X_\xi^{t-1})}\beta_\xi^{(1-X_\xi^t) X_\xi^{t-1}} (1-\beta_\xi)^{X_\xi^tX_\xi^{t-1}},
\end{aligned}$$
then the log-likelihood function for $\left(\alpha_{\xi}, \beta_{\xi}\right)$ is:
$$
\begin{aligned}
l\left(\alpha_{\xi}, \beta_{\xi}\right) & =\log \left(\alpha_{\xi}\right) \sum_{t=1}^n X_{\xi}^t\left(1-X_{\xi}^{t-1}\right)+\log \left(1-\alpha_{\xi}\right) \sum_{t=1}^n\left(1-X_{\xi}^t\right)\left(1-X_{\xi}^{t-1}\right) \\
& +\log \left(\beta_{\xi}\right) \sum_{t=1}^n\left(1-X_{\xi}^t\right) X_{\xi}^{t-1}+\log \left(1-\beta_{\xi}\right) \sum_{t=1}^n X_{\xi}^t X_{\xi}^{t-1}
\end{aligned}
$$
and the score functions are
$$\begin{aligned}
\frac{\partial l\left(\alpha_{\xi}, \beta_{\xi}\right)}{\partial \alpha_{\xi}} & =\frac{1}{\alpha_{\xi}} \sum_{t=1}^n X_{\xi}^t\left(1-X_{\xi}^{t-1}\right)-\frac{1}{1-\alpha_{\xi}} \sum_{t=1}^n\left(1-X_{\xi}^t\right)\left(1-X_{\xi}^{t-1}\right), \\
\frac{\partial l\left(\alpha_{\xi}, \beta_{\xi}\right)}{\partial \beta_{\xi}} & =\frac{1}{\beta_{\xi}} \sum_{t=1}^n\left(1-X_{\xi}^t\right) X_{\xi}^{t-1}-\frac{1}{1-\beta_{\xi}} \sum_{t=1}^n X_{\xi}^t X_{\xi}^{t-1}
\end{aligned}$$
By the consistency of MLEs, we conclude that $\sqrt{n}\left(\widehat{\alpha}_{\xi}-\alpha_{\xi}\right), \sqrt{n}\left(\widehat{\beta}_{\xi}-\beta_{\xi}\right)$ converges to the normal distribution with mean $\mathbf{0}$ and covariance matrix $I\left(\alpha_{\xi}, \beta_{\xi}\right)^{-1}$, where $\bar{I}\left(\alpha_{\xi}, \beta_{\xi}\right)$ is the Fisher information matrix given as:
$$
I\left(\alpha_{\xi}, \beta_{\xi}\right)=\frac{1}{n} E\left[\begin{array}{cc}
\frac{\sum_{t=1}^n X_{\xi}^t\left(1-X_{\xi}^{t-1}\right)}{\alpha_{\xi}^2}+\frac{\sum_{t=1}^n\left(1-X_{\xi}^t\right)\left(1-X_{\xi}^{t-1}\right)}{\left(1-\alpha_{\xi}\right)^2} & 0 \\
0 & \frac{\sum_{t=1}^n\left(1-X_{\xi}^t\right) X_{\xi}^{t-1}}{\beta_{\xi}^2}+\frac{\sum_{t=1}^n X_{\xi}^t X_{\xi}^{t-1}}{\left(1-\beta_{\xi}\right)^2}
\end{array}\right] .
$$
Note that
$$
\begin{aligned}
& \frac{1}{n} E \sum_{t=1}^n X_{\xi}^t\left(1-X_{\xi}^{t-1}\right)=\frac{1}{n} E \sum_{t=1}^n\left(1-X_{\xi}^t\right) X_{\xi}^{t-1}=\frac{\alpha_{\xi} \beta_{\xi}}{\alpha_{\xi}+\beta_{\xi}}, \\
& \frac{1}{n} E \sum_{t=1}^n\left(1-X_{\xi}^t\right)\left(1-X_{\xi}^{t-1}\right)=\frac{\beta_{\xi}}{\alpha_{\xi}+\beta_{\xi}}-\frac{\alpha_{\xi} \beta_{\xi}}{\alpha_{\xi}+\beta_{\xi}}=\frac{\left(1-\alpha_{\xi}\right) \beta_{\xi}}{\alpha_{\xi}+\beta_{\xi}}, \\
& \frac{1}{n} E \sum_{t=1}^n X_{\xi}^t X_{\xi}^{t-1}=\frac{\alpha_{\xi}\left(1-\beta_{\xi}\right)}{\alpha_{\xi}+\beta_{\xi}} .
\end{aligned}
$$
We thus have
$$
\begin{aligned}
I\left(\alpha_{\xi}, \beta_{\xi}\right) & =\left[\begin{array}{cc}
\frac{\beta_{\xi}}{\alpha_{\xi}\left(\alpha_{\xi}+\beta_{\xi}\right)}+\frac{\beta_{\xi}}{\left(\alpha_{\xi}+\beta_{\xi}\right)\left(1-\alpha_{\xi}\right)} & 0 \\
0 & \frac{\alpha_{\xi}}{\beta_{\xi}\left(\alpha_{\xi}+\beta_{\xi}\right)}+\frac{\alpha_{\xi}}{\left(1-\beta_{\xi}\right)\left(\alpha_{\xi}+\beta_{\xi}\right)}
\end{array}\right] \\
& =\left[\begin{array}{cc}
\frac{\beta_{\xi}}{\alpha_{\xi}\left(\alpha_{\xi}+\beta_{\xi}\right)\left(1-\alpha_{\xi}\right)} & 0 \\
0 & \frac{\alpha_{\xi}}{\beta_{\xi}\left(\alpha_{\xi}+\beta_{\xi}\right)\left(1-\beta_{\xi}\right)}
\end{array}\right] .
\end{aligned}
$$

Consequently, we have
$$
\left[\begin{array}{l}
\sqrt{n}\left(\widehat{\alpha}_{\xi}-\alpha_{\xi}\right) \\
\sqrt{n}\left(\widehat{\beta}_{\xi}-\beta_{\xi}\right)
\end{array}\right] \rightarrow N\left(\mathbf{0},\left[\begin{array}{cc}
\frac{\alpha_{\xi}\left(\alpha_{\xi}+\beta_{\xi}\right)\left(1-\alpha_{\xi}\right)}{\beta_{\xi}} & 0 \\
0 & \frac{\beta_{\xi}\left(\alpha_{\xi}+\beta_{\xi}\right)\left(1-\beta_{\xi}\right)}{\alpha_{\xi}}
\end{array}\right]\right)
$$

\subsection*{A.6 Proof of Proposition 6}
   
Notice that $(\lambda, v)$ is an eigenpair of $\mathbf{L}$ if and only if $(2-\lambda, v)$ is an eigenpair of $\mathbf{Y} = \mathbf{D}_1^{-1/2}\mathbf{A}_1\mathbf{D}_1^{-1/2} + \mathbf{D}_2^{-1/2}\mathbf{A}_2\mathbf{D}_2^{-1/2}$. Hence, it is sufficient to study on the eigenvalues of $\mathbf{Y}$. 
Define matrices $\widetilde{\mathbf{D}}_i, \widetilde{\mathbf{J}}_i \in \mathbb{R}^{q \times q}$ for $i \in \{1,2\}$ with $\mathbf{D}_i\mathbf{Z} = \mathbf{Z}\widetilde{\mathbf{D}}_i$ and $\mathbf{J}_i\mathbf{Z} = \mathbf{Z}\widetilde{\mathbf{J}}_i$. We can write that by $\mathbf{G}_1 = \left(\widetilde{\mathbf{D}}_1^{-1} \mathbf{Z}^T \mathbf{Z}\right)^{1 / 2} \boldsymbol{\Omega}_1 \left(\mathbf{Z}^T \mathbf{Z} \widetilde{\mathbf{D}}_1^{-1}\right)^{1 / 2}-\widetilde{\mathbf{J}}_1\widetilde{\mathbf{D}}_1^{-1}$ and $\mathbf{G}_2 = \left(\widetilde{\mathbf{D}}_2^{-1} \mathbf{Z}^T \mathbf{Z}\right)^{1 / 2} \boldsymbol{\Omega}_2 \left(\mathbf{Z}^T \mathbf{Z} \widetilde{\mathbf{D}}_2^{-1}\right)^{1 / 2} -\widetilde{\mathbf{J}}_2\widetilde{\mathbf{D}}_2^{-1}$ with eigen-decomposition $\mathbf{G}_1 +\mathbf{G}_2 = \mathbf{U}\widetilde{\boldsymbol{\Omega}}\mathbf{U}^T$. Then, for $\boldsymbol{\Gamma}=\mathbf{Z}\left(\mathbf{Z}^T \mathbf{Z}\right)^{-1 / 2} \mathbf{U}$
\begin{align*}
    \mathbf{Y} \mathbf{\Gamma} &= \left(\mathbf{D}_1^{-1/2}\mathbf{A}_1\mathbf{D}_1^{-1/2} + \mathbf{D}_2^{-1/2}\mathbf{A}_2\mathbf{D}_2^{-1/2}\right) \mathbf{\Gamma} \\
    &= \mathbf{D}_1^{-1/2}\mathbf{A}_1\mathbf{D}_1^{-1/2}\mathbf{\Gamma} + \mathbf{D}_2^{-1/2}\mathbf{A}_2\mathbf{D}_2^{-1/2} \mathbf{\Gamma}\\
    &= \mathbf{D}_1^{-1/2}(\mathbf{Z} \boldsymbol{\Omega}_1 \mathbf{Z}^{\top} - \mathbf{J}_1)\mathbf{D}_1^{-1/2} \mathbf{Z}(\mathbf{Z}^T \mathbf{Z})^{-1/2} \mathbf{U} + \mathbf{D}_2^{-1/2}(\mathbf{Z} \boldsymbol{\Omega}_2 \mathbf{Z}^{\top} - \mathbf{J}_2)\mathbf{D}_2^{-1/2} \mathbf{Z}(\mathbf{Z}^T \mathbf{Z})^{-1/2} \mathbf{U}\\
    &= \mathbf{D}_1^{-1/2}(\mathbf{Z} \boldsymbol{\Omega}_1 \mathbf{Z}^{\top} - \mathbf{J}_1) \mathbf{Z}(\mathbf{Z}^T \mathbf{Z})^{-1/2}\widetilde{\mathbf{D}}_1^{-1/2} \mathbf{U} + \mathbf{D}_2^{-1/2}(\mathbf{Z} \boldsymbol{\Omega}_2 \mathbf{Z}^{\top} - \mathbf{J}_2) \mathbf{Z}(\mathbf{Z}^T \mathbf{Z})^{-1/2}\widetilde{\mathbf{D}}_2^{-1/2} \mathbf{U}\\
    &= \mathbf{D}_1^{-1/2}(\mathbf{Z} \boldsymbol{\Omega}_1 (\mathbf{Z}^T \mathbf{Z})^{1/2}- \mathbf{J}_1\mathbf{Z}(\mathbf{Z}^T \mathbf{Z})^{-1/2}) \widetilde{\mathbf{D}}_1^{-1/2} \mathbf{U} \\
    & \quad+ \mathbf{D}_2^{-1/2}(\mathbf{Z} \boldsymbol{\Omega}_2 (\mathbf{Z}^T \mathbf{Z})^{1/2} - \mathbf{J}_2\mathbf{Z}(\mathbf{Z}^T \mathbf{Z})^{-1/2}) \widetilde{\mathbf{D}}_2^{-1/2} \mathbf{U} \\ 
     &= (\mathbf{Z}\widetilde{\mathbf{D}}_1^{-1/2} \boldsymbol{\Omega}_1 (\mathbf{Z}^T \mathbf{Z})^{1/2}\widetilde{\mathbf{D}}_1^{-1/2} - \mathbf{Z}(\mathbf{Z}^T \mathbf{Z})^{-1/2}\widetilde{\mathbf{D}}_1^{-1/2}\widetilde{\mathbf{J}}_1\widetilde{\mathbf{D}}_1^{-1/2} ) \mathbf{U} \\
    & \quad+ (\mathbf{Z}\widetilde{\mathbf{D}}_2^{-1/2} \boldsymbol{\Omega}_2 (\mathbf{Z}^T \mathbf{Z})^{1/2}\widetilde{\mathbf{D}}_2^{-1/2} - \mathbf{Z}(\mathbf{Z}^T \mathbf{Z})^{-1/2}\widetilde{\mathbf{D}}_2^{-1/2}\widetilde{\mathbf{J}}_2\widetilde{\mathbf{D}}_2^{-1/2} ) \mathbf{U}\\ 
     &= \mathbf{Z}(\mathbf{Z}^T \mathbf{Z})^{-1/2}((\mathbf{Z}^T \mathbf{Z})^{1/2}\widetilde{\mathbf{D}}_1^{-1/2} \boldsymbol{\Omega}_1 (\mathbf{Z}^T \mathbf{Z})^{1/2}\widetilde{\mathbf{D}}_1^{-1/2} 
     -  \widetilde{\mathbf{J}}_1 \widetilde{\mathbf{D}}_1^{-1}) \mathbf{U} \\
    & \quad+ \mathbf{Z}(\mathbf{Z}^T \mathbf{Z})^{-1/2}((\mathbf{Z}^T \mathbf{Z})^{1/2}\widetilde{\mathbf{D}}_2^{-1/2} \boldsymbol{\Omega}_2 (\mathbf{Z}^T \mathbf{Z})^{1/2}\widetilde{\mathbf{D}}_2^{-1/2} 
     -  \widetilde{\mathbf{J}}_2 \widetilde{\mathbf{D}}_2^{-1}) \mathbf{U} \\
     &= \mathbf{Z}(\mathbf{Z}^T \mathbf{Z})^{-1/2}\mathbf{G}_1 \mathbf{U} + \mathbf{Z}(\mathbf{Z}^T \mathbf{Z})^{-1/2}\mathbf{G}_2 \mathbf{U}\\
     &= \mathbf{Z}(\mathbf{Z}^T \mathbf{Z})^{-1/2}(\mathbf{G}_1+\mathbf{G}_2) \mathbf{U} \\
     &= \mathbf{Z}(\mathbf{Z}^T \mathbf{Z})^{-1/2}\mathbf{U}(\widetilde{\boldsymbol{\Omega}}_1+\widetilde{\boldsymbol{\Omega}}_2) \\
     &= \mathbf{\Gamma} \widetilde{\boldsymbol{\Omega}}
\end{align*}
which implies that the columns of $\boldsymbol{\Gamma}$ are the eigenvectors of $\mathbf{Y}$ corresponding to the $q$ eigenvalues in $\widetilde{\boldsymbol{\Omega}}$. Also, since $\mathbf{U}$ is
orthonormal, it is easy to verify that the columns of $\boldsymbol{\Gamma}$ are orthonormal.

Then we need to derive the condition for leading eigenvectors. Let $\mathbf{\Gamma}'$ be the matrix containing remaining eigenvectors of $\mathbf{Y}$ corresponding to the remaining eigenvalues in $\widetilde{\boldsymbol{\Omega}}'$. Then $\mathbf{Y} \mathbf{\Gamma}'=\mathbf{\Gamma}' \widetilde{\boldsymbol{\Omega}}'.$  We need to show that eigenvalues in $\widetilde{\boldsymbol{\Omega}}$ are strictly larger than other eigenvalues in $\mathbf{Y}$, that is, $\min _{1 \leq i \leq k}\left(\widetilde{\boldsymbol{\Omega}}\right)_{i i}-\max _{1 \leq i \leq(n-k)}\left(\widetilde{\boldsymbol{\Omega}}'\right)_{i i} >0$. 

Since $\mathbf{Y}$ is symmetric, it is diagonalizable and then $(\mathbf{\Gamma}')^T \mathbf{\Gamma}=(\mathbf{\Gamma}')^T\mathbf{Z}\left(\mathbf{Z}^T \mathbf{Z}\right)^{-1 / 2} \mathbf{U} = 0.$ Since $\mathbf{Z}^T \mathbf{Z}$ and $\mathbf{U}$ are invertible, then $(\mathbf{\Gamma}')^T\mathbf{Z} = 0$ and hence, \begin{align*}
\mathbf{\Gamma}' \widetilde{\boldsymbol{\Omega}}'&=\mathbf{Y} \mathbf{\Gamma}'= \mathbf{D}_1^{-1/2}\mathbf{A}_1\mathbf{D}_1^{-1/2}\mathbf{\Gamma}' + \mathbf{D}_2^{-1/2}\mathbf{A}_2\mathbf{D}_2^{-1/2}\mathbf{\Gamma}' \\
&= \mathbf{D}_1^{-1/2}(\mathbf{Z} \mathbf{\Omega}_1 \mathbf{Z}^{\top}-\mathbf{J}_1)\mathbf{D}_1^{-1/2}\mathbf{\Gamma}' + \mathbf{D}_2^{-1/2}(\mathbf{Z} \mathbf{\Omega}_2 \mathbf{Z}^{\top}-\mathbf{J}_2)\mathbf{D}_2^{-1/2}\mathbf{\Gamma}'\\
&= \mathbf{D}_1^{-1/2}\mathbf{Z} \mathbf{\Omega}_1 \widetilde{\mathbf{D}}_1^{-1/2}\mathbf{Z}^{\top}\mathbf{\Gamma}' + \mathbf{D}_2^{-1/2}\mathbf{Z} \mathbf{\Omega}_2 \widetilde{\mathbf{D}}_2^{-1/2}\mathbf{Z}^{\top}\mathbf{\Gamma}-\mathbf{D}_1^{-1/2}\mathbf{J}_1\mathbf{D}_1^{-1/2}\mathbf{\Gamma}'-\mathbf{D}_2^{-1/2}\mathbf{J}_2\mathbf{D}_2^{-1/2}\mathbf{\Gamma}' \\
&=-(\mathbf{D}_1^{-1/2}\mathbf{J}_1\mathbf{D}_1^{-1/2}+ \mathbf{D}_2^{-1/2}\mathbf{J}_2\mathbf{D}_2^{-1/2})\mathbf{\Gamma}'. 
\end{align*}
So the column of $\mathbf{\Gamma}'$ are eigenvectors of $-(\mathbf{D}_1^{-1/2}\mathbf{J}_1\mathbf{D}_1^{-1/2}+ \mathbf{D}_2^{-1/2}\mathbf{J}_2\mathbf{D}_2^{-1/2})$. Since $\mathbf{D}_i, \mathbf{J}_i$ are diagonal, then eigenvalues in $\widetilde{\boldsymbol{\Omega}}'$ are subset of entries of $-(\mathbf{D}_1^{-1}\mathbf{J}_1+ \mathbf{D}_2^{-1}\mathbf{J}_2).$ So, $\max _{1 \leq i \leq(n-k)}\left(\widetilde{\boldsymbol{\Omega}}'\right)_{i i} \leq \min_{1 \leq i \leq n}\frac{(\mathbf{J}_1)_{ii}}{(\mathbf{D}_1)_{ii}}+\frac{(\mathbf{J}_2)_{ii}}{(\mathbf{D}_2)_{ii}}$. 

On the other hand, using Weyl's inequality, we have
\begin{align*}
    \min _{1 \leq i \leq k}\left(\widetilde{\boldsymbol{\Omega}}\right)_{i i}& = \lambda_{\min}(\mathbf{G}_1 +\mathbf{G}_2)\\&
    = \lambda_{\min}(\left(\widetilde{\mathbf{D}}_1^{-1} \mathbf{Z}^T \mathbf{Z}\right)^{1 / 2} \boldsymbol{\Omega}_1 \left(\mathbf{Z}^T \mathbf{Z} \widetilde{\mathbf{D}}_1^{-1}\right)^{1 / 2}+ \left(\widetilde{\mathbf{D}}_2^{-1} \mathbf{Z}^T \mathbf{Z}\right)^{1 / 2} \boldsymbol{\Omega}_2 \left(\mathbf{Z}^T \mathbf{Z} \widetilde{\mathbf{D}}_2^{-1}\right)^{1 / 2} \\& \quad
    -\widetilde{\mathbf{J}}_1\widetilde{\mathbf{D}}_1^{-1 } -\widetilde{\mathbf{J}}_2\widetilde{\mathbf{D}}_2^{-1}) \\& \geq 
\lambda_{\min}\left(\left(\widetilde{\mathbf{D}}_1^{-1} \mathbf{Z}^T \mathbf{Z}\right)^{1 / 2} \boldsymbol{\Omega}_1 \left(\mathbf{Z}^T \mathbf{Z} \widetilde{\mathbf{D}}_1^{-1}\right)^{1 / 2}\right)+ \lambda_{\min}\left(\left(\widetilde{\mathbf{D}}_2^{-1} \mathbf{Z}^T \mathbf{Z}\right)^{1 / 2} \boldsymbol{\Omega}_2 \left(\mathbf{Z}^T \mathbf{Z} \widetilde{\mathbf{D}}_2^{-1}\right)^{1 / 2}\right)\\& \quad - \lambda_{\max}\left(\widetilde{\mathbf{J}}_1\widetilde{\mathbf{D}}_1^{-1 }+\widetilde{\mathbf{J}}_2\widetilde{\mathbf{D}}_2^{-1}\right)  \\&= 
\lambda_{\min}\left(\left(\widetilde{\mathbf{D}}_1^{-1} \mathbf{Z}^T \mathbf{Z}\right)^{1 / 2} \boldsymbol{\Omega}_1 \left(\mathbf{Z}^T \mathbf{Z} \widetilde{\mathbf{D}}_1^{-1}\right)^{1 / 2}\right)+ \lambda_{\min}\left(\left(\widetilde{\mathbf{D}}_2^{-1} \mathbf{Z}^T \mathbf{Z}\right)^{1 / 2} \boldsymbol{\Omega}_2 \left(\mathbf{Z}^T \mathbf{Z} \widetilde{\mathbf{D}}_2^{-1}\right)^{1 / 2}\right)\\& \quad - \max_{1 \leq i \leq n}\frac{(\mathbf{J}_1)_{ii}}{(\mathbf{D}_1)_{ii}}+\frac{(\mathbf{J}_2)_{ii}}{(\mathbf{D}_2)_{ii}} 
\end{align*}
where $\lambda_{\min}$ indicates the minimum eigenvalue. 

Notice that $\mathbf{Z}^T \mathbf{Z}$ is diagonal with the $i_{th}$ entry be the size of $i_{th}$ community. Take $x \in \mathbb{R}^q$. Then $x^T\left(\widetilde{\mathbf{D}}_1^{-1} \mathbf{Z}^T \mathbf{Z}\right)^{1 / 2} \boldsymbol{\Omega}_1 \left(\mathbf{Z}^T \mathbf{Z} \widetilde{\mathbf{D}}_1^{-1}\right)^{1 / 2}x \geq \lambda_{\min}(\boldsymbol{\Omega}_1)x^T\left(\widetilde{\mathbf{D}}_1^{-1} \mathbf{Z}^T \mathbf{Z}\right)x \geq \lambda_{\min}(\boldsymbol{\Omega}_1)$ \\$\min_{1 \leq i\leq q}\frac{(\mathbf{Z}^T \mathbf{Z})_{ii}}{(\mathbf{D})_{ii}}x^Tx$. By Rayleigh's principle, $\lambda_{\min}\left(\left(\widetilde{\mathbf{D}}_1^{-1} \mathbf{Z}^T \mathbf{Z}\right)^{1 / 2} \boldsymbol{\Omega}_1 \left(\mathbf{Z}^T \mathbf{Z} \widetilde{\mathbf{D}}_1^{-1}\right)^{1 / 2}\right) = \min_x \frac{x^T\left(\widetilde{\mathbf{D}}_1^{-1} \mathbf{Z}^T \mathbf{Z}\right)^{1 / 2} \boldsymbol{\Omega}_1 \left(\mathbf{Z}^T \mathbf{Z} \widetilde{\mathbf{D}}_1^{-1}\right)^{1 / 2}x}{x^Tx} \geq \lambda_{\min}(\boldsymbol{\Omega}_1)\min_{1 \leq i\leq q}\frac{(\mathbf{Z}^T \mathbf{Z})_{ii}}{(\mathbf{D}_1)_{ii}}. $ Hence, we have
$\min _{1 \leq i \leq k}\left(\widetilde{\boldsymbol{\Omega}}\right)_{i i}-\max _{1 \leq i \leq(n-k)}\left(\widetilde{\boldsymbol{\Omega}}'\right)_{i i} >\lambda_{\min}(\boldsymbol{\Omega}_1)\min_{1 \leq i\leq q}\frac{(\mathbf{Z}^T \mathbf{Z})_{ii}}{(\mathbf{D}_1)_{ii}}+\lambda_{\min}(\boldsymbol{\Omega}_2)\min_{1 \leq i\leq q}\frac{(\mathbf{Z}^T \mathbf{Z})_{ii}}{(\mathbf{D}_2)_{ii}}- \max_{1 \leq i \leq n}\left(\frac{(\mathbf{J}_1)_{ii}}{(\mathbf{D}_1)_{ii}}+\frac{(\mathbf{J}_2)_{ii}}{(\mathbf{D}_2)_{ii}}\right) $ \\$ -\min_{1 \leq i \leq n}\left(\frac{(\mathbf{J}_1)_{ii}}{(\mathbf{D}_1)_{ii}}+\frac{(\mathbf{J}_2)_{ii}}{(\mathbf{D}_2)_{ii}}\right)=\delta  >0$, which finishes the proof. 

\subsection*{A.7 Proof of Theorem 1}
We firstly introduce some technical lemmas. 

\begin{lemma}
   Under the assumptions of Proposition 2, we have, there exists a constant $C_l>0$ such that
$$
\begin{aligned}
& \operatorname{Cov}\left(\sum_{t=1}^n Y_{\xi}^t, \sum_{t=1}^n\left(1-X_{\xi}^{t-1}\right)\right)=-\operatorname{Cov}\left(\sum_{t=1}^n Y_{\xi}^t, \sum_{t=1}^n X_{\xi}^{t-1}\right) \\
= & \frac{n \alpha_{\xi} \beta_{\xi}\left[2 \alpha_{\xi}\left(1-\beta_{\xi}\right)+\alpha_{\xi}+\beta_{\xi}-2 \beta_{\xi}^2\right]}{\left(\alpha_{\xi}+\beta_{\xi}\right)^3}+C_{\xi},
\end{aligned}
$$
with $\left|C_{\xi}\right| \leq C_l$ for any $C_{\xi},\xi \in \mathcal{E}$.
\end{lemma}

\begin{proof}
    Notice that under condition C1, we have for every $\xi \in \mathcal{E}, \sum_{h=1}^{n-1} (1-\alpha_{\xi}-\beta_{\xi})^{h-1}=\frac{1-(1-\alpha_{\xi}-\beta_{\xi})^n}{\alpha_{\xi}+\beta_{\xi}}=\frac{1}{\alpha_{\xi}+\beta_{\xi}}+O_l(1)$, and $\sum_{h=1}^{n-1} h (1-\alpha_{\xi}-\beta_{\xi})^{h-1}=\frac{1-(1-\alpha_{\xi}-\beta_{\xi})^n-n(\alpha_{\xi}+\beta_{\xi}) (1-\alpha_{\xi}-\beta_{\xi})^{n-1}}{(\alpha_{\xi}+\beta_{\xi})^2}=O_l(1)$. Here, $O_l(1)$ is a constant term bounded by a large enough constant $C_l$ that depends on $l$ only. Then,

$$
\begin{aligned}
& \operatorname{Cov}\left(\sum_{t=1}^n Y_{\xi}^t, \sum_{t=1}^n\left(1-X_{\xi}^{t-1}\right)\right)=-\operatorname{Cov}\left(\sum_{t=1}^n Y_{\xi}^t, \sum_{t=1}^n X_{\xi}^{t-1}\right) \\
= & -\sum_{t=1}^n \sum_{s=1}^n\left[E\left(1-X_{\xi}^{t-1}\right) X_{\xi}^t X_{\xi}^{s-1}-\frac{\alpha_{\xi} \beta_{\xi}}{\alpha_{\xi}+\beta_{\xi}} \cdot \frac{\alpha_{\xi}}{\alpha_{\xi}+\beta_{\xi}}\right] \\
= & -\sum_{t=1}^n \sum_{s=1}^n\left\{\frac{\alpha_{\xi}}{\left(\alpha_{\xi}+\beta_{\xi}\right)^2}\left[\beta_{\xi}\left(1-\alpha_{\xi}-\beta_{\xi}\right)^{|t-s+1|}+\alpha_{\xi}\right]-\frac{\alpha_{\xi}^2 \beta_{\xi}}{\left(\alpha_{\xi}+\beta_{\xi}\right)^2}\right\} \\
& +\sum_{t=1}^n \sum_{s=1}^n E\left(X_{\xi}^{t-1} X_{\xi}^t X_{\xi}^{s-1}\right) \\
= & \left[-\sum_{t=1}^n \sum_{s=1}^n \frac{\alpha_{\xi} \beta_{\xi}\left(1-\alpha_{\xi}-\beta_{\xi}\right)^{|t-s+1|}}{\left(\alpha_{\xi}+\beta_{\xi}\right)^2}-\frac{n^2 \alpha_{\xi}^2\left(1-\beta_{\xi}\right)}{\left(\alpha_{\xi}+\beta_{\xi}\right)^2} +(2 n-1) E\left(X_{\xi}^{t-1} X_{\xi}^t\right)\right]\\
& +\left[\sum_{s<t} E\left(X_{\xi}^{t-1} X_{\xi}^t X_{\xi}^{s-1}\right)+\sum_{s>t+1} E\left(X_{\xi}^{t-1} X_{\xi}^t X_{\xi}^{s-1}\right)\right] \\
= &\left[-\frac{\alpha_{\xi} \beta_{\xi}}{\left(\alpha_{\xi}+\beta_{\xi}\right)^2}\left[n+\frac{2 n\left(1-\alpha_{\xi}-\beta_{\xi}\right)}{\alpha_{\xi}+\beta_{\xi}}\right]-\frac{n^2 \alpha_{\xi}^2\left(1-\beta_{\xi}\right)}{\left(\alpha_{\xi}+\beta_{\xi}\right)^2}+\frac{2 n \alpha_{\xi}\left[\beta_{\xi}\left(1-\alpha_{\xi}-\beta_{\xi}\right)+\alpha_{\xi}\right]}{\left(\alpha_{\xi}+\beta_{\xi}\right)^2}+O_l(1)\right]\\
& +\left[\sum_{s<t} E\left(X_{\xi}^{t-1} X_{\xi}^t X_{\xi}^{s-1}\right)+\sum_{s>t+1} E\left(X_{\xi}^{t-1} X_{\xi}^t X_{\xi}^{s-1}\right)\right].
\end{aligned}
$$

For the last two terms, we have
$$
\begin{aligned}
& \sum_{s<t} E\left(X_{\xi}^{t-1} X_{\xi}^t X_{\xi}^{s-1}\right)+\sum_{s>t+1} E\left(X_{\xi}^{t-1} X_{\xi}^t X_{\xi}^{s-1}\right) \\
= & \left(1-\beta_{\xi}\right) \sum_{s<t} E\left(X_{\xi}^{t-1} X_{\xi}^{s-1}\right)+\left(1-\beta_{\xi}\right) \sum_{s>t+1} E\left(X_{\xi}^{s-1} X_{\xi}^t\right) \\
= & \frac{\left(1-\beta_{\xi}\right) \alpha_{\xi}}{\left(\alpha_{\xi}+\beta_{\xi}\right)^2} \left(\sum_{h=1}^{n-1}(n-h)\left[\beta_{\xi}\left(1-\alpha_{\xi}-\beta_{\xi}\right)^h+\alpha_{\xi}\right] + \sum_{h=2}^{n-1}(n-h)\left[\beta_{\xi}\left(1-\alpha_{\xi}-\beta_{\xi}\right)^{h-1}+\alpha_{\xi}\right]\right) \\
= & \frac{(n-1)^2 \alpha_{\xi}^2\left(1-\beta_{\xi}\right)}{\left(\alpha_{\xi}+\beta_{\xi}\right)^2}+\frac{2 n\left(1-\beta_{\xi}\right) \alpha_{\xi} \beta_{\xi}}{\left(\alpha_{\xi}+\beta_{\xi}\right)^3}+O_l(1) .
\end{aligned}
$$

Hence, we have
$$
\begin{aligned}
& \operatorname{Cov}\left(\sum_{t=1}^n Y_{\xi}^t, \sum_{t=1}^n\left(1-X_{\xi}^{t-1}\right)\right)=-\operatorname{Cov}\left(\sum_{t=1}^n Y_{\xi}^t, \sum_{t=1}^n X_{\xi}^{t-1}\right) \\
&=-\frac{\alpha_{\xi} \beta_{\xi}}{\left(\alpha_{\xi}+\beta_{\xi}\right)^2}\left[n+\frac{2 n\left(1-\alpha_{\xi}-\beta_{\xi}\right)}{\alpha_{\xi}+\beta_{\xi}}\right]-\frac{n^2 \alpha_{\xi}^2\left(1-\beta_{\xi}\right)}{\left(\alpha_{\xi}+\beta_{\xi}\right)^2}+\frac{2 n \alpha_{\xi}\left[\beta_{\xi}\left(1-\alpha_{\xi}-\beta_{\xi}\right)+\alpha_{\xi}\right]}{\left(\alpha_{\xi}+\beta_{\xi}\right)^2}\\
&\quad+\frac{(n-1)^2 \alpha_{\xi}^2\left(1-\beta_{\xi}\right)}{\left(\alpha_{\xi}+\beta_{\xi}\right)^2}+\frac{2 n\left(1-\beta_{\xi}\right) \alpha_{\xi} \beta_{\xi}}{\left(\alpha_{\xi}+\beta_{\xi}\right)^3}+O_l(1)\\
&=\frac{n \alpha_{\xi} \beta_{\xi}\left[2 \alpha_{\xi}\left(1-\beta_{\xi}\right)+\alpha_{\xi}+\beta_{\xi}-2 \beta_{\xi}^2\right]}{\left(\alpha_{\xi}+\beta_{\xi}\right)^3} +O_l(1).
\end{aligned}
$$

This proves the lemma.
\end{proof}

\begin{lemma}
    (Bias of $\widehat{\alpha}_{\xi}$ and $\widehat{\beta}_{\xi}$ ) Let conditions C1, C2 and the assumptions of Proposition 2 hold. We have
$$
\begin{gathered}
E \widehat{\alpha}_{\xi}-\alpha_{\xi}=-\frac{\alpha_{\xi}\left[2 \alpha_{\xi}\left(1-\beta_{\xi}\right)+\alpha_{\xi}+\beta_{\xi}-2 \beta_{\xi}^2\right]}{n\left(\alpha_{\xi}+\beta_{\xi}\right) \beta_{\xi}}+\frac{R_{\xi}^{(1)}}{n}+O\left(n^{-2}\right), \\
E \widehat{\beta}_{\xi}-\beta_{\xi}=\frac{\beta_{\xi}\left[2 \alpha_{\xi}\left(1-\beta_{\xi}\right)+\alpha_{\xi}+\beta_{\xi}-2 \beta_{\xi}^2\right]}{n\left(\alpha_{\xi}+\beta_{\xi}\right) \alpha_{\xi}}+\frac{R_{\xi}^{(2)}}{n}+O\left(n^{-2}\right),
\end{gathered}
$$
where $R_{\xi}^{(1)}$ and $R_{\xi}^{(2)}$ are constants such that when $n$ is large enough we have $0 \leq R_{\xi}^{(1)}, R_{\xi}^{(2)} \leq$ $R_l$ for some constant $R_l$ and all $\xi \in \mathcal{E}$.
\end{lemma}
\begin{proof}

Denote $\mathcal{I}:=$ $I\left(\left|n^{-1} \sum_{t=1}^n X_{\xi}^{t-1}-\pi_{\xi}\right| \leq\left(1-\pi_{\xi}\right) / 2,1 \leq \xi \leq p\right)$. From Lemma 2 we have, under Condition C2, $\mathcal{I}$ holds with probability larger than $1-O\left(n^{-2}\right)$. Expanding $\frac{1}{1-n^{-1} \sum_{t=1}^n X_{\xi}^{t-1}}$ around $\frac{1}{1-\pi_{\xi}}$, we have

$$
\begin{aligned}
E \hat{\alpha}_{\xi} & =E \hat{\alpha}_{\xi} \mathcal{I}+E \hat{\alpha}_{\xi}(1-\mathcal{I}) \\
& =E \frac{1}{n} \sum_{t=1}^n X_{\xi}^t\left(1-X_{\xi}^{t-1}\right)\left[\frac{1}{1-\pi_{\xi}}+\frac{\left(n^{-1} \sum_{t=1}^n X_{\xi}^{t-1}-\pi_{\xi}\right)}{\left(1-\pi_{\xi}\right)^2}\right] \mathcal{I}+\frac{R_{\xi}^{(1)}}{n}+E \hat{\alpha}_{\xi}(1-\mathcal{I}) \\
& =\alpha_{\xi}+\frac{\operatorname{Cov}\left(\sum_{t=1}^n Y_{\xi}^t, \sum_{t=1}^n X_{\xi}^t\right)}{n^2\left(1-\pi_{\xi}\right)^2}+\frac{R_{\xi}^{(1)}}{n}+O\left(n^{-2}\right) \\
& =\alpha_{\xi}-\frac{\alpha_{\xi}\left[2 \alpha_{\xi}\left(1-\beta_{\xi}\right)+\alpha_{\xi}+\beta_{\xi}-2 \beta_{\xi}^2\right]}{n\left(\alpha_{\xi}+\beta_{\xi}\right) \beta_{\xi}}+\frac{R_{\xi}^{(1)}}{n}+O\left(n^{-2}\right)
\end{aligned}
$$
where $R_{\xi}^{(1)}:=E \sum_{t=1}^n X_{\xi}^t\left(1-X_{\xi}^{t-1}\right)\left(\sum_{k=2}^{\infty} \frac{\left(n^{-1} \sum_{t=1}^n X_{\xi}^{t-1}-\pi_{L j}\right)^k}{\left(1-\pi_{\xi}\right)^{k+1}}\right) \mathcal{I}$. 

To bound $R_{\xi}^{(1)}$, by Taylor series with Lagrange remainder we have there exist $r_{\xi}^t \in\left[n^{-1} \sum_{t=1}^n X_{\xi}^{t-1}, \pi_{\xi}\right]$ such that
$$
R_{\xi}^{(1)}=E \sum_{t=1}^n X_{\xi}^t\left(1-X_{\xi}^{t-1}\right)\left(\frac{\left(n^{-1} \sum_{t=1}^n X_{\xi}^{t-1}-\pi_{\xi}\right)^2}{\left(1-r_{\xi}^t\right)^3}\right) \mathcal{I}>0 .
$$
Then we have 
$$
\begin{aligned}
R_{\xi}^{(1)} & \leq E \sum_{t=1}^n\left(\sum_{k=2}^{\infty} \frac{\left|n^{-1} \sum_{t=1}^n X_{\xi}^{t-1}-\pi_{\xi}\right|^k}{\left(1-\pi_{\xi}\right)^{k+1}}\right) \mathcal{I} \\
& \leq E \sum_{t=1}^n\left(\left(n^{-1} \sum_{t=1}^n X_{\xi}^{t-1}-\pi_{\xi}\right)^2 \sum_{k=0}^{\infty} \frac{1}{\left(1-\pi_{\xi}\right)^3 2^k}\right) \\
& = \operatorname{Var}\left(\frac{1}{\sqrt{n}} \sum_{t=1}^n X_{\xi}^{t-1}\right) \frac{2}{\left(1-\pi_{\xi}\right)^3} \\
& =\frac{2}{\left(1-\pi_{\xi}\right)^3} \cdot \frac{\alpha_{\xi} \beta_{\xi}}{\left(\alpha_{\xi}+\beta_{\xi}\right)^2}\left[1+\frac{2}{n} \sum_{h=1}^{n-1}(n-h)\left(1-\alpha_{\xi}-\beta_{\xi}\right)^h\right] \\
& =\frac{2}{\left(1-\pi_{\xi}\right)^4 \pi_{\xi}} \cdot \frac{2-\alpha_{\xi}-\beta_{\xi}}{\alpha_{\xi}+\beta_{\xi}}+O\left(n^{-1}\right) .
\end{aligned}
$$

By condition C1, we conclude that there exists a constant $R_l$ such that $R_{\xi}^{(1)} \leq R_l$. 
Bias of $\hat{\beta}$ can be shown similarly with $R_{\xi}^{(2)}:=E \sum_{t=1}^n X_{\xi}^t\left(1-X_{\xi}^{t-1}\right)\left(\sum_{k=2}^{\infty} \frac{\left(n^{-1} \sum_{t=1}^n X_{\xi}^{t-1}-\pi_{L, j}\right)^k}{(-1)^k \pi_{\xi}^{k+1}}\right) \mathcal{I}^{\prime}$ where $\mathcal{I}^{\prime}:=I\left\{\left|n^{-1} \sum_{t=1}^n\left(1-X_{\xi}^t\right) X_{\xi}^{t-1}\right| \leq \pi_{\xi} / 2\right\}$. 
\end{proof}
\textbf{Proof of Theorem 1.}

Consider $||\hat{\mathbf{L}}_1-\mathbf{L}_1||_2 \leq ||\hat{\mathbf{L}}_1-\Tilde{\mathbf{L}}_1||_2 + ||\Tilde{\mathbf{L}}_1-\mathbf{L}_1||_2$ where $\Tilde{\mathbf{L}}_1 = \textbf{I}-\mathbf{D}_1^{-1/2}\widehat{\mathbf{A}}_1\mathbf{D}_1^{-1/2}$. Then 
\begin{align*}
    \Tilde{\mathbf{L}}_1-\mathbf{L}_1 &= \mathbf{D}_1^{-1/2}\widehat{\mathbf{A}}_1\mathbf{D}_1^{-1/2} - \mathbf{D}_1^{-1/2}\mathbf{A}_1\mathbf{D}_1^{-1/2} \\&= \sum_{\xi}(\hat{\alpha}_{\xi}-\alpha_{\xi})\frac{1}{|\xi|}\mathbf{D}_1^{-1/2}a_{\xi}a_{\xi}^T\mathbf{D}_1^{-1/2} 
    \\&= \sum_{\xi}(\hat{\alpha}_{\xi}-\mathbb{E}(\hat{\alpha}_{\xi}))\frac{1}{|\xi|}\mathbf{D}_1^{-1/2}a_{\xi}a_{\xi}^T\mathbf{D}_1^{-1/2} + \sum_{\xi}(\mathbb{E}(\hat{\alpha}_{\xi})-\alpha_{\xi})\frac{1}{|\xi|}\mathbf{D}_1^{-1/2}a_{\xi}a_{\xi}^T\mathbf{D}_1^{-1/2}
\end{align*}
For the first term, denote each matrix in the sum as $Y_{\xi}$. Then $\{Y_{\xi}\}$ are independent with mean zero. Set $a' = C_a\sqrt{\frac{\log np}{p^{K-1}n}}$ and $a = a' + C_1/n$ for some constant $C_1$. Hence, we can apply matrix Bernstein inequality to obtain

\begin{align*}
    \mathbb{P}(||\sum_{\xi}(Y_{\xi})||_2 \geq a') \leq 2p \exp(\frac{-(a')^2}{2||\sum_{\xi}\operatorname{Var}(Y_{\xi})||_2+\frac{2}{3}a\max_{\xi}||Y_{\xi}||_2})
\end{align*}
where $$\operatorname{Var}(Y_{\xi})=\mathbb{E}(Y^2_{\xi})=\operatorname{Var}(\hat{\alpha}_{\xi})\frac{(a_{\xi}\textbf{D}_1^{-1}a_{\xi})}{|\xi|^2}\mathbf{D}_1^{-1/2}a_{\xi}a_{\xi}^T\mathbf{D}_1^{-1/2}.$$

Note that for any matrix $\textbf{B}, \mathbf{D}_1^{-1 / 2} \textbf{B} \mathbf{D}_!^{-1 / 2}$ and $\mathbf{D}_1^{-1} \textbf{B}$ have the same eigenvalues, and hence, using Gerschgorin's theorem \citep{zbMATH00047363}, one has
$$
\begin{aligned}
\left\|\sum_{\xi} \operatorname{Var}\left(Y_{\xi}\right)\right\|_2 & \leq \max _{1 \leq i \leq p} \frac{1}{d_{i,1}} \sum_{j=1}^p\left(\sum_{\xi} \operatorname{Var}\left(\hat{\alpha}_{\xi}\right) \frac{\left(a_{\xi}^T \mathbf{D}_1^{-1} a_{\xi}\right)}{\left|\xi\right|^2} a_{\xi} a_{\xi}^T\right)_{\xi} \\
& =\max _{1 \leq i \leq p} \frac{1}{d_{i,1}} \sum_{\xi} \operatorname{Var}\left(\hat{\alpha}_{\xi}\right) \frac{\left(a_{\xi}^T \mathbf{D}_1^{-1} a_{\xi}\right)}{\left|\xi\right|^2}\left(a_{\xi}\right)_i \sum_{j=1}^p\left(a_{\xi}\right)_j .
\end{aligned}
$$

Similar to the proofs of Lemma 3 we can show that, when $n$ is large enough, there exists a constant $C_\sigma>(2 l)^{-1}$ and $c_\sigma:=l(1-l)$ such that $c_\sigma n^{-1} \leq \operatorname{Var}(\hat{\alpha}_{\xi}) \leq C_\sigma n^{-1}$ for any $\xi \in \mathcal{E}$. Meanwhile, there exists constants $C_d > \frac{2(K-1)}{(K-1)!}$ and $c_d < \frac{1}{2(K-1)!}$ such that $c_d p^{K-1} \leq \sum_{\xi \in \mathcal{E}}(a_{\xi})_i \leq C_d p^{K-1}. $

Given that $\alpha_{\xi} \geq l$, then $d_{i,1} = \sum_{\xi} \alpha_{\xi}(a_{\xi})_i \geq c_dp^{K-1}l$ and $a_{\xi}^T \mathbf{D}_1^{-1} a_{\xi}=\sum_j \frac{(a_{\xi})_j}{d_{j,1}} \leq \frac{|\xi|}{c_dp^{K-1}l}$, then we have
\begin{align*}
    \left\|\sum_{\xi} \operatorname{Var}\left(Y_{\xi}\right)\right\|_2 
 \leq  \frac{C_dC_{\sigma}}{c_d^2p^{K-1}l^2n}.
\end{align*}
 From Lemma 4 we know that there exists a large enough constant $C_\alpha>0$ such that $\left|E \widehat{\alpha}_{\xi}-\alpha_{\xi}\right| \leq \frac{C_\alpha}{n}$ for all $\xi \in \mathcal{E}$. Notice that when $n$ is large, $C_0\sqrt{\frac{\log n}{n}} < \frac{1}{(\log n) (\log \log n)} $. From similar result of Proposition 4, as $n,p \rightarrow \infty$, $|\hat{\alpha}_{\xi}-\alpha| \leq 6l^{-1}C_0\sqrt{\frac{\log n}{n}}$, then 
\begin{align*}
    |\hat{\alpha}_{\xi}-\mathbb{E}(\hat{\alpha}_{\xi})| \leq& |\alpha_{\xi}-\mathbb{E}(\hat{\alpha}_{\xi})| + |\hat{\alpha}_{\xi}-\alpha| \\
    \leq& \frac{C_{\alpha}}{n}+ 6l^{-1}C_0\sqrt{\frac{\log n}{n}} \\
    \leq& (6l^{-1}+C_{\alpha})\sqrt{\log n/(C_2n)}
\end{align*}
using $\frac{1}{n} \leq \sqrt{\frac{\log n}{n}}$.

Similarly, we have 
\begin{align*}
    ||Y_{\xi}||_2 &\leq |\hat{\alpha}_{\xi}-\mathbb{E}(\hat{\alpha}_{\xi})|\frac{1}{|\xi|}||\mathbf{D}_1^{-1/2}a_{\xi}a_{\xi}^T\mathbf{D}_1^{-1/2}||_2 \\
    &\leq |\hat{\alpha}_{\xi}-\mathbb{E}(\hat{\alpha}_{\xi})|\frac{|a_{\xi}^T\mathbf{D}_1^{-1}a_{\xi}|}{|\xi|}\\
    & \leq (6l^{-1}+C_{\alpha})\sqrt{\log n/(C_2n)}\frac{1}{c_dp^{K-1}l}
\end{align*}

by $\mathbf{D}_1^{-1/2}a_{\xi}a_{\xi}^T\mathbf{D}_1^{-1/2}$ is a rank-1 matrix.

For the second term, for some constant $C_1\geq 3C_dC_{\alpha}/c_dl$,
\begin{align*}
    ||\sum_{\xi}(\mathbb{E}(\hat{\alpha}_{\xi})-\alpha_{\xi})\frac{1}{|\xi|}\mathbf{D}_1^{-1/2}a_{\xi}a_{\xi}^T\mathbf{D}_1^{-1/2}||_2
     & \leq \sum_{\xi}|\mathbb{E}(\hat{\alpha}_{\xi})-\alpha|\frac{1}{|\xi|}||\mathbf{D}_1^{-1/2}a_{\xi}a_{\xi}^T\mathbf{D}_1^{-1/2}||_2\\
    & \leq  \frac{C_{\alpha}}{n}\frac{C_dp^{K-1}}{c_dp^{K-1}l} \\
    & \leq C_1/n
\end{align*}
by $2^{p-1}-1\geq 2^{p-2}$. 

Hence, 
\begin{align*}
    &\mathbb{P}(||\Tilde{\mathbf{L}}_1-\mathbf{L}_1||_2 \geq a) \\
    \leq &\mathbb{P}(||\sum_{\xi}(\hat{\alpha}_{\xi}-\mathbb{E}(\hat{\alpha}_{\xi}))\frac{1}{|\xi|}\mathbf{D}_1^{-1/2}a_{\xi}a_{\xi}^T\mathbf{D}_1^{-1/2}||_2\geq a - ||\sum_{\xi}(\mathbb{E}(\hat{\alpha}_{\xi})-\alpha_{\xi})\frac{1}{|\xi|}\mathbf{D}_1^{-1/2}a_{\xi}a_{\xi}^T\mathbf{D}_1^{-1/2}||_2)\\
    \leq &\mathbb{P}(||\sum_{\xi}(\hat{\alpha}_{\xi}-\mathbb{E}(\hat{\alpha}_{\xi}))\frac{1}{|\xi|}\mathbf{D}_1^{-1/2}a_{\xi}a_{\xi}^T\mathbf{D}_1^{-1/2}||_2\geq a - C_1/n) \\
    \leq &\mathbb{P}(||\sum_{\xi}(\hat{\alpha}_{\xi}-\mathbb{E}(\hat{\alpha}_{\xi}))\frac{1}{|\xi|}\mathbf{D}_1^{-1/2}a_{\xi}a_{\xi}^T\mathbf{D}_1^{-1/2}||_2\geq a') \\
    \leq & 2p\exp(-\frac{C_a^2\frac{\log np}{p^{K-1}n}}{2 \frac{C_dC_{\sigma}}{c_d^2p^{K-1}l^2n}+ 
    \frac{2}{3}C_a\sqrt{\frac{\log np}{p^{K-1}n}}(6l^{-1}+C_{\alpha})\sqrt{\log n/(C_2n)}\frac{1}{c_dp^{K-1}l}}). 
\end{align*}
When $2 \frac{C_dC_{\sigma}}{c_d^2p^{K-1}l^2n}> 
    \frac{2}{3}C_a\sqrt{\frac{\log np}{p^{K-1}n}}(6l^{-1}+C_{\alpha})\sqrt{\log n/(C_2n)}\frac{1}{c_dp^{K-1}l}$, for any constant $B>0$, choosing $C_a > \frac{\sqrt{4C_dC_\sigma(B+1)}}{lc_d}$, it reduces to 

\begin{align*}
    \mathbb{P}(||\Tilde{\mathbf{L}}_1-\mathbf{L}_1||_2 \geq a) 
    \leq 2p\exp(-\frac{C_a^2\frac{\log np}{p^{K-1}n}}{4\frac{C_dC_{\sigma}}{c_d^2p^{K-1}l^2n}} ) = 2p (np)^{-\frac{C_a^2c_d^2l^2}{4C_dC_{\sigma}}}\leq 2p(np)^{-(B+1)}. 
\end{align*}
When $2 \frac{C_dC_{\sigma}}{c_d^2p^{K-1}l^2n}\leq 
    \frac{2}{3}C_a\sqrt{\frac{\log np}{p^{K-1}n}}(6l^{-1}+C_{\alpha})\sqrt{\log n/(C_2n)}\frac{1}{c_dp^{K-1}l}$, choosing $C_a > \frac{4B(6l^{-1}+C_{\alpha})}{3lc_d\sqrt{C_2}} $, it reduces to 
\begin{align*}
    \mathbb{P}(||\Tilde{\mathbf{L}}_1-\mathbf{L}_1||_2 \geq a) 
    &\leq 2p\exp(-\frac{C_a^2\frac{\log np}{p^{K-1}n}}{\frac{4}{3}C_a\sqrt{\frac{\log np}{p^{K-1}n}}(6l^{-1}+C_{\alpha})\sqrt{\log n/(C_2n)}\frac{1}{c_dp^{K-1}l}} )\\&
    = 2p\exp(-\frac{3lC_ac_d\sqrt{C_2}}{4(6l^{-1}+C_{\alpha})}\sqrt{\frac{p^{K-1}\log np}{\log n}})\\&
    \leq 2p\exp(-Bp^{(K-1)/2}). 
\end{align*}
Hence, we conclude that for any $B>0$, by choosing $C_a$ to be large enough, we have 
$$ \mathbb{P}(||\Tilde{\mathbf{L}}_1-\mathbf{L}_1||_2 \geq C_a\sqrt{\frac{\log np}{p^{K-1}n}} + C_1/n) 
    \leq 2p[(np)^{-(B+1)}+\exp(-Bp^{(K-1)/2})].$$

Now, we derive the asymptotic bound for $||\hat{\mathbf{L}}_1-\Tilde{\mathbf{L}}_1||_2$. Note that 
$$
\begin{aligned}
\|\hat{\mathbf{L}}_1-\Tilde{\mathbf{L}}_1||_2 & \leq\left\|\mathbf{D}_1^{-1/2}\widehat{\mathbf{A}}_1\mathbf{D}_1^{-1/2}-\widehat{\mathbf{D}}_1^{-1/2}\widehat{\mathbf{A}}_1\widehat{\mathbf{D}}_1^{-1/2}\right\|_2 \\
& \leq\left\|\left(\mathbf{D}_1^{-1 / 2}-\widehat{\mathbf{D}}_1^{-1 / 2}\right) \widehat{\mathbf{A}}_1 \mathbf{D}_1^{-1 / 2}+\widehat{\mathbf{D}}_1^{-1 / 2} \widehat{\mathbf{A}}_1\left(\mathbf{D}_1^{-1 / 2}-\widehat{\mathbf{D}}_1^{-1 / 2}\right)\right\|_2 \\
& \leq\left\|\left(\left(\mathbf{D}_1^{-1} \widehat{\mathbf{D}}_1\right)^{1 / 2}-\textbf{I}\right)\widehat{\mathbf{D}}_1^{-1 / 2} \widehat{\mathbf{A}}_1 \widehat{\mathbf{D}}_1^{-1 / 2}\left(\widehat{\mathbf{D}}_1 \mathbf{D}_1^{-1}\right)^{1 / 2}\right\|_2\\
&+\left\|\widehat{\mathbf{D}}_1^{-1 / 2} \widehat{\mathbf{A}}_1 \widehat{\mathbf{D}}_1^{-1 / 2}\left(\left(\mathbf{D}_1^{-1} \widehat{\mathbf{D}}_1\right)^{1 / 2}-\textbf{I}\right)\right\|_2\\
& \leq\left\|\left(\mathbf{D}_1^{-1} \widehat{\mathbf{D}}_1\right)^{1 / 2}-\textbf{I}\right\|_2\left\|\left(\widehat{\mathbf{D}}_1 \mathbf{D}_1^{-1}\right)^{1 / 2}\right\|_2+\left\|\left(\mathbf{D}_1^{-1} \widehat{\mathbf{D}}_1\right)^{1 / 2}-\textbf{I}\right\|_2\\
& \leq\left\|\left(\mathbf{D}_1^{-1} \widehat{\mathbf{D}}_1\right)^{1 / 2}-\textbf{I}\right\|_2\left(1+\left\|\left(\mathbf{D}_1^{-1} \widehat{\mathbf{D}}_1\right)^{1 / 2}-\textbf{I}\right\|_2\right)+\left\|\left(\mathbf{D}_1^{-1} \widehat{\mathbf{D}}_1\right)^{1 / 2}-\textbf{I}\right\|_2
\\&\leq \max _{1 \leq i \leq n}\left|\frac{\hat{d}_{i,1}}{d_{i,1}}-1\right|(2+\max _{1 \leq i \leq n}\left|\frac{\hat{d}_{i,1}}{d_{i,1}}-1\right|)\end{aligned}
$$

In above, we use the fact that $\hat{d}_{i,1}=\sum_{\xi} \hat{\alpha}_{\xi}(a_{\xi})_i = \sum_j \sum_{\xi}\frac{\hat{\alpha}_{\xi}}{|\xi|}(a_{\xi})_i(a_{\xi})_j =\sum_j (\widehat{\mathbf{A}}_1)_{ij}$ to conclude that $\left\|\widehat{\mathbf{D}}_1^{-1 / 2} \widehat{\mathbf{A}}_1 \widehat{\mathbf{D}}_1^{-1 / 2}\right\|_2\leq 1$. Note that $\mathbf{D}_1^{-1} \widehat{\mathbf{D}}_1$ is a diagonal matrix with non-negative diagonal entries, and hence,
$$
\left\|\left(\mathbf{D}_1^{-1} \widehat{\mathbf{D}}_1\right)^{1 / 2}-\mathbf{I}\right\|_2=\max _{1 \leq i \leq n}\left|\sqrt{\frac{\hat{d}_{i,1}}{d_{i,1}}}-1\right| \leq \max _{1 \leq i \leq n}\left|\frac{\hat{d}_{i,1}}{d_{i,1}}-1\right|,
$$
where the inequality follows from the fact that $|\sqrt{x}-1| \leq|x-1|$ for all positive $x$. 

Then set $a = \frac{1}{3}(C_D\sqrt{\frac{\log np}{p^{K-1}n}}+C_1/n)$. Use $d_{i,1} \geq c_dp^{K-1}l$ and Bernstein inequality, 
$$
\begin{aligned}
    \mathbb{P}(|\hat{d}_{i,1}-d_{i,1}| >d_{i,1}a)
    \leq & \mathbb{P}(|\hat{d}_{i,1}-d_{i,1}| >c_dp^{K-1}la)\\
    \leq & \mathbb{P}(|\hat{d}_{i,1}-\mathbb{E}(\hat{d}_{i,1})| >c_dp^{K-1}la-|d_{i,1}-\mathbb{E}(\hat{d}_{i,1})|)\\
    \leq & \mathbb{P}(\sum_{\xi:i \in \xi}|\hat{\alpha}_{\xi}-\mathbb{E}(\hat{\alpha}_{\xi})| >c_dp^{K-1}la-\sum_{\xi:i \in \xi}|\alpha_{\xi}-\mathbb{E}(\hat{\alpha}_{\xi})|)\\
    \leq & \mathbb{P}(\sum_{\xi:i \in \xi}|\hat{\alpha}_{\xi}-\mathbb{E}(\hat{\alpha}_{\xi})| >\frac{c_dp^{K-1}lC_D}{3}\sqrt{\frac{\log np}{p^{K-1}n}} + \frac{p^{K-1}}{n}(\frac{c_dC_1l}{3}-C_dC_{\alpha}))\\
    \leq & \mathbb{P}(\sum_{\xi:i \in \xi}|\hat{\alpha}_{\xi}-\mathbb{E}(\hat{\alpha}_{\xi})| >\frac{c_dC_Dl}{3}\sqrt{\frac{p^{K-1}\log np}{n}})\\
    \leq & 2\exp\left(-\frac{c_d^2C_D^2l^2p^{K-1}\log np /(9n)}{2\sum_{\xi:i \in \xi}\operatorname{Var}(\hat{\alpha}_{\xi})+\frac{2}{3}(6l^{-1}+C_{\alpha})\sqrt{\log n/(C_2n)}\frac{c_dC_Dl}{3}\sqrt{\frac{p^{K-1} \log np}{n}}}\right)\\
    \leq & 2\exp\left(-\frac{c_d^2C_D^2l^2p^{K-1}\log np /(9n)}{2C_dC_{\sigma}p^{K-1}/n+\frac{2}{3}(6l^{-1}+C_{\alpha})\sqrt{\log n/(C_2n)}\frac{c_dC_Dl}{3}\sqrt{\frac{p^{K-1} \log np}{n}}}\right).
\end{aligned}
$$
When $2C_dC_{\sigma}p^{K-1}/n>\frac{2}{3}(6l^{-1}+C_{\alpha})\sqrt{\log n/(C_2n)}\frac{c_dC_Dl}{3}\sqrt{\frac{p^{K-1} \log np}{n}}$, for any $B>0$, by choosing $C_D>\frac{\sqrt{36C_dC_{\sigma}(B+1)}}{c_dC_Dl}$, it reduces to 

\begin{align*}
   \mathbb{P}(|\hat{d}_{i,1}-d_{i,1}| >d_{i,1}a)
    \leq & 2\exp\left(-\frac{c_d^2C_D^2l^2p^{K-1}\log np /(9n)}{4C_dC_{\sigma}p^{K-1}/n}\right)\\
    \leq &2\exp\left(-\frac{c_d^2C_D^2l^2\log np }{36C_dC_{\sigma}}\right)<2(pn)^{-(B+1)}.
\end{align*}
When $2C_dC_{\sigma}p^{K-1}/n\leq\frac{2}{3}(6l^{-1}+C_{\alpha})\sqrt{\log n/(C_2n)}\frac{c_dC_Dl}{3}\sqrt{\frac{p^{K-1} \log np}{n}}$, by choosing $C_D>\frac{4(6l^{-1}+C_{\alpha)}}{c_dC_Dl\sqrt{C_2}}$, it reduces to 

\begin{align*}
   \mathbb{P}(|\hat{d}_{i,1}-d_{i,1}| >d_{i,1}a)
    \leq & 2\exp\left(-\frac{c_d^2C_D^2l^2p^{K-1}\log np /(9n)}{\frac{4}{3}(6l^{-1}+C_{\alpha})\sqrt{\log n/(C_2n)}\frac{c_dC_Dl}{3}\sqrt{\frac{p^{K-1} \log np}{n}}}\right)\\
    \leq &2\exp\left(-\frac{\sqrt{C_2p^{K-1}}c_dC_Dl}{4(6l^{-1}+C_{\alpha})}\right)
    <2\exp\left(-Bp^{(K-1)/2}\right).
\end{align*}
Hence, we conclude that for any $B>0$, by choosing $C_D$ to be large enough, we have 
$$ \mathbb{P}(\max_{1\leq i\leq n}|\frac{\hat{d}_{i,1}}{d_{i,1}}-1| \geq \frac{1}{3}(C_D\sqrt{\frac{\log np}{p^{K-1}n}} + C_1/n))
    \leq 2p[(np)^{-(B+1)}+\exp(-B^{(K-1)/2})]$$
    and then by $C_D\sqrt{\frac{\log np}{p^{K-1}n}} + C_1/n <1$
    \begin{align*}
        \mathbb{P}(\|\hat{\mathbf{L}}_1-\Tilde{\mathbf{L}}_1||_2 \geq C_D\sqrt{\frac{\log np}{p^{K-1}n}} + C_1/n) 
        &\leq \mathbb{P}(3\max _{1 \leq i \leq n}\left|\frac{\hat{d}_{i,1}}{d_{i,1}}-1\right|\geq C_D\sqrt{\frac{\log np}{p^{K-1}n}} + C_1/n) 
    \\&\leq 2p[(np)^{-(B+1)}+\exp(-Bp^{(K-1)/2})]
    \end{align*}
Hence, there exists a large enough constant $C>0$ such that when $n$ and $p$ large enough, 
$$\mathbb{P}(||\hat{\mathbf{L}}_1-\mathbf{L}_1||_2 \geq C(\sqrt{\frac{\log np}{p^{K-1}n}} + 1/n)) \leq 4p[(np)^{-(B+1)}+\exp(-Bp^{(K-1)/2})]. $$

Similarly, we can show that 

$$\mathbb{P}(||\hat{\mathbf{L}}_2-\mathbf{L}_2||_2 \geq C(\sqrt{\frac{\log np}{p^{K-1}n}} + 1/n)) \leq 4p[(np)^{-(B+1)}+\exp(-Bp^{(K-1)/2})]. $$

By triangle inequality, we have 
$||\hat{\mathbf{L}}-\mathbf{L}||_2\leq ||\hat{\mathbf{L}}_1-\mathbf{L}_1||_2+||\hat{\mathbf{L}}_2-\mathbf{L}_2||_2$ and then 

$$\mathbb{P}(||\hat{\mathbf{L}}-\mathbf{L}||_2 \geq 2C(\sqrt{\frac{\log np}{p^{K-1}n}} + 1/n)) \leq 8p[(np)^{-(B+1)}+\exp(-Bp^{(K-1)/2})]. $$

By Weyl's inequality, we have 
$$\max_{1 \leq i \leq p}|\lambda_i-\hat{\lambda}_i| \leq ||\hat{\mathbf{L}}-\mathbf{L}||_2 =O_p(\sqrt{\frac{\log np}{p^{K-1}n}} + \frac{1}{n}).$$

Denoting $\lambda_i(\underset{\sim}{\hat{\mathbf{L}}}), \lambda_i(\mathbf{L})$ as the $i^{t h}$ smallest eigenvalues of $\hat{\mathbf{L}}, \mathbf{L}$ respectively, we have $\delta \leq \widetilde{\delta}=\lambda_{k+1}(\mathbf{L})-\lambda_k(\mathbf{L})$, where $\widetilde{\delta}$ is defined in A.6. Assume that $C(\sqrt{\frac{\log np}{p^{K-1}n}} + 1/n)) < \frac{\delta}{2}$. Thus
$$
\lambda_{k+1}(\hat{\mathbf{L}}) \geq \lambda_{k+1}(\mathbf{L})-\frac{\delta}{2} \geq \lambda_k(\mathbf{L})+\frac{\delta}{2}>0.
$$

By Davis-Kahan theorem\citep{yu2015useful, ghoshdastidar_consistency_2017}, we have 

$$
\|\sin \Theta(\widehat{\boldsymbol{\Gamma}}_q, \mathbf{\Gamma}_q)\|_2 \leq \frac{||\hat{\mathbf{L}}-\mathbf{L}||_2}{\delta}
$$

where $\sin \Theta(\widehat{\boldsymbol{\Gamma}}_q, \mathbf{\Gamma}_q) \in \mathbb{R}^{k \times k}$ is diagonal with entries same as the sine of the canonical angles between the subspaces $\widehat{\boldsymbol{\Gamma}}_q$ and $\mathbf{\Gamma}_q$. Let these angles be denoted as $\theta_1, \ldots, \theta_k \in\left[0, \frac{\pi}{2}\right]$ such that $\theta_1 \geq \ldots \geq \theta_k$. Then $\|\sin \Theta(\widehat{\boldsymbol{\Gamma}}_q, \mathbf{\Gamma}_q)\|_2=$ $\sin \theta_1$. On the other hand, one can see that the singular values for the matrix $\widehat{\boldsymbol{\Gamma}}_q^T \mathbf{\Gamma}_q$ are given by $\cos \theta_1, \ldots \cos \theta_k$. Thus, if $\widehat{\boldsymbol{\Gamma}}_q^T \mathbf{\Gamma}_q=\mathbf{U}_1 \mathbf{\Sigma} \mathbf{U}_2^T$ is the singular value decomposition of $\widehat{\boldsymbol{\Gamma}}_q^T \mathbf{\Gamma}_q$, then
$$
\begin{aligned}
\left\|\widehat{\boldsymbol{\Gamma}}_q-\mathbf{\Gamma}_q \mathbf{U}_2 \mathbf{U}_1^T\right\|_F^2 & =\operatorname{Trace}\left(\left(\widehat{\boldsymbol{\Gamma}}_q-\mathbf{\Gamma}_q \mathbf{U}_2 \mathbf{U}_1^T\right)^T\left(\widehat{\boldsymbol{\Gamma}}_q-\mathbf{\Gamma}_q \mathbf{U}_2 \mathbf{U}_1^T\right)\right) \\
& =2 \operatorname{Trace}\left(I-\mathbf{U}_1 \Sigma \mathbf{U}_1^T\right) \\
& =2 \sum_{i=1}^k\left(1-\cos \theta_i\right) \leq 2 \sum_{i=1}^k\left(1-\cos ^2 \theta_i\right) \leq 2 k \sin ^2 \theta_1 .
\end{aligned}
$$

Then we can conclude that for $\delta \geq 2\|L-\mathcal{L}\|_2$, with probability at least $1- 8p[(np)^{-(B+1)}+\exp(-Bp^{(K-1)/2})], $

$$\left\|\widehat{\boldsymbol{\Gamma}}_q-\mathbf{\Gamma}_q \mathbf{O}_q \right\|_F \leq \sqrt{2q}\frac{2||\hat{\mathbf{L}}-\mathbf{L}||_2}{\delta}\leq \frac{\sqrt{8q}C}{\delta}(\sqrt{\frac{\log np}{p^{K-1}n}} + \frac{1}{n}))$$

where $\mathbf{O}_q = \mathbf{U}_2 \mathbf{U}_1^T. $

\subsection*{A.8 Proof of Theorem 2}

Recall that $\boldsymbol{\Gamma}_q=\mathbf{Z}\mathbf{(Z^TZ)}^{-1/2} \mathbf{U}$ where $\mathbf{U}$ is defined as in the proof of Proposition 6. For any $1 \leq i \neq$ $j \leq n$ such that $\mathbf{z}_i \neq \mathbf{z}_j$, we need to show that $\left\|\mathbf{z}_i \mathbf{(Z^TZ)}^{-1/2} \mathbf{U} \mathbf{O}_q-\mathbf{z}_j\mathbf{(Z^TZ)}^{-1/2} \mathbf{U} \mathbf{O}_q\right\|_2=\left\|\mathbf{z}_i \mathbf{(Z^TZ)}^{-1/2} \mathbf{U}-\mathbf{z}_j \mathbf{(Z^TZ)}^{-1/2} \mathbf{U}\right\|_2$ is large
enough, so that the perturbed version (i.e. the rows of $\widehat{\boldsymbol{\Gamma}}_q$ ) is not changing the clustering structure.

Denote the $i$th row of $\boldsymbol{\Gamma}_q \mathbf{O}_q$ and $\widehat{\boldsymbol{\Gamma}}_q$ as $\gamma_i$ and $\hat{\gamma}_i$, respectively, for $i=1, \ldots, p$. Notice that $\mathbf{U}$ is orthonormal and $\mathbf{Z}^T \mathbf{Z}$ is a diagonal matrix with entries being the sizes of the $q$ partitions. Consequently, for any $\mathbf{z}_i \neq \mathbf{z}_j$, we have:
\begin{align*}
    \left\|\gamma_i-\gamma_j\right\|_2&=\left\|\mathbf{z}_i \mathbf{(Z^TZ)}^{-1/2} \mathbf{U} \mathbf{O}_q-\mathbf{z}_j\mathbf{(Z^TZ)}^{-1/2} \mathbf{U} \mathbf{O}_q\right\|_2\\&=\left\|\mathbf{z}_i\mathbf{(Z^TZ)}^{-1/2} \mathbf{U}-\mathbf{z}_j\mathbf{(Z^TZ)}^{-1/2} \mathbf{U}\right\|_2 \geq \sqrt{\frac{2}{s_{\max }}}
\end{align*}

We first show that $\mathbf{z}_i \neq \mathbf{z}_j$ implies $\widehat{\mathbf{c}}_i \neq \widehat{\mathbf{c}}_j$. Notice that $\boldsymbol{\Gamma}_q \mathbf{O}_q \in \mathcal{M}_{p, q}$. Denote $\widehat{\mathbf{C}}=$ $\left(\widehat{\mathbf{c}}_1, \cdots, \widehat{\mathbf{c}}_p\right)^{\top}$. By the definition of $\widehat{\mathbf{C}}$ we have
$$
\left\|\boldsymbol{\Gamma}_q \mathbf{O}_q-\widehat{\mathbf{C}}\right\|_F^2 \leq\left\|\widehat{\boldsymbol{\Gamma}}_q-\widehat{\mathbf{C}}\right\|_F^2+\left\|\widehat{\boldsymbol{\Gamma}}_q-\mathbf{\Gamma}_q \mathbf{O}_q\right\|_F^2 \leq 2\left\|\widehat{\boldsymbol{\Gamma}}_q-\boldsymbol{\Gamma}_q \mathbf{O}_q\right\|_F^2
$$

Suppose there exist $i, j \in\{1, \ldots, p\}$ such that $\mathbf{z}_i \neq \mathbf{z}_j$ but $\widehat{\mathbf{c}}_i=\widehat{\mathbf{c}}_j$. Notice that $\left\|\Gamma_q O_q-\widehat{C}\right\|_F^2=\sum_i\left\|\gamma_i-\widehat{\mathbf{c}}_i\right\|_2^2$, we have
\begin{align*}
    \left\|\boldsymbol{\Gamma}_q \mathbf{O}_q-\widehat{\mathbf{C}}\right\|_F^2 &\geq\left\|\mathbf{z}_i \mathbf{(Z^TZ)}^{-1/2} \mathbf{U} \mathbf{O}_q-\widehat{\mathbf{c}}_i\right\|_2^2+\left\|\mathbf{z}_j \mathbf{(Z^TZ)}^{-1/2} \mathbf{U} \mathbf{O}_q-\widehat{\mathbf{c}}_j\right\|_2^2 \\
    &\geq\left\|\mathbf{z}_i \mathbf{(Z^TZ)}^{-1/2} \mathbf{U} \mathbf{O}_q-\mathbf{z}_j \mathbf{(Z^TZ)}^{-1/2} \mathbf{U} \mathbf{O}_q\right\|_2^2 .
\end{align*}

Hence, we have:
$$
\sqrt{\frac{2}{s_{\max }}} \leq\left\|\boldsymbol{\Gamma}_q \mathbf{O}_q-\widehat{\mathbf{C}}\right\|_F \leq \sqrt{2}\left\|\widehat{\boldsymbol{\Gamma}}_q-\boldsymbol{\Gamma}_q \mathbf{O}_q\right\|_F \leq \frac{4\sqrt{q}C}{\delta}(\sqrt{\frac{\log np}{p^{K-1}n}} + \frac{1}{n}))
$$

We have reach a contradictory with (3). Therefore we conclude that $\widehat{\mathbf{c}}_i \neq \widehat{\mathbf{c}}_j$.

Next we show that if $\mathbf{z}_i=\mathbf{z}_j$ we must have $\widehat{\mathbf{c}}_i=\widehat{\mathbf{c}}_j$. Assume that there exist $1 \leq$ $i \neq j \leq p$ such that $\mathbf{z}_i=\mathbf{z}_j$ and $\widehat{\mathbf{c}}_i \neq \widehat{\mathbf{c}}_j$. Notice that from the previous conclusion (i.e., that different $z_i$ implies different $\widehat{\mathbf{c}}_i$ ), since there are $q$ distinct rows in $\mathbf{Z}$, there are correspondingly $q$ different rows in $\widehat{\mathbf{C}}$. Consequently for any $\mathbf{z}_i=\mathbf{z}_j$, if $\widehat{\mathbf{c}}_i \neq \widehat{\mathbf{c}}_j$ there must exist a $k \neq i, j$ such that $\mathbf{z}_i=\mathbf{z}_j \neq \mathbf{z}_k$ and $\widehat{\mathbf{c}}_j=\widehat{\mathbf{c}}_k$. Let $\widehat{\mathbf{C}}^*$ be $\widehat{\mathbf{C}}$ with the $j$ th row replaced by $\widehat{\mathbf{c}}_i$. We have
$$
\begin{aligned}
& \left\|\widehat{\boldsymbol{\Gamma}}_q-\widehat{\mathbf{C}}^*\right\|_F^2-\left\|\widehat{\boldsymbol{\Gamma}}_q-\widehat{\mathbf{C}}\right\|_F^2 \\
= & \left\|\widehat{\gamma}_j-\widehat{\mathbf{c}}_i\right\|_2^2-\left\|\widehat{\gamma}_j-\widehat{\mathbf{c}}_k\right\|_2^2 \\
= & \left\|\widehat{\gamma}_j-\gamma_j+\gamma_i-\widehat{\mathbf{c}}_i\right\|_2^2-\left\|\widehat{\gamma}_j-\gamma_j+\gamma_i-\gamma_k+\gamma_k-\widehat{\mathbf{c}}_k\right\|_2^2 \\
\leq & \left\|\widehat{\gamma}_j-\gamma_j+\gamma_i-\widehat{\mathbf{c}}_i\right\|_2^2+\left\|\widehat{\gamma}_j-\gamma_j+\gamma_k-\widehat{\mathbf{c}}_k\right\|_2^2-\left\|\gamma_i-\gamma_k\right\|_2^2 \\
\leq & 4\left\|\widehat{\gamma}_j-\gamma_j\right\|_2^2+2(\left\|\gamma_i-\widehat{\mathbf{c}}_i\right\|_2^2+\left\|\gamma_k-\widehat{\mathbf{c}}_k\right\|_2^2)-\left\|\gamma_i-\gamma_k\right\|_2^2 \\
\leq & 4\left\|\widehat{\boldsymbol{\Gamma}}_q-\mathbf{\Gamma}_q \mathbf{O}_q\right\|_F^2+2\left\|\mathbf{\Gamma}_q \mathbf{O}_q-\widehat{\mathbf{C}}\right\|_F^2-\frac{2}{s_{\max }} \\
\leq & 6\left\{ \frac{\sqrt{8q}C}{\delta}(\sqrt{\frac{\log np}{p^{K-1}n}} + \frac{1}{n}))\right\}^2-\frac{2}{s_{\max }} \\
< & 0 
\end{aligned}
$$
using $||a+b||_2^2 \leq 2||a||_2^2 + 2||b||_2^2$ for $a,b$ be vectors.

Again, we reach a contradiction and so we conclude that if $\mathbf{z}_i=\mathbf{z}_j$ we must have $\widehat{\mathbf{c}}_i=\widehat{\mathbf{c}}_j$.

\subsection*{A.9 Proof of Theorem 4}

Note that from Theorem 2, we have the memberships can be recovered with probability tending to 1 , i,e, $P(\widehat{\psi} \neq \psi) \rightarrow 0$. On the other hand, given $\widehat{\psi}=\psi$, we have, the log likelihood function of $\left(\theta_{\zeta}, \eta_{\zeta}\right), \zeta \in \mathbf{\Psi}$, is
$$
\begin{aligned}
l\left(\left\{\theta_{\zeta}, \eta_{\zeta}\right\} ; \psi\right)= & \sum_{\xi \in S_{\zeta}} \sum_{t=1}^n\left\{X_{\xi}^t\left(1-X_{\xi}^{t-1}\right) \log \theta_{\zeta}+\left(1-X_{\xi}^t\right)\left(1-X_{\xi}^{t-1}\right) \log \left(1-\theta_{\zeta}\right)\right. \\
& \left.+\left(1-X_{\xi}^t\right) X_{\xi}^{t-1} \log \eta_{\zeta}+X_{\xi}^t X_{\xi}^{t-1} \log \left(1-\eta_{\zeta}\right)\right\} .
\end{aligned}
$$

Using the same arguments as in the proof of Proposition 7, we can conclude that when $\widehat{\psi}=\psi, \sqrt{n} \mathbf{N}_{J_1, J_2}^{\frac{1}{2}}\left(\widehat{\boldsymbol{\Phi}}_{\mathbf{\Psi}_1, \mathbf{\Psi}_2}-\boldsymbol{\Phi}_{\mathbf{\Psi}_1, \mathbf{\Psi}_2}\right) \rightarrow N\left(\mathbf{0}, \widetilde{\boldsymbol{\Sigma}}_{\mathbf{\Psi}_1, \mathbf{\Psi}_2}\right)$. Let $\mathbf{Y} \sim N\left(\mathbf{0}, \widetilde{\boldsymbol{\Sigma}}_{\mathbf{\Psi}_1, \mathbf{\Psi}_2}\right)$. For any $\mathcal{Y} \subset \mathcal{R}^{m_1+m_2}$, let $\boldsymbol{\Phi}(\mathcal{Y}):=P(\mathbf{Y} \in \mathcal{Y})$, we have:
$$
\begin{aligned}
& \left|P\left(\sqrt{n} \mathbf{N}_{\mathbf{\Psi}_1, \mathbf{\Psi}_2}^{\frac{1}{2}}\left(\widehat{\boldsymbol{\Phi}}_{\mathbf{\Psi}_1, \mathbf{\Psi}_2}-\boldsymbol{\Phi}_{\mathbf{\Psi}_1, \mathbf{\Psi}_2}\right) \in \mathcal{Y}\right)-\boldsymbol{\Phi}(\mathcal{Y})\right| \\
\leq & P(\widehat{\psi} \neq \psi)+\left|P\left(\left.\sqrt{n} \mathbf{N}_{\mathbf{\Psi}_1, \mathbf{\Psi}_2}^{\frac{1}{2}}\left(\widehat{\boldsymbol{\Phi}}_{\mathbf{\Psi}_1, \mathbf{\Psi}_2}-\boldsymbol{\Phi}_{\mathbf{\Psi}_1, \mathbf{\Psi}_2}\right) \in \mathcal{Y} \right\rvert\, \widehat{\psi}=\psi\right)-\boldsymbol{\Phi}(\mathcal{Y})\right| \\
= & o(1) .
\end{aligned}
$$

\subsection*{A.10 Proof of Theorem 5}

Without loss of generality, we consider the case where $\tau \in\left[n_0, \tau_0\right]$, as the convergence rate for $\tau \in\left[\tau_0, n-n_0\right]$ can be similarly derived. The idea is to break the time interval $\left[n_0, \tau_0\right]$ into two consecutive parts: $\left[n_0, \tau_{n, p}\right]$ and $\left[\tau_{n, p}, \tau_0\right]$, where $\tau_{n, p}=\left\lfloor\tau_0-\kappa n \Delta_F^{-2}\left[\frac{\log (n p)}{n}+\sqrt{\frac{\log (n p)}{n p^2}}\right]\right\rfloor$ for some large enough $\kappa>0$. Here $\lfloor\cdot\rfloor$ denotes the least integer function. 

\subsubsection*{A.10.1 Change point estimation with $\psi^{1, \tau_0}=\psi^{\tau_0+1, n}=\psi$}

We first consider the case where the membership structures remain unchanged, while the connectivity matrices before/after the change point are different. Specifically, we assume that $\psi^{1, \tau_0}=\psi^{\tau_0+1, n}=\psi$ for some $\psi$, and $\left(\theta_{1, \boldsymbol{c}}, \eta_{1, \boldsymbol{c}}\right) \neq\left(\theta_{2, \boldsymbol{c}}, \eta_{2, \boldsymbol{c}}\right)$ for some $1 \leq k \leq l \leq q$. For brevity, we shall be using the notations $S_{\boldsymbol{c}}, s_k, s_{\min }$ and $n_{\boldsymbol{c}}$ defined as in Section 3 , and introduce some new notations as follows:

Define

$$
\theta^{\tau}_{2, \boldsymbol{c}}=\frac{\frac{\tau_0-\tau}{n-\tau} \frac{\theta_{1, \boldsymbol{c}} \eta_{1, \boldsymbol{c}}}{\theta_{1, \boldsymbol{c}}+\eta_{1, \boldsymbol{c}}}+\frac{n-\tau_0}{n-\tau} \frac{\theta_{2, \boldsymbol{c}} \eta_{2, \boldsymbol{c}}}{\theta_{2, \boldsymbol{c}}+\eta_{2, \boldsymbol{c}}}}{\frac{\tau_0-\tau}{n-\tau} \frac{\eta_{1, \boldsymbol{c}}}{\theta_{1, \boldsymbol{c}}+\eta_{1, \boldsymbol{c}}}+\frac{n-\tau_0}{n-\tau} \frac{\eta_{2, \boldsymbol{c}}}{\theta_{2, \boldsymbol{c}}+\eta_{2, \boldsymbol{c}}}}, \quad \eta^{\tau}_{2, \boldsymbol{c}}=\frac{\frac{\tau_0-\tau}{n-\tau} \frac{\theta_{1, \boldsymbol{c}} \eta_{1, \boldsymbol{c}}}{\theta_{1, \boldsymbol{c}}+\eta_{1, \boldsymbol{c}}}+\frac{n-\tau_0}{n-\tau} \frac{\theta_{2, \boldsymbol{c}} \eta_{2, \boldsymbol{c}}}{\theta_{2, \boldsymbol{c}}+\eta_{2, \boldsymbol{c}}}}{\frac{\tau_0-\tau}{n-\tau} \frac{\theta_{1, \boldsymbol{c}}}{\theta_{1, \boldsymbol{c}}+\eta_{1, \boldsymbol{c}}}+\frac{n-\tau_0}{n-\tau} \frac{\theta_{2, \boldsymbol{c}}}{\theta_{2, \boldsymbol{c}}+\eta_{2, \boldsymbol{c}}}} .
$$
Clearly when $\tau=\tau_0$ we have $\theta_{2, \boldsymbol{c}}^{\tau_0}=\theta_{2, \boldsymbol{c}}$, $\eta_{2, \boldsymbol{c}}^{\tau_0}=\eta_{2, \boldsymbol{c}}$ and 

$$\theta^{\tau}_{2, \boldsymbol{c}} - \theta_{1, \boldsymbol{c}} = \frac{\frac{n-\tau_0}{n-\tau} \frac{\eta_{2, \boldsymbol{c}}}{\theta_{2, \boldsymbol{c}}+\eta_{2, \boldsymbol{c}}}(\theta_{2, \boldsymbol{c}} -\theta_{1, \boldsymbol{c}})}{\frac{\tau_0-\tau}{n-\tau} \frac{\eta_{1, \boldsymbol{c}}}{\theta_{1, \boldsymbol{c}}+\eta_{1, \boldsymbol{c}}}+\frac{n-\tau_0}{n-\tau} \frac{\eta_{2, \boldsymbol{c}}}{\theta_{2, \boldsymbol{c}}+\eta_{2, \boldsymbol{c}}}}, \quad 
\eta^{\tau}_{2, \boldsymbol{c}} - \eta_{1, \boldsymbol{c}} = \frac{\frac{n-\tau_0}{n-\tau} \frac{\theta_{2, \boldsymbol{c}} }{\theta_{2, \boldsymbol{c}}+\eta_{2, \boldsymbol{c}}}(\eta_{2, \boldsymbol{c}} - \eta_{1, \boldsymbol{c}})}{\frac{\tau_0-\tau}{n-\tau} \frac{\theta_{1, \boldsymbol{c}}}{\theta_{1, \boldsymbol{c}}+\eta_{1, \boldsymbol{c}}}+\frac{n-\tau_0}{n-\tau} \frac{\theta_{2, \boldsymbol{c}}}{\theta_{2, \boldsymbol{c}}+\eta_{2, \boldsymbol{c}}}}$$
$$\theta^{\tau}_{2, \boldsymbol{c}} - \theta_{2, \boldsymbol{c}} = \frac{\frac{\tau_0-\tau}{n-\tau} \frac{\eta_{1, \boldsymbol{c}}}{\theta_{1, \boldsymbol{c}}+\eta_{1, \boldsymbol{c}}}(\theta_{1, \boldsymbol{c}} -\theta_{2, \boldsymbol{c}})}{\frac{\tau_0-\tau}{n-\tau} \frac{\eta_{1, \boldsymbol{c}}}{\theta_{1, \boldsymbol{c}}+\eta_{1, \boldsymbol{c}}}+\frac{n-\tau_0}{n-\tau} \frac{\eta_{2, \boldsymbol{c}}}{\theta_{2, \boldsymbol{c}}+\eta_{2, \boldsymbol{c}}}}, \quad 
\eta^{\tau}_{2, \boldsymbol{c}} - \eta_{2, \boldsymbol{c}} = \frac{\frac{\tau_0-\tau}{n-\tau} \frac{\theta_{1, \boldsymbol{c}}}{\theta_{1, \boldsymbol{c}}}(\eta_{1, \boldsymbol{c}} - \eta_{2, \boldsymbol{c}})}{\frac{\tau_0-\tau}{n-\tau} \frac{\theta_{1, \boldsymbol{c}}}{\theta_{1, \boldsymbol{c}}+\eta_{1, \boldsymbol{c}}}+\frac{n-\tau_0}{n-\tau} \frac{\theta_{2, \boldsymbol{c}}}{\theta_{2, \boldsymbol{c}}+\eta_{2, \boldsymbol{c}}}}.$$

Correspondingly, we denote the MLEs as

$$
\begin{aligned}
& \widehat{\theta}_{1, \boldsymbol{c}}^\tau=\sum_{\xi \in \widehat{S}_{1, \boldsymbol{c}}^{\tau}} \sum_{t=1}^\tau X_{\xi}^t\left(1-X_{\xi}^{t-1}\right) / \sum_{\xi \in \widehat{S}_{1, \boldsymbol{c}}^\tau} \sum_{t=1}^\tau\left(1-X_{\xi}^{t-1}\right), \\
& \widehat{\eta}_{1, \boldsymbol{c}}^\tau=\sum_{\xi \in \widehat{S}_{1, \boldsymbol{c}}^\tau} \sum_{t=1}^\tau\left(1-X_{\xi}^t\right) X_{\xi}^{t-1} / \sum_{\xi \in \widehat{S}_{1, \boldsymbol{c}}^\tau} \sum_{t=1}^\tau X_{\xi}^{t-1}, \\
& \widehat{\theta}_{2, \boldsymbol{c}}^\tau=\sum_{\xi \in \widehat{S}_{2, \boldsymbol{c}}^\tau} \sum_{t=\tau+1}^n X_{\xi}^t\left(1-X_{\xi}^{t-1}\right) / \sum_{\xi \in \widehat{S}_{2, \boldsymbol{c}}^\tau} \sum_{t=\tau+1}^n\left(1-X_{\xi}^{t-1}\right), \\
& \widehat{\eta}_{2, \boldsymbol{c}}^\tau=\sum_{\xi \in \widehat{S}_{2, \boldsymbol{c}}^\tau} \sum_{t=\tau+1}^n\left(1-X_{\xi}^t\right) X_{\xi}^{t-1} / \sum_{\xi \in \widehat{S}_{2, \boldsymbol{c}}^\tau} \sum_{t=\tau+1}^n X_{\xi}^{t-1}
\end{aligned}
$$

where $\widehat{S}_{1, \boldsymbol{c}}^\tau$ and $\widehat{S}_{2, \boldsymbol{c}}^\tau$ are defined in a similar way to $\widehat{S}_{\boldsymbol{c}}$ (cf. Section 3.2.3), based on the estimated memberships $\widehat{\psi}^{1, \tau}$ and $\widehat{\psi}^{\tau+1, n}$, respectively.

Recall that
$$
\begin{aligned}
& l\left(\left\{\theta_{\boldsymbol{c}}, \eta_{\boldsymbol{c}}\right\} ; \psi^{1,n}\right)=\sum_{\boldsymbol{c}   \in \mathbf{\Psi}} \sum_{\xi \in S_{\boldsymbol{c}}} \sum_{t=1}^n\left\{X_{\xi}^t\left(1-X_{\xi}^{t-1}\right) \log \theta_{\boldsymbol{c}}\right. \\
& \left.\quad+\left(1-X_{\xi}^t\right)\left(1-X_{\xi}^{t-1}\right) \log \left(1-\theta_{\boldsymbol{c}}\right)+\left(1-X_{\xi}^t\right) X_{\xi}^{t-1} \log \eta_{\boldsymbol{c}}+X_{\xi}^t X_{\xi}^{t-1} \log \left(1-\eta_{\boldsymbol{c}}\right)\right\} \\
& = \sum_{\boldsymbol{c}   \in \mathbf{\Psi}} \sum_{\xi \in S_{\boldsymbol{c}}}g_{\xi}(\theta, \eta ; 1,n)
\end{aligned}
$$
where 
$$
\begin{aligned}
 g_{\xi}(\theta, \eta ; \tau_1, \tau_2)&=\sum_{t=\tau_1}^{\tau_2}\left\{X_{\xi}^t\left(1-X_{\xi}^{t-1}\right) \log \theta +\left(1-X_{\xi}^t\right)\left(1-X_{\xi}^{t-1}\right) \log (1-\theta)\right.\\&\left. +\left(1-X_{\xi}^t\right) X_{\xi}^{t-1} \log \eta+X_{\xi}^t X_{\xi}^{t-1} \log (1-\eta)\right\}. 
\end{aligned}
$$
Note that for any $\xi \in S_{\boldsymbol{c}}$ and $t \leq \tau_0, \mathbb{E} X_{\xi}^t\left(1-X_{\xi}^{t-1}\right)=\mathbb{E} X_{\xi}^{t-1}\left(1-X_{\xi}^t\right)=\frac{\theta_{1, \boldsymbol{c}} \eta_{1, \boldsymbol{c}}}{\theta_{1, \boldsymbol{c}}+\eta_{1, \boldsymbol{c}}}, \mathbb{E}(1-\left.X_{\xi}^t\right)\left(1-X_{\xi}^{t-1}\right)=\frac{\left(1-\theta_{1, \boldsymbol{c}}\right) \eta_{1, \boldsymbol{c}}}{\theta_{1, \boldsymbol{c}}+\eta_{1, \boldsymbol{c}}}$, and $\mathbb{E} X_{\xi}^t X_{\xi}^{t-1}=\frac{\left(1-\eta_{1, \boldsymbol{c}}\right) \theta_{1, \boldsymbol{c}}}{\theta_{1, \boldsymbol{c}}+\eta_{1, \boldsymbol{c}}}$, then by the 2nd order Taylor expansion and the fact that the partial derivative of the expected likelihood evaluated at the true values equals zero we have, there exist $\theta^* \in [\theta_2, \theta_1], \eta_* \in [\eta_2, \eta_1]$ such that 
\begin{align}
    &\mathbb{E}g_{\xi}(\theta_1, \eta_1 ; \tau_1, \tau_2)- \mathbb{E}g_{\xi}(\theta_2, \eta_2 ; \tau_1, \tau_2) \nonumber\\\leq& \frac{1}{2} \left(\tau_2-\tau_1+1\right)\left\{\mathbb{E}X_{\xi}^t\left(1-X_{\xi}^{t-1}\right) \left(\frac{\theta_{2}-\theta_{1}}{\theta^*}\right)^2+\mathbb{E}\left(1-X_{\xi}^t\right)\left(1-X_{\xi}^{t-1}\right)\left(\frac{\theta_{2}-\theta_{1}}{1-\theta^*}\right)^2\right. \nonumber\\
 +&\left.\mathbb{E}\left(1-X_{\xi}^t\right) X_{\xi}^{t-1}\left(\frac{\eta_{2}-\eta_{1}}{\eta^*}\right)^2+\mathbb{E}X_{\xi}^t X_{\xi}^{t-1}\left(\frac{\eta_{2}-\eta_{1}}{1-\eta^*}\right)^2\right\}  \nonumber\\
 \leq& \frac{1}{2} \left(\tau_2-\tau_1+1\right)\left\{\frac{\theta_{1, \boldsymbol{c}} \eta_{1, \boldsymbol{c}}}{\theta_{1, \boldsymbol{c}}+\eta_{1, \boldsymbol{c}}} \left(\frac{\theta_{2}-\theta_{1}}{\theta^*}\right)^2+\frac{\left(1-\theta_{1, \boldsymbol{c}}\right) \eta_{1, \boldsymbol{c}}}{\theta_{1, \boldsymbol{c}}+\eta_{1, \boldsymbol{c}}}\left(\frac{\theta_{2}-\theta_{1}}{1-\theta^*}\right)^2\right. \nonumber\\
 +&\left.\frac{\theta_{1, \boldsymbol{c}} \eta_{1, \boldsymbol{c}}}{\theta_{1, \boldsymbol{c}}+\eta_{1, \boldsymbol{c}}}\left(\frac{\eta_{2}-\eta_{1}}{\eta^*}\right)^2+\frac{\left(1-\eta_{1, \boldsymbol{c}}\right) \theta_{1, \boldsymbol{c}}}{\theta_{1, \boldsymbol{c}}+\eta_{1, \boldsymbol{c}}}\left(\frac{\eta_{2}-\eta_{1}}{1-\eta^*}\right)^2\right\}  .
\end{align}
Similar to Lemma 2, for any constant $B$, there are exist a large enough constant $B_1$, such that with probability at least $1-O((np)^{-B})$, 
\begin{align}
    &|\sum_{\boldsymbol{c}   \in \mathbf{\Psi}} \sum_{\xi \in S_{\boldsymbol{c}}} g_{\xi}(\theta, \eta ; \tau_1, \tau_2)- \sum_{\boldsymbol{c}   \in \mathbf{\Psi}} \sum_{\xi \in S_{\boldsymbol{c}}} \mathbb{E}g_{\xi}(\theta, \eta ; \tau_1, \tau_2)| \nonumber\\\leq & B_1 \left(\tau_2-\tau_1+1\right)p^{K-1} \sqrt{\frac{\log(np)}{\left(\tau_2-\tau_1+1\right)p^{K-1}}}
\end{align}
Denote
$$
\begin{aligned}
 \mathbb{M}_n(\tau)&:=l\left(\left\{\widehat{\theta}_{1, \boldsymbol{c}}^\tau, \widehat{\eta}_{1, \boldsymbol{c}}^\tau\right\} ; \widehat{\psi}^{1, \tau}\right)+l\left(\left\{\widehat{\theta}_{2, \boldsymbol{c}}^\tau, \widehat{\eta}_{2, \boldsymbol{c}}^\tau\right\} ; \widehat{\psi}^{\tau+1, n}\right) \\
 \mathbb{M}(\tau)&:=\mathbb{E}l\left(\left\{\theta_{1, \boldsymbol{c}}, \eta_{1, \boldsymbol{c}}\right\} ; \psi^{1, \tau}\right)+\mathbb{E}l\left(\left\{\theta^{\tau}_{2, \boldsymbol{c}}, \eta^{\tau}_{2, \boldsymbol{c}}\right\} ; \psi^{\tau+1, n}\right) .
\end{aligned}
$$
Notice that $\hat{\tau}=\operatorname{argmax}_{\tau \in\left[n_0, \tau_0\right]} \mathbb{M}_n(\tau)$ and $\tau_0=\operatorname{argmax}_{\tau \in\left[n_0, n-n_0\right]} \mathbb{M}(\tau)$. 

For any $0<\varepsilon \leqslant \tau_0-\tau_{n, p}$,
\begin{align}
 &\mathbb{P}\left(\tau_0-\tau>\varepsilon\right) \nonumber\\
 =&\mathbb{P}\left(\underset{\tau \in\left[n_0, \tau_0\right]}{\operatorname{argmax}}  \mathbb{M}_n(\tau)=\underset{\tau \in\left[n_0, \tau_0-\varepsilon\right]}{\operatorname{argmax}} \mathbb{M}_n(\tau)\right) \nonumber\\
 =&\mathbb{P}\left(\sup_{\tau \in\left[n, \tau_0-\varepsilon\right]} \mathbb{M}_n(\tau) \geqslant \sup _{\tau \in\left[\tau_0-\varepsilon, \tau\right]} \mathbb{M}_n(\tau)\right) \nonumber\\
 \leq& \mathbb{P}\left(\sup_{\tau \in[n, \tau_0-\varepsilon]} \mathbb{M}_n(\tau) - \mathbb{M}_n\left(\tau_0\right)\geq 0\right) \nonumber\\
 = &\mathbb{P}\left(\sup_{\tau \in\left[\tau_{n,p}, \tau_0-\varepsilon\right]} \mathbb{M}_n(\tau)-\mathbb{M}_n\left(\tau_0\right) \geq 0\right) + \mathbb{P}\left(\sup_{\tau \in\left[n_0, \tau_{n,p}\right]} \mathbb{M}_n(\tau)-\mathbb{M}_n\left(\tau_0\right) \geq 0\right) \nonumber\\
 \leq &  \mathbb{P}\left(\sup _{\tau \in\left[\tau_{n, p}, \tau_0-\epsilon\right]}\left[\left(\mathbb{M}_n(\tau)-\mathbb{M}(\tau)-\mathbb{M}_n\left(\tau_0\right)+\mathbb{M}\left(\tau_0\right)\right)-\left(\mathbb{M}\left(\tau_0\right)-\mathbb{M}(\tau)\right)\right] \geq 0\right) \nonumber\\
 &+ \mathbb{P}\left(\sup_{\tau \in\left[n_0, \tau_{n,p}\right]} \mathbb{M}_n(\tau)-\mathbb{M}_n\left(\tau_0\right) \geq 0\right) \nonumber\\
\leq &  \mathbb{P}\left(\sup _{\tau \in\left[\tau_{n, p}, \tau_0-\epsilon\right]}\left|\mathbb{M}_n(\tau)-\mathbb{M}(\tau)-\mathbb{M}_n\left(\tau_0\right)+\mathbb{M}\left(\tau_0\right)\right| \geq \sup _{\tau \in\left[\tau_{n, p}, \tau_0-\epsilon\right]} \left(\mathbb{M}\left(\tau_0\right)-\mathbb{M}(\tau)\right)\right)\nonumber\\
&+ \mathbb{P}\left(\sup_{\tau \in\left[n_0, \tau_{n,p}\right]} \mathbb{M}_n(\tau)-\mathbb{M}_n\left(\tau_0\right) \geq 0\right) \nonumber \\
\leq & \frac{\mathbb{E} \sup _{\tau \in\left[\tau_{n, p}, \tau_0-\epsilon\right]}\left|\mathbb{M}_n(\tau)-\mathbb{M}(\tau)-\mathbb{M}_n\left(\tau_0\right)+\mathbb{M}\left(\tau_0\right)\right|}{\sup _{\tau \in\left[\tau_{n, p}, \tau_0-\epsilon\right]} \left(\mathbb{M}\left(\tau_0\right)-\mathbb{M}(\tau)\right)}+ \mathbb{P}\left(\sup_{\tau \in\left[n_0, \tau_{n,p}\right]} \mathbb{M}_n(\tau)-\mathbb{M}_n\left(\tau_0\right) \geq 0\right) . 
\end{align}

In the last inequality, we used Markov inequality.

Notice that when $\tau \in\left[n_0, \tau_{n, p}\right]$, $\widehat{\psi}^{\tau+1, n}$ might be inconsistent in estimating $\psi^{\tau_0+1, n}$. On the other hand, when $\tau \in\left[\tau_{n, p}, \tau_0\right]$, we shall see that the membership maps can be consistently recovered with $\widehat{\psi}(\tau)=\psi$ in probability for all $\tau \in\left[\tau_{n, p}, \tau_0\right]$. Then we can assume that $\widehat{S}_{1, \boldsymbol{c}}^\tau=\widehat{S}_{2, \boldsymbol{c}}^\tau=S_{\boldsymbol{c}}$ (or equivalently $\widehat{\psi}^{1, \tau}=\widehat{\psi}^{\tau+1, n}=\psi$ ) holds for all $\boldsymbol{c}   \in \mathbf{\Psi}$ and $\tau_{n, p} \leq \tau \leq \tau_0$ and rewrite 
$$
\begin{aligned}
 \mathbb{M}_n(\tau)&:=l\left(\left\{\widehat{\theta}_{1, \boldsymbol{c}}^\tau, \widehat{\eta}_{1, \boldsymbol{c}}^\tau\right\} ; \widehat{\psi}^{1, \tau}\right)+l\left(\left\{\widehat{\theta}_{2, \boldsymbol{c}}^\tau, \widehat{\eta}_{2, \boldsymbol{c}}^\tau\right\} ; \widehat{\psi}^{\tau+1, n}\right) \\
 & = \sum_{\boldsymbol{c}   \in \mathbf{\Psi}} \sum_{\xi \in S_{\boldsymbol{c}}} g_{\xi}(\widehat{\theta}_{1, \boldsymbol{c}}^\tau, \widehat{\eta}_{1, \boldsymbol{c}}^\tau; 1, \tau) + g_{\xi}(\widehat{\theta}_{2, \boldsymbol{c}}^\tau, \widehat{\eta}_{2, \boldsymbol{c}}^\tau ; \tau + 1, n) \\
 \mathbb{M}(\tau)&:=\mathbb{E}l\left(\left\{\theta_{1, \boldsymbol{c}}, \eta_{1, \boldsymbol{c}}\right\} ; \psi^{1, \tau}\right)+\mathbb{E}l\left(\left\{\theta^{\tau}_{2, \boldsymbol{c}}, \eta^{\tau}_{2, \boldsymbol{c}}\right\} ; \psi^{\tau+1, n}\right)\\
 & = \mathbb{E}\sum_{\boldsymbol{c}   \in \mathbf{\Psi}} \sum_{\xi \in S_{\boldsymbol{c}}} g_{\xi}(\theta_{1, \boldsymbol{c}},\eta_{1, \boldsymbol{c}} ; 1, \tau) + g_{\xi}(\theta^{\tau}_{2, \boldsymbol{c}}, \eta^{\tau}_{2, \boldsymbol{c}} ; \tau + 1, n) .
\end{aligned}
$$
We firstly show that $\widehat{\psi}(\tau)=\psi$ in probability and then find the bounds of each term in (6). 

\subsubsection*{A.10.1.1 Evaluating $\sup _{\tau \in\left[\tau_{n, p}, \tau_0\right]} \mathbb{P}(\widehat{\psi}(\tau) \neq \psi)$}

Let $\widehat{\psi}(\tau)$ be either $\widehat{\psi}^{1, \tau}$ or $\widehat{\psi}^{\tau+1, n}$. From Theorems 1 and 2, we have for any constant $B>0$, there exists a large enough constant $C_B$ such that

$$
\sup _{\tau \in\left[\tau_{n, p}, \tau_0\right]} \mathbb{P}(\widehat{\psi}(\tau) \neq \psi) \leq C_B\left(\tau_0-\tau_{n, p}\right) p\left[(p n)^{-(B+1)}+\exp \{-B \sqrt{p^{K-1}}\}\right] .
$$
Note that by choosing $B$ to be large enough, we have $p\left(\tau_0-\tau_{n, p}\right)(p n)^{-(B+1)}=o\left(\sqrt{\frac{\left(\tau_0-\tau_{n, p}\right) \log (n p)}{n^2 s_{\min }^2}}\right)$. On the other hand, the assumption that $\frac{\log (n p)}{\sqrt{p^{K-1}}} \rightarrow 0$ in condition C4 implies $p n \sqrt{\frac{\left(\tau_0-\tau_{n, p}\right) s_{\min }^2}{\log (n p)}}=o(\exp \{B \sqrt{p^{K-1}}\})$ for some large enough constant $B$. 

Consequently, we have $\left(\tau_0-\tau_{n, p}\right) p \exp \{-B \sqrt{p^{K-1}}\}=o\left(\sqrt{\frac{\left(\tau_0-\tau_{n, p}\right) \log (n p)}{n^2 s_{\min }^2}}\right)$, and hence we conclude that 
\begin{align}
    \sup _{\tau \in\left[\tau_{n, p}, \tau_0\right]} \mathbb{P}(\widehat{\psi}(\tau) \neq \psi)=o\left(\sqrt{\frac{\left(\tau_0-\tau_{n, p}\right) \log (n p)}{n^2 s_{\min }^2}}\right).
\end{align}

\subsubsection*{A.10.1.2 Evaluating $\mathbb{M}(\tau)-\mathbb{M}\left(\tau_0\right) $ }
Note that for any $\tau \in [1,n]$, $l\left(\left\{\theta_{\boldsymbol{c}}, \eta_{\boldsymbol{c}}\right\} ; \psi^{1,n}\right) = l\left(\left\{\theta_{\boldsymbol{c}}, \eta_{\boldsymbol{c}}\right\} ; \psi^{1,\tau}\right) +l\left(\left\{\theta_{\boldsymbol{c}}, \eta_{\boldsymbol{c}}\right\} ; \psi^{1+\tau,n}\right). $
Then for any $\tau \in\left[n_0, \tau_0\right]$,
$$
\begin{aligned}
\mathbb{M}(\tau)-\mathbb{M}\left(\tau_0\right)= & \mathbb{E} l\left(\left\{\theta_{1, \boldsymbol{c}}, \eta_{1, \boldsymbol{c}}\right\} ; \psi^{1, \tau}\right)+\mathbb{E} l\left(\left\{\theta^{\tau}_{2, \boldsymbol{c}}, \eta^{\tau}_{2, \boldsymbol{c}}\right\} ; \psi^{\tau+1, n}\right) \\
& -\mathbb{E} l\left(\left\{\theta_{1, \boldsymbol{c}}, \eta_{1, \boldsymbol{c}}\right\} ; \psi^{1, \tau_0}\right)-\mathbb{E} l\left(\left\{\theta_{2, \boldsymbol{c}} \eta_{2, \boldsymbol{c}}\right\} ; \psi^{\tau_0+1, n}\right) \\
= & \mathbb{E}l\left(\left\{\theta^{\tau}_{2, \boldsymbol{c}}, \eta^{\tau}_{2, \boldsymbol{c}}\right\} ; \psi^{\tau+1, \tau_0}\right)-\mathbb{E} l\left(\left\{\theta_{1, \boldsymbol{c}}, \eta_{1, \boldsymbol{c}}\right\} ; \psi^{\tau+1, \tau_0}\right) \\
& +\mathbb{E} l\left(\left\{\theta^{\tau}_{2, \boldsymbol{c}}, \eta^{\tau}_{2, \boldsymbol{c}}\right\} ; \psi^{\tau_0+1, n}\right)-\mathbb{E} l\left(\left\{\theta_{2, \boldsymbol{c}}, \eta_{2, \boldsymbol{c}}\right\} ; \psi^{\tau_0+1, n}\right) .
\end{aligned}
$$
 By (4), we have, there exist $\theta_{\boldsymbol{c}}^* \in\left[\theta_{1, \boldsymbol{c}}, \theta^{\tau}_{2, \boldsymbol{c}}\right], \eta_{\boldsymbol{c}}^* \in \left[\eta_{1, \boldsymbol{c}}, \eta^{\tau}_{2, \boldsymbol{c}}\right], \boldsymbol{c}   \in \mathbf{\Psi}$, such that
$$
\begin{aligned}
& \mathbb{E}l\left(\left\{\theta^{\tau}_{2, \boldsymbol{c}}, \eta^{\tau}_{2, \boldsymbol{c}}\right\} ; \psi^{\tau+1, \tau_0}\right)-\mathbb{E} l\left(\left\{\theta_{1, \boldsymbol{c}}, \eta_{1, \boldsymbol{c}}\right\} ; \psi^{\tau+1, \tau_0}\right) \\
\leq  & -\frac{1}{2} \sum_{\boldsymbol{c}   \in \mathbf{\Psi}} n_{\boldsymbol{c}}\left(\tau_0-\tau\right)\left\{\frac{\theta_{1, \boldsymbol{c}} \eta_{1, \boldsymbol{c}}}{\theta_{1, \boldsymbol{c}}+\eta_{1, \boldsymbol{c}}}\left(\frac{\theta^{\tau}_{2, \boldsymbol{c}}-\theta_{1, \boldsymbol{c}}}{\theta_{\boldsymbol{c}}^*}\right)^2+\frac{\left(1-\theta_{1, \boldsymbol{c}}\right) \eta_{1, \boldsymbol{c}}}{\theta_{1, \boldsymbol{c}}+\eta_{1, \boldsymbol{c}}}\left(\frac{\theta^{\tau}_{2, \boldsymbol{c}}-\theta_{1, \boldsymbol{c}}}{1-\theta_{\boldsymbol{c}}^*}\right)^2\right. \\
& \left.+\frac{\theta_{1, \boldsymbol{c}} \eta_{1, \boldsymbol{c}}}{\theta_{1, \boldsymbol{c}}+\eta_{1, \boldsymbol{c}}}\left(\frac{\eta^{\tau}_{2, \boldsymbol{c}}-\eta_{1, \boldsymbol{c}}}{\eta_{\boldsymbol{c}}^*}\right)^2+\frac{\left(1-\eta_{1, \boldsymbol{c}}\right) \theta_{1, \boldsymbol{c}}}{\theta_{1, \boldsymbol{c}}+\eta_{1, \boldsymbol{c}}}\left(\frac{\eta^{\tau}_{2, \boldsymbol{c}}-\eta_{1, \boldsymbol{c}}}{1-\eta_{\boldsymbol{c}}^*}\right)^2\right\} \\
\leq & -C_1\left(\tau_0-\tau\right) \sum_{\boldsymbol{c}   \in \mathbf{\Psi}} \sum_{\xi \in S_{\boldsymbol{c}}}\left[\left(\theta_{1, \boldsymbol{c}}-\theta_{2, \boldsymbol{c}}\right)^2+\left(\eta_{1, \boldsymbol{c}}-\eta_{2, \boldsymbol{c}}\right)^2\right] \\
\leq & -C'_1\left(\tau_0-\tau\right)\left[\left\|\mathbf{A}_{1,1}-\mathbf{A}_{2,1}\right\|_F^2+\left\|\mathbf{A}_{1,2}-\mathbf{A}_{2,2}\right\|_F^2\right]
\end{aligned}
$$
for some constant $C_1,C'_1>0$. The last inequality can be shown by Cauchy–Schwarz inequality. 

Similarly, there exist $\theta_{\boldsymbol{c}}^{\dagger} \in\left[\theta_{2, \boldsymbol{c}}, \theta^{\tau}_{2, \boldsymbol{c}}\right], \eta_{\boldsymbol{c}}^{\dagger} \in\left[\eta_{2, \boldsymbol{c}}, \eta^{\tau}_{2, \boldsymbol{c}}\right], \boldsymbol{c}   \in \mathbf{\Psi}$, such that
$$
\begin{aligned}
& \mathbb{E} l\left(\left\{\theta^{\tau}_{2, \boldsymbol{c}}, \eta^{\tau}_{2, \boldsymbol{c}}\right\} ; \psi^{\tau_0+1, n}\right)-\mathbb{E} l\left(\left\{\theta_{2, \boldsymbol{c}}, \eta_{2, \boldsymbol{c}}\right\} ; \psi^{\tau_0+1, n}\right) \\
\leq & -\frac{1}{2} \sum_{\boldsymbol{c}   \in \mathbf{\Psi}} \sum_{\xi \in S_{\boldsymbol{c}}}\left(n-\tau_0\right)\left\{\frac{\theta_{2, \boldsymbol{c}} \eta_{2, \boldsymbol{c}}}{\theta_{2, \boldsymbol{c}}+\eta_{2, \boldsymbol{c}}}\left(\frac{\theta^{\tau}_{2, \boldsymbol{c}}-\theta_{2, \boldsymbol{c}}}{\theta_{\boldsymbol{c}}^{\dagger}}\right)^2+\frac{\left(1-\theta_{2, \boldsymbol{c}}\right) \eta_{2, \boldsymbol{c}}}{\theta_{2, \boldsymbol{c}}+\eta_{2, \boldsymbol{c}}}\left(\frac{\theta^{\tau}_{2, \boldsymbol{c}}-\theta_{2, \boldsymbol{c}}}{1-\theta_{\boldsymbol{c}}^{\dagger}}\right)^2\right. \\
& \left.+\frac{\theta_{2, \boldsymbol{c}} \eta_{2, \boldsymbol{c}}}{\theta_{2, \boldsymbol{c}}+\eta_{2, \boldsymbol{c}}}\left(\frac{\eta^{\tau}_{2, \boldsymbol{c}}-\eta_{2, \boldsymbol{c}}}{\eta_{\boldsymbol{c}}^{\dagger}}\right)^2+\frac{\left(1-\eta_{2, \boldsymbol{c}}\right) \theta_{2, \boldsymbol{c}}}{\theta_{2, \boldsymbol{c}}+\eta_{2, \boldsymbol{c}}}\left(\frac{\eta^{\tau}_{2, \boldsymbol{c}}-\eta_{2, \boldsymbol{c}}}{1-\eta_{\boldsymbol{c}}^{\dagger}}\right)^2\right\} \\
\leq & -C_2^{\prime}\left(n-\tau_0\right) \sum_{\boldsymbol{c}   \in \mathbf{\Psi}} \sum_{\xi \in S_{\boldsymbol{c}}} \frac{\left(\tau_0-\tau\right)^2}{(n-\tau)^2}\left[\left(\theta_{1, \boldsymbol{c}}-\theta_{2, \boldsymbol{c}}\right)^2+\left(\eta_{1, \boldsymbol{c}}-\eta_{2, \boldsymbol{c}}\right)^2\right] \\
\leq & -\frac{C_2\left(\tau_0-\tau\right)^2}{n-\tau}\left[\left\|\mathbf{A}_{1,1}-\mathbf{A}_{2,1}\right\|_F^2+\left\|\mathbf{A}_{1,2}-\mathbf{A}_{2,2}\right\|_F^2\right]
\end{aligned}
$$
for some constants $C_2^{\prime}, C_2>0$. Consequently, we conclude that there exists a constant $C_3>0$ such that for any $n_0 \leq \tau \leq \tau_0$, we have
\begin{align*}
    \mathbb{M}(\tau)-\mathbb{M}\left(\tau_0\right) \leq-C_3\left(\tau_0-\tau\right)\left[\left\|\mathbf{A}_{1,1}-\mathbf{A}_{2,1}\right\|_F^2+\left\|\mathbf{A}_{1,2}-\mathbf{A}_{2,2}\right\|_F^2\right].
\end{align*}

\subsubsection*{A.10.1.3 Evaluating $\mathbb{E} \sup _{\tau \in\left[\tau_{n, p}, \tau_0\right]}\left|\mathbb{M}_n(\tau)-\mathbb{M}(\tau)-\mathbb{M}_n\left(\tau_0\right)+\mathbb{M}\left(\tau_0\right)\right|$ }
Note that when $\hat{\psi}=\psi$,
$$
\begin{aligned}
&\mathbb{M}_n(\tau)-\mathbb{M}(\tau)-\mathbb{M}_n\left(\tau_0\right)+\mathbb{M}\left(\tau_0\right)\\= & \sum_{\boldsymbol{c}   \in \mathbf{\Psi}} \sum_{\xi \in S_{\boldsymbol{c}}} g_{\xi}\left(\mathbb{\theta}_{1, \boldsymbol{c}}^\tau, \widehat{\eta}_{1, \boldsymbol{c}}^\tau ; 1,\tau\right)+\sum_{\boldsymbol{c}   \in \mathbf{\Psi}} \sum_{\xi \in S_{\boldsymbol{c}}} g_{ \xi}\left(\widehat{\theta}_{2, \boldsymbol{c}}^\tau, \widehat{\eta}_{1, \boldsymbol{c}}^\tau ; \tau+1,n\right) \\
& -\mathbb{E} \sum_{\boldsymbol{c}   \in \mathbf{\Psi}} \sum_{\xi \in S_{\boldsymbol{c}}} g_{ \xi}\left(\theta_{1, \boldsymbol{c}}, \eta_{1, \boldsymbol{c}} ; 1,\tau\right)-\mathbb{E} \sum_{\boldsymbol{c}   \in \mathbf{\Psi}} \sum_{\xi \in S_{\boldsymbol{c}}} g_{ \xi}\left(\theta^{\tau}_{2, \boldsymbol{c}}, \eta^{\tau}_{2, \boldsymbol{c}} ; \tau+1,n\right) \\
& -\sum_{\boldsymbol{c}   \in \mathbf{\Psi}} \sum_{\xi \in S_{\boldsymbol{c}}} g_{ \xi}\left(\widehat{\theta}_{1, \boldsymbol{c}}^{\tau_0}, \widehat{\eta}_{1, \boldsymbol{c}}^{\tau_0} ;1, \tau_0\right)-\sum_{\boldsymbol{c}   \in \mathbf{\Psi}} \sum_{\xi \in S_{\boldsymbol{c}}} g_{\xi}\left(\widehat{\theta}_{2, \boldsymbol{c}}^{\tau_0}, \widehat{\eta}_{2, \boldsymbol{c}}^{\tau_0} ; \tau_0+1,n\right) \\
& +\mathbb{E} \sum_{\boldsymbol{c}   \in \mathbf{\Psi}} \sum_{\xi \in S_{\boldsymbol{c}}} g_{ \xi}\left(\theta_{1, \boldsymbol{c}}, \eta_{1, \boldsymbol{c}} ; 1, \tau_0\right)+\mathbb{E} \sum_{\boldsymbol{c}   \in \mathbf{\Psi}} g_{\xi \in S_{\boldsymbol{c}}} g_{ \xi}\left(\theta_{2, \boldsymbol{c}}, \eta_{2, \boldsymbol{c}} ; \tau_0+1,n\right).
\end{aligned}
$$
Then,
$$
\begin{aligned}
& g_{\xi}\left(\widehat{\theta}_{1, \boldsymbol{c}}^\tau, \widehat{\eta}_{1, \boldsymbol{c}}^\tau ;1, \tau\right)-g_{\xi}\left(\widehat{\theta}_{1, \boldsymbol{c}}^{\tau_0}, \widehat{\eta}_{1, \boldsymbol{c}}^{\tau_0} ;1, \tau_0\right)-\mathbb{E}\left[g_{ \xi}\left(\theta_{1, \boldsymbol{c}}, \eta_{1, \boldsymbol{c}} ;1, \tau\right)-g_{ \xi}\left(\theta_{1, \boldsymbol{c}}, \eta_{1, \boldsymbol{c}} ;1, \tau_0\right)\right] \\
= & [g_{\xi}\left(\widehat{\theta}_{1, \boldsymbol{c}}^\tau, \widehat{\eta}_{1, \boldsymbol{c}}^\tau ;1, \tau\right)-g_{\xi}\left(\widehat{\theta}_{1, \boldsymbol{c}}^{\tau_0}, \widehat{\eta}_{1, \boldsymbol{c}}^{\tau_0} ;1, \tau\right)]
+[g_{\xi}\left(\theta_{1, \boldsymbol{c}}, \eta_{1, \boldsymbol{c}} ;1+\tau, \tau_0\right)\\&-g_{\xi}\left(\widehat{\theta}_{1, \boldsymbol{c}}^{\tau_0}, \widehat{\eta}_{1, \boldsymbol{c}}^{\tau_0} ;1+\tau, \tau_0\right)]
+[\mathbb{E}\left[g_{ \xi}\left(\theta_{1, \boldsymbol{c}}, \eta_{1, \boldsymbol{c}} ;1+ \tau, \tau_0\right)\right]-g_{\xi}\left(\theta_{1, \boldsymbol{c}}, \eta_{1, \boldsymbol{c}} ;1+\tau, \tau_0\right)]
\end{aligned}
$$
For the first term, by the mean value theorem, it suffices to find the upper bound of $\sup _{\boldsymbol{c}   \in \mathbf{\Psi}}\left|\widehat{\theta}_{1, \boldsymbol{c}}^\tau-\widehat{\theta}_{1, \boldsymbol{c}}^{\tau_0}\right|$ and $\sup _{\boldsymbol{c}   \in \mathbf{\Psi}}\left|\widehat{\eta}_{1, \boldsymbol{c}}^\tau-\widehat{\eta}_{1, \boldsymbol{c}}^\tau\right|$. Notice that $$
\left|\widehat{\theta}_{1, \boldsymbol{c}}^\tau-\widehat{\theta}_{1, \boldsymbol{c}}^{\tau_0}\right|=\left|\frac{\sum_{\xi \in S_{\boldsymbol{c}}} \sum_{t=1}^\tau X_{\xi}^t\left(1-X_{\xi}^{t-1}\right)}{\sum_{\xi \in S_{\boldsymbol{c}}} \sum_{t=1}^\tau\left(1-X_{\xi}^{t-1}\right)}-\frac{\sum_{\xi \in S_{\boldsymbol{c}}} \sum_{t=1}^{\tau_0} X_{\xi}^t\left(1-X_{\xi}^{t-1}\right)}{\sum_{\xi \in S_{\boldsymbol{c}}} \sum_{t=1}^{\tau_0}\left(1-X_{\xi}^{t-1}\right)}\right| .
$$
Similar to Lemma 2, we can show that for any constant $B>0$, there exists a large enough constant $B_1$ such that with probability larger than $1-O\left((n p)^{-(B+2)}\right)$,
$$
\begin{aligned}
 \left|\frac{1}{\tau n_{\boldsymbol{c}}} \sum_{\xi \in S_{\boldsymbol{c}}} \sum_{t=1}^\tau\left(1-X_{\xi}^{t-1}\right)-\frac{\eta_{1, \boldsymbol{c}}}{\theta_{1, \boldsymbol{c}}+\eta_{1, \boldsymbol{c}}}\right| \leq B_1 \sqrt{\frac{\log (n p)}{\tau n_{\boldsymbol{c}}}}, \\
 \left|\frac{1}{\tau n_{\boldsymbol{c}}} \sum_{\xi \in S_{\boldsymbol{c}}} \sum_{t=1}^\tau X_{\xi}^t\left(1-X_{\xi}^{t-1}\right)-\frac{\eta_{1, \boldsymbol{c}}}{\theta_{1, \boldsymbol{c}}+\eta_{1, \boldsymbol{c}}}\right| \leq B_1 \sqrt{\frac{\log (n p)}{\tau n_{\boldsymbol{c}}}}.
\end{aligned}
$$
Then
$$
\begin{aligned}
 &\frac{1}{\tau\left(\tau_0-\tau\right) n_{\boldsymbol{c}}^2}\left\lvert\left[\sum_{\xi \in S_{\boldsymbol{c}}} \sum_{t=1}^\tau X_{\xi}^t\left(1-X_{\xi}^{t-1}\right)\right]\left[\sum_{\xi \in S_{\boldsymbol{c}}} \sum_{t=\tau+1}^{\tau_0}\left(1-X_{\xi}^{t-1}\right)\right]\right. \\
 &-\left[\sum_{\xi \in S_{\boldsymbol{c}}} \sum_{t=\tau+1}^{\tau_0} X_{\xi}^t\left(1-X_{\xi}^{t-1}\right)\right]\left[\sum_{\xi \in S_{\boldsymbol{c}}} \sum_{t=1}^\tau\left(1-X_{\xi}^{t-1}\right)\right] \left\lvert\, \right.\\
 &\leq \left(\frac{\eta_{1,k,\ell}}{\eta_{1,k,\ell} + \theta_{1,k,\ell}}+B_1 \sqrt{\frac{\log (n p)}{\tau n_{\boldsymbol{c}}}}\right)\left(\frac{\eta_{1,k,\ell}}{\eta_{1,k,\ell} + \theta_{1,k,\ell}}+B_1 \sqrt{\frac{\log (n p)}{(\tau_0-\tau) n_{\boldsymbol{c}}}}\right)\\
 &-\left(\frac{\eta_{1,k,\ell}}{\eta_{1,k,\ell} + \theta_{1,k,\ell}}-B_1 \sqrt{\frac{\log (n p)}{\tau n_{\boldsymbol{c}}}}\right)\left(\frac{\eta_{1,k,\ell}}{\eta_{1,k,\ell} + \theta_{1,k,\ell}}-B_1 \sqrt{\frac{\log (n p)}{(\tau_0-\tau) n_{\boldsymbol{c}}}}\right)\\
 & \leq 2B_1 (\sqrt{\frac{\log (n p)}{\left(\tau_0-\tau\right) n_{\boldsymbol{c}}}}+\sqrt{\frac{\log (n p)}{\tau n_{\boldsymbol{c}}}}) \leq B_2 \sqrt{\frac{\log (np)}{n n_{\boldsymbol{c}}}} .
\end{aligned}
$$
Since $n_{\boldsymbol{c}} \geq s_{\min }^{|\boldsymbol{c}|}$, we have with probability larger than $1-O\left((n p)^{-(B+2)}\right)$,
$$
\left|\widehat{\theta}_{1, \boldsymbol{c}}^\tau-\widehat{\theta}_{1, \boldsymbol{c}}^{\tau_0}\right| \leq \frac{c_0 \tau\left(\tau_0-\tau\right) n_{\boldsymbol{c}}^2}{\tau_0 \tau n_{\boldsymbol{c}}^2} \sqrt{\frac{\log (np)}{n n_{\boldsymbol{c}}}}  \leq \frac{c_0 (\tau_0-\tau)}{\tau_0} \sqrt{\frac{\log (n p)}{n n_{\boldsymbol{c}}}} \leq c_1\sqrt{\frac{\tau_0-\tau}{\tau_0}} \sqrt{\frac{\log (n p)}{n s_{\min }^{|\boldsymbol{c}|}}},
$$
for some constant $c_0, c_1>0$. Hence, 
$$
\begin{aligned}
    &\sup _{\boldsymbol{c}   \in \mathbf{\Psi}}\left|\widehat{\theta}_{1, \boldsymbol{c}}^\tau-\widehat{\theta}_{1, \boldsymbol{c}}^{\tau_0}\right| \leq c_1 \sqrt{\frac{\tau_0-\tau}{\tau_0}} \sqrt{\frac{\log (n p)}{n s_{\min }^2}} \\
    & \sup _{\boldsymbol{c}   \in \mathbf{\Psi}}\left|\widehat{\theta}_{1, \boldsymbol{c}}^\tau-\widehat{\theta}_{1, \boldsymbol{c}}^{\tau_0}\right| \leq c_1 \sqrt{\frac{\tau_0-\tau}{\tau_0}} \sqrt{\frac{\log (n p)}{n s_{\min }^2}} \\
    &\sup _{\boldsymbol{c}   \in \mathbf{\Psi}}\left|\widehat{\theta}_{2, \boldsymbol{c}}^\tau-\widehat{\theta}_{2, \boldsymbol{c}}^{\tau_0}\right| \leq c_2 \sqrt{\frac{\tau_0-\tau}{\tau_0}} \sqrt{\frac{\log (n p)}{n s_{\min }^2}}\\
    &\sup _{\boldsymbol{c}   \in \mathbf{\Psi}}\left|\widehat{\theta}_{2, \boldsymbol{c}}^\tau-\widehat{\theta}_{2, \boldsymbol{c}}^{\tau_0}\right| \leq c_2 \sqrt{\frac{\tau_0-\tau}{\tau_0}} \sqrt{\frac{\log (n p)}{n s_{\min }^2}}\\
\end{aligned}
$$
 with probability larger than $1-O((n p)^{-B})$ for some constants $c_1, c_2$, and then 
$$
\begin{aligned}
&g_{\xi}\left(\widehat{\theta}_{1, \boldsymbol{c}}^\tau, \widehat{\eta}_{1, \boldsymbol{c}}^\tau ;1, \tau\right)-g_{\xi}\left(\widehat{\theta}_{1, \boldsymbol{c}}^{\tau_0}, \widehat{\eta}_{1, \boldsymbol{c}}^{\tau_0} ;1, \tau\right) \\= & \sum_{t=1}^\tau\left\{X_{\xi}^t\left(1-X_{\xi}^{t-1}\right) \log \frac{\widehat{\theta}_{1, \boldsymbol{c}}^\tau}{\widehat{\theta}_{1, \boldsymbol{c}}^{\tau_0}}+\left(1-X_{\xi}^t\right)\left(1-X_{\xi}^{t-1}\right) \log \frac{1-\widehat{\theta}_{1, \boldsymbol{c}}^\tau}{1-\widehat{\theta}_{1, \boldsymbol{c}}^{\tau_0}}\right. \\
& \left.+\left(1-X_{\xi}^t\right) X_{\xi}^{t-1} \log \frac{\widehat{\eta}_{1, \boldsymbol{c}}^\tau}{\widehat{\eta}_{1, \boldsymbol{c}}^{\tau_0}}+X_{\xi}^t X_{\xi}^{t-1} \log \frac{1-\widehat{\eta}_{1, \boldsymbol{c}}^\tau}{1-\widehat{\eta}_{1, \boldsymbol{c}}^{\tau_0}}\right\} \\\leq& c_1 \tau \sqrt{\frac{\tau_0-\tau}{\tau_0}} \sqrt{\frac{\log (n p)}{n s_{\min }^2}}. 
\end{aligned}$$

For the second term, note that $\left\{\widehat{\theta}_{1, \boldsymbol{c}}^\tau, \widehat{\eta}_{1, \boldsymbol{c}}^\tau\right\}$ is the maximizer of $\sum_{\boldsymbol{c}   \in \mathbf{\Psi}} \sum_{\xi \in S_{\boldsymbol{c}}} g_{1, \xi}\left(\theta_{\boldsymbol{c}}, \eta_{\boldsymbol{c}} ; \tau\right)$. Applying Taylor's expansion we have, there exist random scalars $\theta_{\boldsymbol{c}}^{-} \in\left[\widehat{\theta}_{1, \boldsymbol{c}}^\tau, \theta_{1, \boldsymbol{c}}\right], \eta_{\boldsymbol{c}}^{-} \in \left[\widehat{\eta}_{1, \boldsymbol{c}}^\tau, \eta_{1, \boldsymbol{c}}\right]$ such that
$$
\begin{aligned}
& g_{\xi}\left(\theta_{1, \boldsymbol{c}}, \eta_{1, \boldsymbol{c}} ;1+\tau, \tau_0\right)-g_{\xi}\left(\widehat{\theta}_{1, \boldsymbol{c}}^{\tau_0}, \widehat{\eta}_{1, \boldsymbol{c}}^{\tau_0} ;1+\tau, \tau_0\right) \\
\leq & \frac{1}{2} (\tau_0-\tau)\left\{\left(\frac{\theta_{1, \boldsymbol{c}}-\widehat{\theta}_{1, \boldsymbol{c}}^\tau}{\theta_{\boldsymbol{c}}^{-}}\right)^2+\left(\frac{\theta_{1, \boldsymbol{c}}-\widehat{\theta}_{1, \boldsymbol{c}}^\tau}{1-\theta_{\boldsymbol{c}}^{-}}\right)^2+\left(\frac{\eta_{1, \boldsymbol{c}}-\widehat{\eta}_{1, \boldsymbol{c}}^\tau}{\eta_{\boldsymbol{c}}^{-}}\right)^2+\left(\frac{\eta_{1, \boldsymbol{c}}-\widehat{\eta}_{1, \boldsymbol{c}}^\tau}{1-\eta_{\boldsymbol{c}}^{-}}\right)^2\right\} .
\end{aligned}
$$

On the other hand, when $\widehat{\psi}=\psi$, similar to Proposition 4 and Theorem 3, we can show that for any $B>0$, there exists a large enough constant $C^{-}$such that $\max _{\boldsymbol{c}   \in \mathbf{\Psi}, \tau \in\left[\tau_{n, p}, \tau_0\right]} \mid \widehat{\theta}_{1, \boldsymbol{c}^{-}}^\tau$ $\theta_{1, \boldsymbol{c}} \left\lvert\, \leq C^{-} \sqrt{\frac{\log (n p)}{n s_{\min }^2}}\right.$, and $\max _{\boldsymbol{c}   \in \mathbf{\Psi}, \tau \in\left[\tau_{n, p}, \tau_0\right]}\left|\widehat{\eta}_{1, \boldsymbol{c}}^\tau-\eta_{1, \boldsymbol{c}}\right|=C^{-} \sqrt{\frac{\log (n p)}{n s_{\min }^2}}$ hold with probability greater than $1-O\left((n p)^{-B}\right)$. Consequently, we have, when $\widehat{\psi}=\psi$, there exists a large enough constant $C_4>0$ such that
\begin{align}
      g_{\xi}\left(\theta_{1, \boldsymbol{c}}, \eta_{1, \boldsymbol{c}} ;1+\tau, \tau_0\right)-g_{\xi}\left(\widehat{\theta}_{1, \boldsymbol{c}}^{\tau_0}, \widehat{\eta}_{1, \boldsymbol{c}}^{\tau_0} ;1+\tau, \tau_0\right) \leq C_4 (\tau_0-\tau)  \frac{\log (n p)}{n s_{\min }^2} .
\end{align}

Lastly, the upper bound of the last term can be obtained by (5), 
$$\begin{aligned}
    &|\sum_{\boldsymbol{c}   \in \mathbf{\Psi}} \sum_{\xi \in S_{\boldsymbol{c}}} g_{\xi}(\theta_{1, \boldsymbol{c}}, \eta_{1, \boldsymbol{c}} ;1+ \tau, \tau_0)- \sum_{\boldsymbol{c}   \in \mathbf{\Psi}} \sum_{\xi \in S_{\boldsymbol{c}}} \mathbb{E}g_{\xi}(\theta_{1, \boldsymbol{c}}, \eta_{1, \boldsymbol{c}} ;1+ \tau, \tau_0)| \nonumber\\\leq & B_1 \left(\tau_0-\tau\right)p^{K-1} \sqrt{\frac{\log(np)}{\left(\tau_0-\tau\right)p^{K-1}}}
\end{aligned}$$
Combining these three terms,using the fact that $\tau_0 \simeq O(n),  \sqrt{\frac{\log (n p)}{p^{K-1}}} \leq \sqrt{\frac{\log (n p)}{s_{\min }^2}}$, and $\frac{\left(\tau_0-\tau\right) \log (n p)}{n s_{\min }^2}=o\left(\sqrt{\frac{\left(\tau_0-\tau\right) \log (n p)}{s_{\min }^2}}\right)$, we have 
\begin{align}
    &\sum_{\boldsymbol{c}   \in \mathbf{\Psi}} \sum_{\xi \in S_{\boldsymbol{c}}}\left\{g_{\xi}\left(\widehat{\theta}_{1, \boldsymbol{c}}^\tau, \widehat{\eta}_{1, \boldsymbol{c}}^\tau ;1, \tau\right)-g_{\xi}\left(\widehat{\theta}_{1, \boldsymbol{c}}^{\tau_0}, \widehat{\eta}_{1, \boldsymbol{c}}^{\tau_0} ;1, \tau_0\right)\right.\nonumber\\
    &\left.-\mathbb{E}\left[g_{ \xi}\left(\theta_{1, \boldsymbol{c}}, \eta_{1, \boldsymbol{c}} ;1, \tau\right)-g_{ \xi}\left(\theta_{1, \boldsymbol{c}}, \eta_{1, \boldsymbol{c}} ;1, \tau_0\right)\right]\right\}\nonumber\\&
    \leq B_1 \left(\tau_0-\tau\right)p^{K-1} \sqrt{\frac{\log(np)}{\left(\tau_0-\tau\right)p^{K-1}}} + C_4 (\tau_0-\tau) p^{K-1}  \frac{\log (n p)}{n s_{\min }^2} +c_1 p^{K-1} \tau \sqrt{\frac{\tau_0-\tau}{\tau_0}} \sqrt{\frac{\log (n p)}{n s_{\min }^2}}\nonumber\\ &\leq
    c_2p^{K-1}\sqrt{\frac{(\tau_0-\tau)\log np}{s_{\min }^2}}. 
\end{align}
Similarly with (5) and (10), there exists a constant $c_7>0$ such that with probability larger than $\left(1-O(n p)^{-B}\right)$,
\begin{align}
& g_{\xi}\left(\widehat{\theta}_{2, \boldsymbol{c}}^\tau, \widehat{\eta}_{2, \boldsymbol{c}}^\tau ; \tau+1, n\right)-g_{ \xi}\left(\widehat{\theta}_{2, \boldsymbol{c}}^{\tau_0} \widehat{\eta}_{2, \boldsymbol{c}}^{\tau_0} ; \tau_0+1, n\right)\nonumber\\
&-\mathbb{E}\left[g_{\xi}\left(\theta^{\tau}_{2, \boldsymbol{c}}, \eta^{\tau}_{2, \boldsymbol{c}} ; \tau+1, n\right)-g_{\xi}\left(\theta_{2, \boldsymbol{c}}, \eta_{2, \boldsymbol{c}} ; \tau_0+1, n\right)\right] \nonumber\\
=& \left\{g_{\xi}\left(\widehat{\theta}_{2, \boldsymbol{c}}^\tau, \widehat{\eta}_{2, \boldsymbol{c}}^\tau ; \tau+1, n\right)-g_{\xi}\left(\theta^{\tau}_{2, \boldsymbol{c}}, \eta^{\tau}_{2, \boldsymbol{c}} ; \tau+1, n\right)\right\}\nonumber\\
 &+  \left\{g_{\xi}\left(\theta_{2, \boldsymbol{c}}, \eta_{2, \boldsymbol{c}} ; \tau_0+1, n\right)-g_{ \xi}\left(\widehat{\theta}_{2, \boldsymbol{c}}^{\tau_0} + \widehat{\eta}_{2, \boldsymbol{c}}^{\tau_0} ; \tau_0+1, n\right)\right\}\nonumber\\
& -\left\{\mathbb{E}\left[g_{\xi}\left(\theta^{\tau}_{2, \boldsymbol{c}}, \eta^{\tau}_{2, \boldsymbol{c}} ; \tau+1, n\right)\right]-g_{\xi}\left(\theta^{\tau}_{2, \boldsymbol{c}}, \eta^{\tau}_{2, \boldsymbol{c}} ; \tau+1, n\right)\right\}\nonumber\\
&+\left\{\mathbb{E}\left[  g_{\xi}\left(\theta_{2, \boldsymbol{c}}, \eta_{2, \boldsymbol{c}}; \tau_0+1, n\right)\right]  - g_{\xi}\left(\theta_{2, \boldsymbol{c}}, \eta_{2, \boldsymbol{c}} ; \tau_0+1, n\right)\right\} \nonumber\\
\leq & c_7 p^{K-1} \sqrt{\frac{\left(\tau_0-\tau\right) \log (n p)}{s_{\min }^2}}.
\end{align}

Combining (7), (11), (11), we conclude that there exists a constant $C_0>0$ such that
\begin{align*}
& \mathbb{E} \sup _{\tau \in\left[\tau_{n, p}, \tau_0\right]}\left|\mathbb{M}_n(\tau)-\mathbb{M}(\tau)-\mathbb{M}_n\left(\tau_0\right)+\mathbb{M}\left(\tau_0\right)\right| \nonumber\\
\leq & C_0 n p^{K-1}\left\{\sqrt{\frac{\left(\tau_0-\tau_{n, p}\right) \log (n p)}{n^2 s_{\min }^2}}+o\left(\sqrt{\frac{\left(\tau_0-\tau_{n, p}\right) \log (n p)}{n^2 s_{\min }^2}}\right)\right\} \nonumber\\
\leq & 2 C_0 p^{K-1} \sqrt{\frac{\left(\tau_0-\tau_{n, p}\right) \log (n p)}{s_{\min }^2}}.
\end{align*}

\subsubsection*{A.10.1.4 Evaluating $\sup_{\tau \in\left[n_0, \tau_{n,p}\right]} \mathbb{M}_n(\tau)-\mathbb{M}_n\left(\tau_0\right)$ }
Rewrite for any $\tau \in\left[n_0, \tau_{n, p}\right]$,
$$
\mathbb{M}_n(\tau)-\mathbb{M}_n\left(\tau_0\right)=\left[\mathbb{M}\left(\tau_0\right)-\mathbb{M}_n\left(\tau_0\right)\right]+\left[\mathbb{M}_n(\tau)-\mathbb{M}(\tau)\right]-\left[\mathbb{M}\left(\tau_0\right)-\mathbb{M}(\tau)\right]. 
$$
For the first term, consider $\tau \in\left[\tau_{n, p}, \tau_0\right]$. When $\widehat{\psi}=\psi$, we have,
$$
\begin{aligned}
\mathbb{M}_n\left(\tau\right)-\mathbb{M}\left(\tau\right)= & \sum_{\boldsymbol{c}   \in \mathbf{\Psi}} \sum_{\xi \in S_{\boldsymbol{c}}} g_{\xi}\left(\widehat{\theta}_{1, \boldsymbol{c}}^\tau, \widehat{\eta}_{1, \boldsymbol{c}}^\tau ; 1,\tau\right)+\sum_{\boldsymbol{c}   \in \mathbf{\Psi}} \sum_{\xi \in S_{\boldsymbol{c}}} g_{ \xi}\left(\widehat{\theta}_{2, \boldsymbol{c}}^\tau, \widehat{\eta}_{2, \boldsymbol{c}}^\tau ; \tau+1,n\right) \\
& -\mathbb{E} \sum_{\boldsymbol{c}   \in \mathbf{\Psi}} \sum_{\xi \in S_{\boldsymbol{c}}} g_{ \xi}\left(\theta_{1, \boldsymbol{c}}, \eta_{1, \boldsymbol{c}} ; 1,\tau\right)-\mathbb{E} \sum_{\boldsymbol{c}   \in \mathbf{\Psi}} \sum_{\xi \in S_{\boldsymbol{c}}} g_{ \xi}\left(\theta^{\tau}_{2, \boldsymbol{c}}, \eta^{\tau}_{2, \boldsymbol{c}} ; \tau+1,n\right) \\
= & \sum_{\boldsymbol{c}   \in \mathbf{\Psi}} \sum_{\xi \in S_{\boldsymbol{c}}} g_{ \xi}\left(\widehat{\theta}_{1, \boldsymbol{c}}^\tau, \widehat{\eta}_{1, \boldsymbol{c}}^\tau ; 1,\tau\right) -\sum_{\boldsymbol{c}   \in \mathbf{\Psi}} \sum_{\xi \in S_{\boldsymbol{c}}} g_{ \xi}\left(\theta_{1, \boldsymbol{c}}, \eta_{1, \boldsymbol{c}} ; 1,\tau\right)\\
& +\sum_{\boldsymbol{c}   \in \mathbf{\Psi}} \sum_{\xi \in S_{\boldsymbol{c}}} g_{ \xi}\left(\widehat{\theta}_{2, \boldsymbol{c}}^\tau, \widehat{\eta}_{2, \boldsymbol{c}}^\tau ; \tau+1,n\right)-\sum_{\boldsymbol{c}   \in \mathbf{\Psi}} \sum_{\xi \in S_{\boldsymbol{c}}} g_{ \xi}\left(\theta^{\tau}_{2, \boldsymbol{c}}, \eta^{\tau}_{2, \boldsymbol{c}} ; \tau+1,n\right) \\
& +\sum_{\boldsymbol{c}   \in \mathbf{\Psi}} \sum_{\xi \in S_{\boldsymbol{c}}} g_{ \xi}\left(\theta_{1, \boldsymbol{c}}, \eta_{1, \boldsymbol{c}} ; 1,\tau\right) -\mathbb{E} \sum_{\boldsymbol{c}   \in \mathbf{\Psi}} \sum_{\xi \in S_{\boldsymbol{c}}} g_{\xi}\left(\theta_{1, \boldsymbol{c}}, \eta_{1, \boldsymbol{c}} ; 1,\tau\right)\\
& +\sum_{\boldsymbol{c}   \in \mathbf{\Psi}}\sum_{\xi \in S_{\boldsymbol{c}}} g_{ \xi}\left(\theta_{2, j, \ell}^\tau, \eta^{\tau}_{2, \boldsymbol{c}} ; \tau+1,n\right)-\mathbb{E} \sum_{\boldsymbol{c}   \in \mathbf{\Psi}} \sum_{\xi \in S_{\boldsymbol{c}}} g_{ \xi}\left(\theta^{\tau}_{2, \boldsymbol{c}}, \eta^{\tau}_{2, \boldsymbol{c}} ; \tau+1,n\right)
\end{aligned}
$$
Similar to (4) and (5), there exists a large enough constant $C_0>0$ such that with probability greater than $1-O\left((n p)^{-B}\right)$,
\begin{align*}
\sup _{\tau \in\left[\tau_{n, p}, \tau_0\right]}\left|\mathbb{M}_n(\tau)-\mathbb{M}(\tau)\right| \leq C_0 n p^{K-1}\left\{\frac{\log (n p)}{n s_{\min }^2}+\sqrt{\frac{\log (n p)}{n p^{K-1}}}\right\}=O\left(n p^{K-1} \sqrt{\frac{\log (n p)}{n s_{\min }^2}}\right)
\end{align*}
and then 
\begin{align*}
    \left|\mathbb{M}_n(\tau_0)-\mathbb{M}(\tau_0)\right| \leq C_0 n p^{K-1}\left\{\frac{\log (n p)}{n s_{\min }^2}+\sqrt{\frac{\log (n p)}{n p^{K-1}}}\right\}.
\end{align*} 

For the second term, note that when $\tau \in\left[n_0, \tau_{n, p}\right]$, $\widehat{\psi}^{\tau+1, n}$ might be inconsistent in estimating $\psi^{\tau_0+1, n}$, so we cannot directly calculate upper bound of $\left|\mathbb{M}_n(\tau)-\mathbb{M}(\tau)\right| $. 

Given $\widehat{\psi}^{1, \tau}$ and $\widehat{\psi}^{\tau+1, n}$, we define an intermediate term

$$
\mathbb{M}_n^*(\tau):=l\left(\left\{\theta_{\tau, \boldsymbol{c}}^{-}, \eta_{\tau, \boldsymbol{c}}^{-}\right\} ; \widehat{\psi}^{1, \tau}\right)+l\left(\left\{\theta_{\tau, \boldsymbol{c}}^*, \eta_{\tau, \boldsymbol{c}}^*\right\} ; \widehat{\psi}^{\tau+1, n}\right)
$$

where

$$
\theta_{\tau, \boldsymbol{c}}^{-}=\frac{\sum_{\xi \in \widehat{S}_{1, \boldsymbol{c}}^\tau} \frac{\theta_{1, \psi(i), \psi(j)} \eta_{1, \psi(i), \psi(j)}}{\theta_{1, \psi(i), \psi(j)}+\eta_{1, \psi(i), \psi(j)}}}{\sum_{\xi \in \widehat{S}_{1, \boldsymbol{c}}^\tau} \frac{\eta_{1, \psi(i), \psi(j)}}{\theta_{1, \psi(i), \psi(j)}+\eta_{1, \psi(i), \psi(j)}}}, \quad \eta_{\tau, \boldsymbol{c}}^{-}=\frac{\sum_{\xi \in \widehat{S}_{1, \boldsymbol{c}}^\tau} \frac{\theta_{1, \psi(i), \psi(j)} \eta_{1, \psi(i), \psi(j)}}{\theta_{1, \psi(i), \psi(j)}+\eta_{1, \psi(i), \psi(j)}}}{\sum_{\xi \in \widehat{S}_{1, \boldsymbol{c}}^\tau} \frac{\theta_{1, \psi(i), \psi(j)}}{\theta_{1, \psi(i), \psi(j)}+\eta_{1, \psi(i), \psi(j)}}},
$$

and

$$
\begin{aligned}
& \theta_{\tau, \boldsymbol{c}}^*=\frac{\sum_{\xi \in \widehat{S}_{2, \boldsymbol{c}}^\tau}\left[\frac{\left(\tau_0-\tau\right) \theta_{1, \psi(i), \psi(j)} \eta_{1, \psi(i), \psi(j)}}{\theta_{1, \psi(i), \psi(j)}+\eta_{1, \psi(i), \psi(j)}}+\frac{\left(n-\tau_0\right) \theta_{2, \psi(i), \psi(j)} \eta_{2, \psi(i), \psi(j)}}{\theta_{2, \psi(i), \psi(j)}+\eta_{2, \psi(i), \psi(j)}}\right]}{\sum_{\xi \in \widehat{S}_{2, \boldsymbol{c}}^\tau}\left[\frac{\left(\tau_0-\tau\right) \eta_{1, \psi(i), \psi(j)}}{\theta_{1, \psi(i), \psi(j)}+\eta_{1, \psi(i), \psi(j)}}+\frac{\left(n-\tau_0\right) \eta_{2, \psi(i), \psi(j)}}{\theta_{2, \psi(i), \psi(j)}+\eta_{2, \psi(i), \psi(j)}}\right]}, \\
& \eta_{\tau, \boldsymbol{c}}^*=\frac{\sum_{\xi \in \widehat{S}_{2, \boldsymbol{c}}^\tau}\left[\frac{\left(\tau_0-\tau\right) \theta_{1, \psi(i), \psi(j)} \eta_{1, \psi(i), \psi(j)}}{\theta_{1, \psi(i), \psi(j)}+\eta_{1, \psi(i), \psi(j)}}+\frac{\left(n-\tau_0\right) \theta_{2, \psi(i), \psi(j)} \eta_{2, \psi(i), \psi(j)}}{\theta_{2, \psi(i), \psi(j)}+\eta_{2, \psi(i), \psi(j)}}\right]}{\sum_{\xi \in \widehat{S}_{2, \boldsymbol{c}}^\tau}\left[\frac{\left(\tau_0-\tau\right) \theta_{1, \psi(i), \psi(j)}}{\theta_{1, \psi(i), \psi(j)}+\eta_{1, \psi(i), \psi(j)}}+\frac{\left(n-\tau_0\right) \theta_{2, \psi(i), \psi(j)}}{\theta_{2, \psi(i), \psi(j)}+\eta_{2, \psi(i), \psi(j)}}\right]} .
\end{aligned}
$$

Notice that given $\widehat{\psi},\left\{\theta_{\tau, \boldsymbol{c}}^{-}, \eta_{\tau, \boldsymbol{c}}^{-}\right\}$is the maximizer of $\mathbb{E}l\left(\left\{\theta_{\boldsymbol{c}}, \eta_{\boldsymbol{c}}\right\} ; \widehat{\psi}^{1, \tau}\right)$ and $\left\{\theta_{\tau, \boldsymbol{c}}^*, \eta_{\tau, \boldsymbol{c}}^*\right\}$ is the maximizer of $\mathbb{E}l\left(\left\{\theta_{\boldsymbol{c}}, \eta_{\boldsymbol{c}}\right\} ; \widehat{\psi}^{\tau+1, n}\right)$. 

We have

$$
\mathbb{M}_n(\tau)-\mathbb{M}(\tau)=\mathbb{M}_n(\tau)-\mathbb{E} \mathbb{M}_n^*(\tau)+\mathbb{E} \mathbb{M}_n^*(\tau)-\mathbb{M}(\tau)
$$

Note that the expected $\log$-likelihood $\mathbb{E} \sum_{\xi \in \mathcal{E}} g_{1, \xi}\left(\alpha_{1, \xi}, \beta_{1, \xi}, \tau\right)$ is maximized at $\alpha_{1, \xi}=$ $\theta_{1, \psi(i), \psi(j)}, \beta_{1, \xi}=\eta_{1, \psi(i), \psi(j)}$, and $\mathbb{E} \sum_{\xi \in \mathcal{E}} g_{2, \xi}\left(\alpha_{2, \xi}, \beta_{2, \xi}, \tau\right)$ is maximized at $\alpha_{2, \xi}=$ $\theta_{2, \psi(i), \psi(j)}, \beta_{2, \xi}=\eta_{2, \psi(i), \psi(j)}$, we have

$$
\mathbb{E} \mathbb{M}_n^*(\tau)-\mathbb{M}(\tau) \leq 0
$$

On the other hand, similar to (13), there exists a large enough constant $C_7>0$ such that with probability greater than $1-O\left((n p)^{-B}\right)$,

$$
\sup _{\tau \in\left[n_0, \tau_{n, p}\right]}\left|\mathbb{M}_n(\tau)-\mathbb{E} \mathbb{M}_n^*(\tau)\right| \leq C_7 n p^{K-1}\left\{\frac{\log (n p)}{n}+\sqrt{\frac{\log (n p)}{n p^{K-1}}}\right\}
$$

Consequently we have, with probability greater than $1-O\left((n p)^{-B}\right)$,

$$
\sup _{\tau \in\left[n_0, \tau_{n, p}\right]}\left[\mathbb{M}_n(\tau)-\mathbb{M}(\tau)\right] \leq C_7 n p^{K-1}\left\{\frac{\log (n p)}{n}+\sqrt{\frac{\log (n p)}{n p^{K-1}}}\right\}
$$

We remark that since the membership structure $\widehat{\psi}^{\tau+1, n}$ can be very different from the original $\psi$, the $s_{\min }$ in upper bound of (14) is simply replaced by the lower bound 1, and hence the upper bound above is independent of $\widehat{\psi}^{1 ; \tau}$ and $\widehat{\psi}^{\tau+1, n}$.

Hence, combining (9), (14), (15), by choosing $\kappa>0$ to be large enough, we have with probability greater than $1-O\left((n p)^{-B}\right)$,

$$
\begin{aligned}
& \sup _{\tau \in\left[n_0, \tau_{n, p}\right]}\left[\mathbb{M}_n(\tau)-\mathbb{M}_n\left(\tau_0\right)\right] \\
\leq & C_7 n p^{K-1}\left\{\frac{\log (n p)}{n}+\sqrt{\frac{\log (n p)}{n p^{K-1}}}\right\}+C_0 n p^{K-1}\left\{\frac{\log (n p)}{n s_{\min }^2}+\sqrt{\frac{\log (n p)}{n p^{K-1}}}\right\} \\
& -C_3\left(\tau_0-\tau_{n, p}\right)\left[\left\|\mathbf{A}_{1,1}-\mathbf{A}_{2,1}\right\|_F^2+\left\|\mathbf{A}_{1,2}-\mathbf{A}_{2,2}\right\|_F^2\right] 
< 0.
\end{aligned}
$$
\subsubsection*{A.10.1.5 Error bound for $\tau_0-\widehat{\tau}$.}
Combining results from (A.10.2) and (A.10.3), we have  
\begin{align*}
 \mathbb{P}\left(\sup _{\tau \in\left[\tau_{n, p}, \tau_0-\epsilon\right]} \mathbb{M}_n(\tau)-\mathbb{M}_n\left(\tau_0\right) \geq 0\right) \leq  \frac{2 C_0 p^{K-1} \sqrt{\frac{\left(\tau_0-\tau_{n, p}\right) \log (n p)}{s_{\min }^2}}}{C_3 \epsilon p^{K-1} \Delta_F^2}
\end{align*}

Together with results from (A.10.4), We thus conclude that $\tau_0-\widehat{\tau}=O_p\left(\Delta_F^{-2} \sqrt{\frac{\left(\tau_0-\tau_{n, p}\right) \log (n p)}{s_{\min }^2}}\right)$. By the definition of $\tau_{n, p}$ and condition C5 we have,

$$
\Delta_F^{-2} \sqrt{\frac{\left(\tau_0-\tau_{n, p}\right) \log (n p)}{s_{\min }^2}}=O\left(\frac{\tau_0-\tau_{n, p}}{\Delta_F} \sqrt{\frac{\log (n p)}{n s_{\min }^2}}\left[\frac{\log (n p)}{n}+\sqrt{\frac{\log (n p)}{n p^{K-1}}}\right]^{-1 / 2}\right)
$$

Consequently, we conclude that

$$
\tau_0-\widehat{\tau}=O_p\left(\left(\tau_0-\tau_{n, p}\right) \min \left\{1, \frac{\min \left\{1,\left(n^{-1} p^{K-1} \log (n p)\right)^{\frac{1}{4}}\right\}}{\Delta_F s_{\min }}\right\}\right)
$$
\subsubsection*{A.10.2 Change point estimation with $\psi^{1, \tau_0}\neq\psi^{\tau_0+1, n}$}

We modify steps (A.10.1.1)-(A.10.1.4) to the case where $\psi^{1, \tau_0} \neq \psi^{\tau_0+1, n}$.
With some abuse of notations, we put $\alpha_{1, \xi}=\theta_{1, \psi^{1, \tau_0}(\xi)}$, $\beta_{1, \xi}=\eta_{1, \psi^{1, \tau_0}(\xi)}$,  $\alpha_{2, \xi}=$ $\theta_{2, \psi^{\tau_0+1, n}(\xi)}$, and $\beta_{2, \xi}=\eta_{2, \psi^{\tau_0+1, n}(\xi)}$. Similar to previous proofs we define

$$
\begin{aligned}
& \mathbb{M}_n(\tau):=\sum_{\xi \in \mathcal{E}} g_{\xi}\left(\widehat{\alpha}_{1, \xi}^\tau, \widehat{\beta}_{1, \xi}^\tau; 1,\tau\right)+\sum_{\xi \in \mathcal{E}} g_{2, \xi}\left(\widehat{\alpha}_{2, \xi}^\tau, \widehat{\beta}_{2, \xi}^\tau;\tau+1,n\right), \\
& \mathbb{M}(\tau):=E \sum_{\xi \in \mathcal{E}} g_{1, \xi}\left(\alpha_{1, \xi}, \beta_{1, \xi}; 1,\tau\right)+E \sum_{\xi \in \mathcal{E}} g_{2, \xi}\left(\alpha_{2, \xi}^\tau, \beta_{2, \xi}^\tau; \tau+1,n\right),
\end{aligned}
$$

where

$$
\begin{aligned}
\alpha_{2, \xi}^\tau & =\frac{\frac{\tau_0-\tau}{n-\tau} \frac{\alpha_{1, \xi} \beta_{1, \xi}}{\alpha_{1, \xi}+\beta_{1, \xi}}+\frac{n-\tau_0}{n-\tau} \frac{\alpha_{2, \xi} \beta_{2, \xi}}{\alpha_{2, \xi}+\beta_{2, \xi}}}{\frac{\tau_0-\tau}{n-\tau} \frac{\beta_{1, \xi}}{\alpha_{1, \xi}+\beta_{1, \xi}}+\frac{n-\tau_0}{n-\tau} \frac{\beta_{2, \xi}}{\alpha_{2, \xi}+\beta_{2, \xi}}} \\
\beta_{2, \xi}^\tau & =\frac{\frac{\tau_0-\tau}{n-\tau} \frac{\alpha_{1, \xi} \beta_{1, \xi}}{\alpha_{1, \xi}+\beta_{1, \xi}}+\frac{n-\tau_0}{n-\tau} \frac{\alpha_{2, \xi} \beta_{2, \xi}}{\alpha_{2, \xi}+\beta_{2, \xi}}}{\frac{\tau_0-\tau}{n-\tau} \frac{\alpha_{1, \xi}}{\alpha_{1, \xi}+\beta_{1, \xi}}+\frac{n-\tau_0}{n-\tau} \frac{\alpha_{2, \xi}}{\alpha_{2, \xi}+\beta_{2, \xi}}}
\end{aligned}
$$
and

$$
\begin{aligned}
& \widehat{\alpha}_{1, \xi}^\tau=\widehat{\theta}_{1, \widehat{\psi}^{1, \tau}(\xi)}, \quad \widehat{\beta}_{1, \xi}^\tau=\widehat{\eta}_{1, \widehat{\psi}^{1, \tau}(\xi)}^\tau \\
& \widehat{\alpha}_{2, \xi}^\tau=\widehat{\theta}_{2, \widehat{\psi}^{\tau+1, n}(\xi)}^\tau, \quad \widehat{\beta}_{2, \xi}^\tau=\widehat{\eta}_{\tau, \widehat{\psi}^{\tau+1, n}(\xi)}^2
\end{aligned}
$$

Note that the definition of $M(\tau)$ here is now slightly different from the previous definition in that the $\alpha_{2, \xi}^\tau$ and $\beta_{2, \xi}^\tau$ will generally be different from $\theta_{2, \psi^{\tau_0+1, n}(\xi)}^\tau$ and $\eta_{2, \psi^{\tau_0+1, n}(\xi)}^\tau$, unless $\psi^{1, \tau_0}=\psi^{\tau_0+1, n}$. We first of all point out the main difference we are facing in the case where $\psi^{1, \tau_0} \neq \psi^{\tau_0+1, n}$. Consider a detection time $\tau \in\left[\tau_{n, p}, \tau_0\right]$. In the case where $\widehat{\psi}^{1, \tau}=\widehat{\psi}^{\tau+1, n}=\psi$, we have $\alpha_{2, \xi}^\tau=\theta_{2, \boldsymbol{c}}^\tau$ for all $\xi \in S_{\boldsymbol{c}}$, and we have $\left|\widehat{\theta}_{2,\boldsymbol{c}}^\tau-\theta_{2, \boldsymbol{c}}^\tau\right|=O_p\left(\sqrt{\frac{\log (n p)}{n s_{\min }^2}}\right)$ for all $\boldsymbol{c} \in \mathbf{\Psi}$, or equivalently, $\left|\widehat{\alpha}_{2, \xi}^\tau-\theta_{2, \boldsymbol{c}}^\tau\right|=$ $O_p\left(\sqrt{\frac{\log (n p)}{n s_{\min }^2}}\right)$ for all $\xi \in \mathcal{E}$. However, when $\widehat{\psi}^{1, \tau}=\psi^{1, \tau_0}, \widehat{\psi}^{\tau+1, n}=\psi^{\tau_0+1, n}$ but $\psi^{1, \tau_0} \neq \psi^{\tau_0+1, n}$, the order of the estimation error becomes $O_p\left(\sqrt{\frac{\log (n p)}{n s_{\min }^2}}+\frac{\tau_0-\tau}{n}\right)$. Here $\frac{\tau_0-\tau}{n}$ is a bias terms brought by the fact that $\widehat{\psi}^{1, \tau} \neq \widehat{\psi}^{\tau+1, n}$. The main issue is that the the following terms from the definition of $\widehat{\theta}_{2, \boldsymbol{c}}^\tau$ :

$$
\sum_{\xi \in \widehat{S}_{2, \boldsymbol{c}}^\tau} \sum_{t=\tau+1}^{\tau_0} X_{\xi}^t\left(1-X_{\xi}^{t-1}\right), \quad \sum_{\xi \in \widehat{S}_{2, \boldsymbol{c}}^\tau} \sum_{t=\tau+1}^{\tau_0}\left(1-X_{\xi}^{t-1}\right),
$$

are no longer unbiased estimators (subject to a normalization) of the following corresponding terms in the definition of $\theta_{2, \boldsymbol{c}}^\tau$ :

$$
\frac{\theta_{1, \boldsymbol{c}} \eta_{1, \boldsymbol{c}}}{\theta_{1, \boldsymbol{c}}+\eta_{1, \boldsymbol{c}}}, \quad \frac{\theta_{1, \boldsymbol{c}}}{\theta_{1, \boldsymbol{c}}+\eta_{1, \boldsymbol{c}}} .
$$

For (A.10.1.1), note that $\left|\widehat{\alpha}_{2, \xi}^\tau-\alpha_{2, \xi}\right| \leq\left|\widehat{\alpha}_{2, \xi}^\tau-\alpha_{2, \xi}^\tau\right|+O\left(\frac{\tau_0-\tau}{n}\right)$ holds for all $\xi \in \mathcal{E}$, where the $O\left(\frac{\tau_0-\tau}{n}\right)$ is independent of $\xi$. This implies that when estimating the $\alpha_{2, \xi}$, we have introduced a bias term $O\left(\frac{\tau_0-\tau}{n}\right)$ by including the $\tau_0-\tau$ samples before the change point. From Theorem 1 and Theorem 2, and condition C4, we conclude that (A.10.1.1) hold for $\widehat{\psi}^{\tau+1, n}$.

The proof of (A.10.1.2) does not involve any parameter estimators and hence can be established similarly.

For (A.10.1.3), the error bounds related to $g_{\xi}(\cdot, \cdot ; 1, \cdot)$ remain unchanged. Note that the decomposition (12) still holds with $\theta_{2, \boldsymbol{c}}^\tau, \eta_{2, \boldsymbol{c}}^\tau$ replaced be $\alpha_{2, \xi}^\tau, \beta_{2, \xi}^\tau$ and $\widehat{\theta}_{2, \boldsymbol{c}}^\tau, \widehat{\eta}_{2, \boldsymbol{c}}^\tau$ replaced be $\widehat{\alpha}_{2, \xi}^\tau, \widehat{\beta}_{2, \xi}^\tau$. The bound for (11) would become $C_0 p^2\left\{\sqrt{\frac{\left(\tau_0-\tau_{n, p}\right) \log (n p)}{s_{\min }^2}}+\left(\tau_0-\tau_{n, p}\right)\right\} .$

For (A.10.1.4), replacing the order of the error bound for $\widehat{\theta}_{2, \boldsymbol{c}}^{\tau}$and $\widehat{\eta}_{2, \boldsymbol{c}}^{\tau}$from $\sqrt{\frac{\log (n p)}{n s_{\min }^2}}$ to $\sqrt{\frac{\log (n p)}{n s_{\min }^2}}+\frac{\tau_0-\tau}{n}$, we have there exists a large enough constant $C_0>0$ such that

$$
\begin{aligned}
\sup _{\tau \in\left[\tau_{n, p}, \tau_0\right]}\left|\mathbb{M}_n(\tau)-\mathbb{M}(\tau)\right| & \leq C_0 n p^2\left\{\frac{\log (n p)}{n s_{\min }^2}+\sqrt{\frac{\log (n p)}{n p^2}}+\frac{\left(\tau_0-\tau_{n, p}\right)^2}{n^2}\right\} \\
& =O\left(n p^2\left\{\sqrt{\frac{\log (n p)}{n s_{\min }^2}}+\frac{\left(\tau_0-\tau_{n, p}\right)^2}{n^2}\right\}\right)
\end{aligned}
$$
By noticing that $\left\{\alpha_{1, \xi}, \beta_{1, \xi}, \alpha_{2, \xi}^\tau, \beta_{2, \xi}^\tau\right\}$ is the maximizer of $\mathbb{M}(\tau)$, we conclude that (A.10.1.4) also holds. 

Consequently, for (A.10.1.5), we have

$$
P\left(\sup _{\tau \in\left[\tau_{n, p}, \tau_0-\epsilon\right]} \mathbb{M}_n(\tau)-\mathbb{M}_n\left(\tau_0\right) \geq 0\right) \leq \frac{C_0 p^{K-1} \sqrt{\frac{\left(\tau_0-\tau_{n, p}\right) \log (n p)}{s_{\min }}}+C_0 p^{K-1}\left(\tau_0-\tau_{n, p}\right)}{C_3 \epsilon p^{K-1} \Delta_F^2} .
$$

Consequently, we conclude that

$$
\tau_0-\widehat{\tau}=O_p\left(\left(\tau_0-\tau_{n, p}\right) \min \left\{1, \frac{\min \left\{1,\left(n^{-1} p^{K-1} \log (n p)\right)^{\frac{1}{4}}\right\}}{\Delta_F s_{\min }}+\frac{1}{\Delta_F^2}\right\}\right) .
$$

\end{document}